\documentclass[a4paper,USenglish]{lipics-v2018}

\nolinenumbers

\usepackage{microtype}
\hideLIPIcs

\usepackage{xspace}
\usepackage{mathtools}
\newcommand{\F}{\ensuremath{\mathcal{F}}\xspace}
\newcommand{\Q}{\ensuremath{\mathcal{Q}}\xspace}
\newcommand{\Oh}{\mathcal{O}}
\newcommand{\notcontainment}{\ensuremath{\mathsf{NP \not\subseteq coNP/poly}}\xspace}
\newcommand{\containment}{\ensuremath{\mathsf{NP  \subseteq coNP/poly}}\xspace}
\newcommand{\yes}{\textsc{yes}\xspace}

\newcommand{\minFdel}{\ensuremath{\textsc{opt}_\F}\xspace}
\newcommand{\minFdelsol}{\ensuremath{\textsc{optsol}_\F}\xspace}
\newcommand{\minFdelsolwith}{\ensuremath{\minFdelsol\textsc{st}}\xspace}
\newcommand{\FDeletion}{\textsc{$\mathcal{F}$-Deletion}\xspace}
\newcommand{\tw}{\ensuremath{\mathrm{\textsc{tw}}}\xspace}
\newcommand{\td}{\ensuremath{\mathrm{\textsc{td}}}\xspace}

\newcommand{\pw}{\ensuremath{\mathrm{\textsc{pw}}}\xspace}
\newcommand{\bound}{\ensuremath{b}\xspace}
\newcommand{\lab}{\ensuremath{L}\xspace}
\newcommand{\graphs}{\ensuremath{\mathcal{G}}\xspace}
\newcommand{\lgraphs}{\ensuremath{\mathcal{G}^X}\xspace}
\newcommand{\bgraphs}{\ensuremath{\mathcal{G}_t}\xspace}
\newcommand{\lbgraphs}{\ensuremath{\mathcal{G}^X_t}\xspace}
\newcommand{\congraphs}{\ensuremath{\mathcal{G}_{\textsc{con}}}\xspace}
\newcommand{\lcongraphs}{\ensuremath{\mathcal{G}^X_{\textsc{con}}}\xspace}
\newcommand{\attlbgraphs}{\ensuremath{\mathcal{G}^X_{t,\textsc{att}}}\xspace}
\newcommand{\attbgraphs}{\ensuremath{\mathcal{G}_{t,\textsc{att}}}\xspace}
\newcommand{\nonisobgraphs}[1]{\ensuremath{\mathcal{G}_{#1,\textsc{def}(#1)}}\xspace}

\newcommand{\numberof}{\ensuremath{N}\xspace}

\newcommand{\components}{\ensuremath{\textsc{\#cc}}\xspace}
\newcommand{\iscon}{\ensuremath{\textsc{isCon}}\xspace}
\newcommand{\setcomponents}{\ensuremath{\textsc{cc}}\xspace}

\newcommand{\extend}{\ensuremath{\textsc{ext}}\xspace}
\newcommand{\forget}{\ensuremath{\textsc{forget}}\xspace}
\newcommand{\multipieces}{\ensuremath{\textsc{mpcs}}\xspace}
\newcommand{\pieces}{\ensuremath{\textsc{pcs}}\xspace}

\newcommand{\splitPi}{\ensuremath{\textsc{split}}\xspace}

\newcommand{\leqm}{\ensuremath{\preceq_m}\xspace}
\newcommand{\leqb}{\leqm\xspace}
\newcommand{\leql}{\leqm\xspace}
\newcommand{\leqlb}{\leqm\xspace}

\newcommand{\folio}{\ensuremath{\textsc{folio}}\xspace}
\newcommand{\foliotstar}{\ensuremath{\textsc{folio}^*_{\Q,t}}\xspace}
\newcommand{\foliotplusonestar}{\ensuremath{\textsc{folio}^*_{\Q,t+1}}\xspace}

\newcommand{\Fdel}{solution\xspace}

\newcommand{\piaone}{\ensuremath{\Pi_{A_1}}\xspace}
\newcommand{\piatwo}{\ensuremath{\Pi_{A_2}}\xspace}
\newcommand{\pib}{\ensuremath{\Pi_{B}}\xspace}
\newcommand{\Rhat}[1]{\ensuremath{\hat{\mathcal{R}}^{\piaone,\piatwo,#1}}\xspace}
\newcommand{\Rtilde}{\ensuremath{\widetilde{\mathcal{R}}^{\piatwo,\piaone}}\xspace}
\newcommand{\Qhat}[1]{\ensuremath{\hat{\mathcal{Q}}^{\piaone,\piatwo,#1}}\xspace}
\newcommand{\Qtilde}{\ensuremath{\widetilde{\mathcal{Q}}^{\piatwo,\piaone}}\xspace}

\newcommand{\shrink}{\hspace{-0.4cm}}

\theoremstyle{plain}
\newtheorem{innerclaim}[theorem]{Claim}

\newtheorem{proposition}[theorem]{Proposition}
\newtheorem{observation}[theorem]{Observation}
\newtheorem{procedure}[theorem]{Procedure}
\let\plainqed\qedsymbol
\newcommand{\claimqed}{$\lrcorner$}
\newenvironment{innerproof}[1][\proofname]{\begin{proof}[#1]\renewcommand{\qedsymbol}{\claimqed}}{\end{proof}\renewcommand{\qedsymbol}{\plainqed}}

\newcommand{\defparproblem}[4]{\par
 \vspace{3mm}
\noindent\fbox{
 \begin{minipage}{0.96\textwidth}
 \begin{tabular*}{\textwidth}{@{\extracolsep{\fill}}lr} #1 & {\bf{Parameter:}} #3 \vspace{1mm} \\ \end{tabular*}
 {\textbf{Input:}} #2
	\vspace{1mm}\\%
 {\textbf{Question:}} #4
 \end{minipage}
 }
 \vspace{3mm}
\par
}

\title{Polynomial Kernels for Hitting Forbidden Minors  under Structural Parameterizations}

\author{Bart M.\,P. Jansen}{Eindhoven University of Technology\\{P.O. Box 513, 5600 MB Eindhoven, The Netherlands}}{b.m.p.jansen@tue.nl}{http://orcid.org/0000-0001-8204-1268}{}

\author{Astrid Pieterse}{Eindhoven University of Technology\\{P.O. Box 513, 5600 MB Eindhoven, The Netherlands}}{a.pieterse@tue.nl}{http://orcid.org/0000-0003-3721-6721}{}

\authorrunning{B.\,M.\,P. Jansen and A. Pieterse}

\Copyright{Bart M.\,P. Jansen and Astrid Pieterse}

\subjclass{%
\ccsdesc[500]{Theory of computation~Graph algorithms analysis},
\ccsdesc[500]{Theory of computation~Parameterized complexity and exact algorithms}}

\keywords{Kernelization,
\F-minor free deletion,
Treedepth modulator,
Structural parameterization}

\category{}

\relatedversion{}

\supplement{}

\funding{This work was supported by NWO Veni grant ``Frontiers in Parameterized Preprocessing'' and NWO Gravitation grant ``Networks''. }

\EventEditors{John Q. Open and Joan R. Access}
\EventNoEds{2}
\EventLongTitle{42nd Conference on Very Important Topics (ESA 2018)}
\EventShortTitle{ESA 2018}
\EventAcronym{CVIT}
\EventYear{2016}
\EventDate{December 24--27, 2016}
\EventLocation{Little Whinging, United Kingdom}
\EventLogo{}
\SeriesVolume{42}
\ArticleNo{23}

\begin{document}

\maketitle

\begin{abstract}
We investigate polynomial-time preprocessing for the problem of hitting forbidden minors in a graph, using the framework of kernelization. For a fixed finite set of connected graphs~$\F$, the \textsc{$\F$-Deletion} problem is the following: given a graph~$G$ and integer~$k$, is it possible to delete~$k$ vertices from~$G$ to ensure the resulting graph does not contain any graph from~$\F$ as a minor? Earlier work by Fomin, Lokshtanov, Misra, and Saurabh~[FOCS'12] showed that when~$\F$ contains a planar graph, an instance~$(G,k)$ can be reduced in polynomial time to an equivalent one of size~$k^{\Oh(1)}$. In this work we focus on structural measures of the complexity of an instance, with the aim of giving nontrivial preprocessing guarantees for instances whose solutions are large. Motivated by several impossibility results, we parameterize the \textsc{$\F$-Deletion} problem by the size of a vertex modulator whose removal results in a graph of constant treedepth~$\eta$. 

We prove that for each set~$\F$ of connected graphs and constant~$\eta$, the \textsc{$\F$-Deletion} problem parameterized by the size of a treedepth-$\eta$ modulator has a polynomial kernel. Our kernelization is fully explicit and does not depend on protrusion reduction or well-quasi-ordering, which are sources of algorithmic non-constructivity in earlier works on \textsc{$\F$-Deletion}. Our main technical contribution is to analyze how models of a forbidden minor in a graph~$G$ with modulator~$X$, interact with the various connected components of~$G-X$. Using the language of labeled minors, we analyze the fragments of potential forbidden minor models that can remain after removing an optimal \textsc{$\F$-Deletion} solution from a single connected component of~$G-X$. By bounding the number of different types of behavior that can occur by a polynomial in~$|X|$, we obtain a polynomial kernel using a recursive preprocessing strategy. Our results extend earlier work for specific instances of \textsc{$\F$-Deletion} such as \textsc{Vertex Cover} and \textsc{Feedback Vertex Set}. It also generalizes earlier preprocessing results for \textsc{$\F$-Deletion} parameterized by a vertex cover, which is a treedepth-one modulator. 
\end{abstract}

\section{Introduction}
How, and under which circumstances, can a polynomial-time algorithm prune the easy parts of an NP-hard problem input, without changing its answer? This question can rigorously be answered using the notion of kernelization~\cite{Bodlaender09,GuoN07a,Kratsch14} which originated in parameterized complexity theory~\cite{CyganFKLMPPS15,DowneyF13} where it can be naturally framed. After choosing a \emph{complexity parameter} for the NP-hard problem of interest, which associates to every input~$x \in \Sigma^*$ an integer~$k \in \mathbb{N}$ that expresses its difficulty under the chosen type of measurement, the theory postulates that a good preprocessing algorithm can be captured by the notion of a \emph{polynomial kernelization}: a polynomial-time algorithm that, given a parameterized instance~$(x,k) \in \Sigma^* \times \mathbb{N}$, outputs an instance~$(x',k')$ with the same answer whose size is bounded polynomially in~$k$. Not all parameterized problems admit polynomial kernelizations, and one can find meaningful ways to preprocess an NP-hard problem by studying those parameterizations for which it does. The study of kernelization has blossomed over the last decade, resulting in a myriad of interesting techniques for obtaining polynomial kernelizations~\cite{BodlaenderFLPST16,Fernau16,Gutin16,KratschW12,Misra16a}, as well as frameworks for proving the non-existence of polynomial kernelizations under complexity-theoretic assumptions~\cite{Bodlaender09,BodlaenderDFH09,DellM14,Drucker12,FortnowS11}.

Originally, the study of kernelization focused on the \emph{natural parameterizations} of (the decision variants of) search problems, where the complexity parameter~$k$ measures the size of the solution. A classic example~\cite{ChenKJ01,NemhauserT75} is that an instance~$(G,k)$ of the \textsc{$k$-Vertex Cover} problem, which asks whether an undirected graph~$G$ has a vertex cover of size~$k$, can efficiently be reduced to an equivalent instance with at most~$2k$ vertices. This guarantees that efficient pruning can be done on large inputs that have small vertex covers. However, such guarantees are meaningless when the smallest vertex cover contains more than half the vertices. By choosing a parameter that measures the structure of the input graph, rather than the size of the desired solution, one can hope to develop provably good preprocessing procedures even for inputs whose solutions are large. An early example of this approach was given by Jansen and Bodlaender~\cite{JansenB13}, who showed that an instance of the \textsc{Vertex Cover} problem can efficiently be reduced to size~$\Oh(\ell^3)$, where~$\ell$ is the size of a smallest \emph{feedback vertex set} in~$G$: \textsc{Vertex Cover} parameterized by the size of a feedback vertex set has a cubic-vertex kernel. The result effectively conveys that large instances of \textsc{Vertex Cover} that are~$\ell$ vertex-deletions away from being acyclic, can be shrunk to size~$\Oh(\ell^3)$ in polynomial time.

\subparagraph*{Problem statement}
To understand the power of polynomial-time preprocessing algorithms over inputs to NP-hard problems that exhibit some structural regularities, but whose solutions are generally large, we set out to answer the following question:

\begin{quote}
For which structural parameterizations of NP-hard graph problems is it possible to obtain polynomial kernelizations?
\end{quote}

Our goal is to answer this question for a rich class of problems, in terms of a rich class of structural parameterizations. Existing lower bounds show that, in general graphs, it is unlikely that a logical characterization exists of the problems admitting polynomial kernelizations for structural parameterizations (cf.~\cite[\S 1]{FominJP14}), even though meta-theorems in terms of logical definability or finite integer index are possible when dealing with inputs from sparse graph families~\cite{BodlaenderFLPST16,GajarskyHOORRVS17}. We therefore target the class of \textsc{$\mathcal{F}$-Minor-Free Deletion} problems, henceforth abbreviated as \FDeletion problems, to capture a wide class of NP-hard graph problems. Such a problem is instantiated by specifying a finite set~$\mathcal{F}$ of forbidden minors. An input then consists of a graph~$G$ and integer~$k$, and asks whether it is possible to find a set~$Y \subseteq V(G)$ of size~$k$ such that~$G-Y$ contains no graph from~$\mathcal{F}$ as a minor. This is a rich class of problems: by choosing~$\mathcal{F} = \{K_2\}$ we obtain \textsc{Vertex Cover}, for~$\mathcal{F} = \{K_3\}$ we have \textsc{Feedback Vertex Set}, and for~$\mathcal{F} = \{K_5, K_{3,3}\}$ we obtain the problem of making a graph planar by vertex deletions. The kernelization complexity of the solution-size parameterization of \FDeletion has been the subject of intensive research~\cite{FominLMPS11,FominLMS12,GiannopoulouJLS17,KimLPRRSS16,Thomasse10}. In this work we attempt to find the widest class of structural parameterizations for which \FDeletion admits polynomial kernels, continuing a long line of investigation into structural parameterizations for \textsc{Vertex Cover}~\cite{BougeretS16,FominS16,JansenB13,Kratsch16,KratschW12,MajumdarRS15}, \textsc{Feedback Vertex Set}~\cite{JansenRV14,Majumdar17}, and other \FDeletion problems~\cite{FominJP14,GajarskyHOORRVS17}.

When it comes to measuring graph complexity, a natural choice is to consider a \emph{width measure} such as treewidth. Alas, it has long been known that even \textsc{Vertex Cover}, the simplest \FDeletion problem, does not admit a polynomial kernelization when parameterized by the treewidth of the input graph.\footnote{Bodlaender et al.~\cite[Theorem 1]{BodlaenderDFH09} show a superpolynomial kernelization lower bound for \textsc{Independent Set} parameterized by treewidth. Since the parameter is not related to the solution size, this is equivalent to \textsc{Vertex Cover} parameterized by treewidth. The lower bound holds under the assumption that \notcontainment, which we implicitly assume when stating further lower bounds in this section.} Generally speaking, graph problems do not admit polynomial kernels under parameterizations that attain the maximum, rather than the sum, of the values of the connected components. We therefore use the \emph{vertex-deletion} distance to simple graph classes~$\mathcal{G}$ as the parameter. The aforementioned result by Jansen and Bodlaender~\cite{JansenB13} shows that \textsc{Vertex Cover} has a polynomial kernelization when parameterized by the vertex-deletion distance to an acyclic graph, i.e., to a graph of treewidth one. Unfortunately this formulation leaves little room for generalizations: no polynomial kernelization is possible parameterized by the distance to a graph of treewidth two~\cite[Theorem 11]{CyganLPPS14}, or even pathwidth two.\footnote{The lower bound is stated for distance to treewidth two, but the same proof works for pathwidth two.} We therefore cannot use the deletion distance to constant treewidth (\tw) or pathwidth (\pw) as our graph parameter, and use the deletion distance to constant \emph{treedepth} (\td) instead. The parameter treedepth has recently attracted much interest~\cite{ChenM14,ElberfeldGT16,ReidlRVS14}, sometimes allowing better upper bounds than are possible in terms of treewidth~\cite{GajarskyHOORRVS17,PilipczukW16}. It plays an important role in the study of structural sparsity~\cite{NesetrilO12}. All graphs~$G$ satisfy~$\td(G) \geq \pw(G) \geq \tw(G)$, so graphs of constant treedepth are more restricted than those of constant treewidth. We therefore study the following problem for a fixed set~$\mathcal{F}$ of connected graphs and constant~$\eta \geq 1$.

\defparproblem{\FDeletion parameterized by treedepth-$\eta$ modulator}
{A graph~$G$, integer~$k$, and a modulator~$X \subseteq V(G)$ such that~$\td(G-X) \leq \eta$.}
{$|X|$.}
{Is there a set~$Y \subseteq V(G)$ of size~$k$ such that~$G-Y$ is $\mathcal{F}$-minor-free?}

The restriction that~$\F$ contains only connected graphs is needed to ensure that a solution on a disconnected graph can be formed from solutions on its connected components, which we require in some of our proofs. This restriction was also considered in previous work~\cite{FominLMS12} on kernelization, but can be avoided when targeting single-exponential FPT algorithms~\cite{KimLPRRSS16}.

For technical reasons, we assume that a modulator~$X$ is given in the input. If no modulator is known, one can compute an approximate modulator and use it as~$X$. For example, Gajarsk\'y et al.~\cite[Lemma 4.2]{GajarskyHOORRVS17} showed that a modulator of size at most~$2^\eta$ times the optimum can be found in quadratic time. Our problem setting is related to that of Gajarsk\'y et al.~\cite{GajarskyHOORRVS17}. They studied kernelization for a general class of graph problems that includes \FDeletion, parameterized by a constant-treedepth modulator, but under the additional restriction that the input graph has bounded expansion or is nowhere dense. Under this severe restriction they obtained kernelizations of linear size for a wide range of problems. This prompted Somnath Sikdar during the 2013 Workshop on Kernelization~\cite{Worker2013} to ask which types of problems admit polynomial kernelizations in \emph{general graphs}, when parameterized by a constant-treedepth modulator; we address this question in this work.

\newtheorem*{thm:main:statement}{Theorem \ref{thm:main}}
\newcommand{\mainthm}{For every fixed finite set~$\mathcal{F}$ of connected graphs and every constant~$\eta$, the \FDeletion problem parameterized by a treedepth-$\eta$ modulator has a polynomial kernelization.}

\subparagraph*{Our results}
Our main result proves the existence of polynomial kernelizations for \FDeletion parameterized by a modulator whose removal leaves a graph of constant treedepth.
\begin{theorem} \label{thm:main}
\mainthm
\end{theorem}
This answers a question posed by Bougeret and Sau~\cite{BougeretS16} (cf.~\cite{BougeretS18}). They obtained polynomial kernels for \textsc{Vertex Cover} parameterized by a constant-treedepth modulator, and asked whether their result can be extended to the \textsc{Feedback Vertex Set} problem. As \textsc{Feedback Vertex Set} is an \FDeletion problem for~$\F = \{K_3\}$, Theorem~\ref{thm:main} shows that this is indeed the case. Theorem~\ref{thm:main} greatly generalizes an earlier result of Fomin, Jansen, and Pilipczuk~\cite[Corollary 1]{FominJP14}, who proved that \FDeletion parameterized by a \emph{vertex cover} has a polynomial kernel for every fixed~$\mathcal{F}$; note that a vertex cover is precisely a treedepth-$1$ modulator.

Our kernelization is fully explicit and does not depend on protrusion replacement techniques or well-quasi-ordering, which are sources of algorithmic non-constructivity in other works~\cite{FominLMPS11,FominLMS12} on kernelization for \FDeletion. Moreover, our general theorem allows~$\F$ to be any set of connected graphs, including nonplanar ones. In contrast, the kernelization for the solution-size parameterization by Fomin et al.~\cite{FominLMS12} only applies when~$\mathcal{F}$ contains at least one planar graph. Hence they only capture problems where, after removing a solution, the remaining graph has constant treewidth~\cite{RobertsonS86}. In our case, even though the parameter value is expressed in terms of a modulator to a graph of constant treedepth and therefore constant treewidth, the graphs that result after removing an optimal solution may have unbounded treewidth. This occurs, for example, when using~$\mathcal{F} = \{K_5, K_{3,3}\}$ to capture the \textsc{Vertex Planarization} problem. (Whether the solution-size parameterization of \textsc{Vertex Planarization} has a polynomial kernel is a notorious open problem~\cite{FominLMS12}.)

The degree of the polynomial in the kernel size bound grows very quickly with~$\eta$. We prove that this is unavoidable, even for the simplest case of \textsc{Vertex Cover}.

\newtheorem*{thm:lowerbound:statement}{Theorem \ref{thm:lowerbound}}
\newcommand{\lowerbound}{For every~$\eta \geq 6$, the \textsc{Vertex Cover} problem parameterized by the size of a given treedepth-$\eta$ modulator~$X$ does not admit a kernelization of bitsize~$\Oh(|X|^{2^{\eta-4}-\varepsilon})$ for any~$\varepsilon > 0$, unless \containment.}

\begin{theorem} \label{thm:lowerbound}
\lowerbound
\end{theorem}

\subparagraph*{Techniques}
To obtain a polynomial kernel for an instance~$(G,X,k)$ of \FDeletion, the main challenge is to understand how the connected components~$\mathcal{C}$ of~$G-X$ interact through their connections to the modulator~$X$. Using the language of labeled minors, we analyze how minor models of a forbidden graph in~$\mathcal{F}$ may intersect the various components of~$G-X$. Using these insights, we are able to characterize which components of~$\mathcal{C}$ affect the structure of optimal solutions in an essential way. On a high level, the kernelization strategy is as follows. We use the fact that a single constant-treedepth component can be analyzed efficiently, to identify a subset~$\mathcal{C'}$ of~$\mathcal{C}$ that contains~$|X|^{\Oh(1)}$ essential components under our characterization. We prove that the remaining ones can be safely removed, because their interaction with the rest of the instance can be ignored. Formally speaking, we show that any optimal solution on~$G' := G[X \cup \bigcup _{C \in \mathcal{C'}} C]$ can be lifted to a solution on~$G$ by including~$\Delta = \sum _{C \in \mathcal{C} \setminus \mathcal{C'}} \minFdel(C)$ additional vertices:~$(G,X,k)$ is a \textsc{yes}-instance if and only if~$(G', X, k-\Delta)$ is. This effectively shows that there is an optimal solution~$Y$ on~$G$ in which the non-essential components act in isolation:~$Y$ does not delete more vertices from such a component, than would be deleted by a solution on the graph~$G[C]$.

The overall kernelization follows straight-forwardly from this pruning of non-essential components by a recursive approach, similarly as in earlier work~\cite{BougeretS16,GajarskyHOORRVS17}. 
The main challenge is therefore to understand which components are essential and which are not, and this is where our contribution lies. We present a stand-alone combinatorial lemma that captures our key insight in this direction. To state it, we introduce some terminology.

We work with a nonstandard notion of labeled graphs. For a finite set~$X$, an \emph{$X$-labeled graph} is a graph in which each vertex is assigned a (possibly empty) subset of~$X$ as its labelset; we stress that multiple vertices may carry the same label on their labelset. The minor relation on graphs extends to labeled graphs in a natural way: a labeled graph~$H$ is a minor of a labeled graph~$G$, if~$H$ can be obtained from~$G$ by repeatedly deleting an edge, deleting a vertex, deleting a label from the labelset of a vertex, or contracting an edge. When contracting an edge~$\{u,v\}$ into a single vertex~$w$, the labelset of~$w$ is formed as the union of the labelsets of~$u$ and~$v$.

For a collection~$\mathcal{S}$ of vertex subsets~$Y$ of an $X$-labeled graph~$C$, and a set of $X$-labeled graphs~$\Q$, we say that all~$Y \in \mathcal{S}$ \emph{leave a $\Q$-minor} in~$C$, if for all~$Y \in \mathcal{S}$ the graph~$C-Y$ contains some graph~$H \in \Q$ as a labeled minor. We say that a set~$\Q$ of $X$-labeled graphs is \emph{$\theta$}-saturated for an integer~$\theta$, if for each subset~$X' \subseteq X$ of size~$\theta$, the graph consisting of one vertex with labelset~$X'$ belongs to~$\Q$. Our main lemma states that if all optimal solutions to \FDeletion on~$C$ leave a $\Q$-minor for some suitably saturated~$\Q$, then there is a small subset~$\Q^*$ for which the same holds.

\newtheorem*{lemma:main:statement}{Lemma \ref{lem:main}}
\newcommand{\mainlemma}{Let~$\mathcal{F}$ be a finite set of (unlabeled) connected graphs, let~$X$ be a set of labels, let~$\Q$ be a~$(\min_{H \in \mathcal{F}}|V(H)|)$-saturated set of connected~$X$-labeled graphs of at most~$\max_{H \in \F}|E(H)|+1$ vertices each, and let~$C$ be an $X$-labeled graph. If all optimal solutions to \FDeletion on~$C$ leave a $\Q$-minor, then there is a subset~$\Q^*\subseteq\Q$ whose size depends only on~$(\mathcal{F}, \td(C))$, such that all optimal solutions leave a~$\Q^*$-minor.
}

\begin{lemma}[Main lemma] \label{lem:main}
\mainlemma
\end{lemma}


In several aspects, the statement in the lemma is best-possible. In particular, we will show in Section \ref{sec:main-lemma-overview} that the dependence of the size of~$\Q^*$ on~$\td(G)$ rather than~$\tw(G)$ is essential and that the precondition that~$\Q$ is~$\Oh(1)$-saturated cannot be avoided.

Lemma~\ref{lem:main} is the cornerstone in our understanding of which components of~$G-X$ are essential. In our applications of the lemma, the graph~$C$ consists of a connected component of~$G-X$ whose labels encode the adjacency of those vertices to the modulator~$X$. The set~$\Q$ contains potential fragments of models of forbidden $\mathcal{F}$-minors, again labeled by adjacency to~$X$, which we may be interested in destroying in~$C$ so that connections through~$X$ cannot form $\mathcal{F}$-minors with fragments that remain in other components of~$G-X$. The lemma then essentially says that if it is not possible to select a solution that deletes a minimum number of vertices from~$C$ while simultaneously destroying all fragments in~$\Q$, then there is a bounded-size subset of fragments~$\Q^*$ that cannot all be destroyed by such a solution. The full importance of Lemma~\ref{lem:main} will become clear in Section~\ref{sec:kernel}.

\subparagraph*{Organization}
Section~\ref{sec:basic:prelims} provides basic preliminaries. In Section~\ref{sec:main-lemma-overview}, we give some of the main ideas of the proof of Lemma~\ref{lem:main}. In Section~\ref{sec:kernel} we show how Theorem~\ref{thm:main} follows from a procedure that identifies relevant components. We give the procedure and its correctness proof later in the same section, while relying on Lemma~\ref{lem:main}. The proof of Lemma~\ref{lem:main} is long and technical. In the appendix, we first develop a framework for boundaried labeled graphs and establish some useful auxiliary lemmata (Section~\ref{sec:prelims}) and finally use these to prove the main lemma (Section~\ref{sec:main:lemma}). Theorem~\ref{thm:lowerbound} is proven in Section~\ref{sec:lb} in the appendix. The proofs of statements marked ($\bigstar$) can be found in the appendix, Section~\ref{sec:omitted:proofs}. 

\section{Preliminaries} \label{sec:basic:prelims}
For a positive integer~$n$ we use~$[n]$ as a shorthand for~$\{1,\ldots, n\}$. For a set~$S$, let $2^S$ to denote the set of all subsets of~$S$. All graphs we consider are finite, undirected, and simple. A graph~$G$ consists of a vertex set~$V(G)$ and edge set~$E(G) \subseteq \binom{V(G)}{2}$. The open neighborhood of a vertex~$v$ is denoted~$N_G(v)$. For a vertex set~$S \subseteq V(G)$, its open neighborhood is~$N_G(S) := \bigcup_{v \in S} N_G(v) \setminus S$. For an edge~$\{u,v\}$ in a graph~$G$, \emph{contracting~$\{u,v\}$} results in the graph~$G'$ obtained from~$G$ by removing~$u$ and~$v$, and replacing them by a new vertex~$w$ with~$N_{G'}(w) = N_G(\{u,v\})$. For a vertex set~$S \subseteq V(G)$, we use~$G - S$ to denote the graph obtained from~$G$ by deleting all vertices in~$S$ and their incident edges. The subgraph of~$G$ induced by vertex set~$S$ is denoted~$G[S]$.

\begin{definition}[treedepth] \label{def:treedepth}
Treedepth is defined as follows. The trivial one-vertex graph has treedepth~$1$. The treedepth of a disconnected graph~$G$ with connected components~$C^1, \ldots, C^t$ is~$\max_{i \in [t]} \td(C^i)$. The treedepth of a connected graph~$G$ is~$\min_{v \in V(G)} \td(G - \{v\}) + 1$.
\end{definition}

\begin{definition}[labeled graph]
Let $X$ be a set. An \emph{$X$-labeled graph} $G$ is a graph $G$ together with label function $\lab_G \colon V(G) \rightarrow 2^X$, assigning a (potentially empty) subset of labels to each vertex in $G$. The labeled graph~$G$ is $\theta$-restricted if each vertex has at most $\theta$ labels.
\end{definition}

If an edge~$\{u,v\}$ is contracted in a labeled graph to obtain a new vertex~$w$, then the labelset of~$w$ is defined as~$\lab_G(u) \cup \lab_G(v)$.

\begin{definition}[minor model] \label{def:minor:model:plain}
A \emph{minor model} of a graph $H$ in a graph $G$ is a mapping $\varphi \colon V(H) \rightarrow 2^{V(G)}$ assigning a \emph{branch set} $\varphi(v) \subseteq V(G)$ to each vertex $v \in V(H)$, such that:
\begin{itemize}
\item $G[\varphi(v)]$ is nonempty and connected for all $v \in V(H)$,
\item $\varphi(v) \cap \varphi(u) = \emptyset$ for all $u \neq v \in V(H)$, and
\item if $\{u,v\} \in E(H)$, then there exist $u' \in \varphi(u)$ and $v' \in \varphi(v)$ such that $\{u',v'\}\in E(G)$.
\end{itemize}
The third condition implies that one can find an \emph{edge mapping}~$\psi \colon E(H) \rightarrow E(G)$ such that:
\begin{itemize}
\item For all $\{u,v\} \in E(H)$, edge~$\psi(\{u,v\})$ has one endpoint in~$\varphi(u)$ and the other in~$\varphi(v)$.
\end{itemize}
We will often use the existence of this edge mapping in our proofs.
\end{definition}

For~$S \subseteq V(H)$ we define~$\varphi(S) := \bigcup_{v \in S}\varphi(v)$, and we define~$\varphi(V(H))$ as the \emph{range} of the minor model. A minor model~$\varphi$ of~$H$ in~$G$ is called \emph{minimal} if no minor model~$\varphi'$ exists with~$\varphi'(V(H)) \subsetneq \varphi(V(H))$.

\begin{definition}[labeled minor model] \label{def:minor:model:labeled}
A \emph{labeled minor model} of an $X$-labeled graph $H$ in an $X$-labeled graph $G$ is a mapping~$\varphi$ as in Definition~\ref{def:minor:model:plain}, that additionally satisfies the following: for all $v \in V(H)$ and~$\ell \in \lab_H(v)$ there exists $v' \in \varphi(v)$ such that $\ell \in \lab_G(v')$.
\end{definition}

If~$G$ contains a (labeled) minor model of~$H$, then we say that~$G$ contains~$H$ as a (labeled) minor and denote this as~$H \leqm G$. Observe that~$G$ contains~$H$ as a (labeled) minor if and only if~$H$ can be obtained from~$G$ by deleting edges and vertices (and potentially labels), and contracting edges.

\begin{lemma}[$\bigstar$] \label{lem:minormodel:countcomp}
Let~$G$ and~$H$ be unlabeled graphs, let~$X \subseteq V(G)$, and let~$\varphi$ be a minimal minor model of~$H$ in~$G$. Then~$\varphi(V(H))$ intersects at most~$|X| + |V(H)| + |E(H)|$ connected components of~$G - X$.
\end{lemma}

We denote the size of an optimal \FDeletion solution on~$G$ by~$\minFdel(G)$, and the set of optimal solutions by~$\minFdelsol(G)$. In our bounds, we use the notation~$\Oh_{z}(1)$ for some identifier(s)~$z$ to denote a constant that only depends on~$z$.

\begin{lemma}[$\bigstar$] \label{lem:compute:kernel}
Let~$\F$ be a fixed set of (unlabeled) graphs, let~$\eta \geq 1$ be a constant, and let~$X$ be a set. For any set~$\Q$ of $X$-labeled graphs and host graph~$C$ with~$\td(C) \leq \eta$, one can:
\begin{itemize}
	\item compute~$\minFdel(C)$ in~$\Oh_{\F, \eta}(|V(C)|)$ time;
	\item determine whether there is a solution~$Y \in \minFdelsol(C)$ such that~$C-Y$ contains no graph from~$\Q$ as a labeled minor, in time~$f(L, \sum _{H \in \Q}|V(H)|, \eta) \cdot |V(C)|$ for some function~$f$.
\end{itemize}
Here~$L$ counts the number of elements of~$X$ that appear in the labelset of at least one vertex in at least one graph of~$\Q$.
\end{lemma} 

\section{Overview of the main lemma}\label{sec:main-lemma-overview}

In this section we discuss Lemma~\ref{lem:main}, whose long and technical proof is deferred to the appendix. The strength of the lemma comes from the fact that the bound on~$|\Q^*|$ is \emph{independent} of the size of the graph~$C$ and of the number of labels~$|X|$ used on labelsets of vertices of~$C$. 

 \begin{figure}
 \centering
 \includegraphics{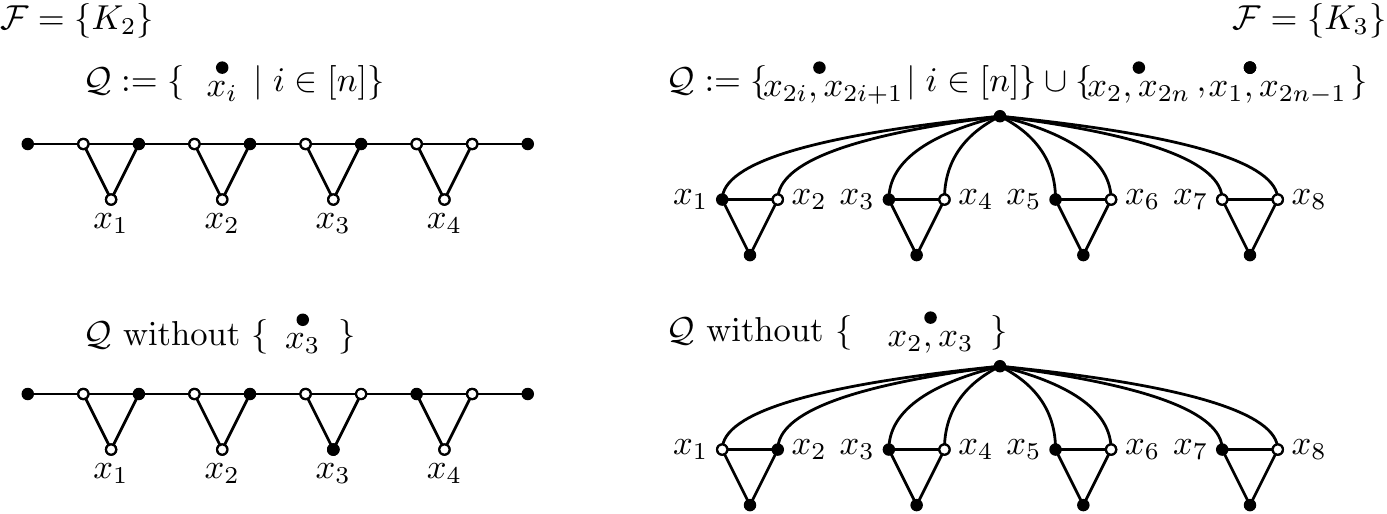}
\caption{Two constructions of graphs and sets $\Q$ for $n = 4$, where no optimal $\F$-deletion breaks~$\Q$, but for any $Q \in \Q$ there exists an optimal $\F$-deletion breaking $\Q \setminus Q$. Top: any solution breaking both~$\F$ and~$\Q$ (white vertices at the top) is larger than~$\minFdel$, but for any~$Q \in \Q$ there is a solution of size~$\minFdel$ breaking both~$\F$ and~$\Q \setminus \{Q\}$ (white vertices at the bottom).}
 \label{fig:assumptions-main-lemma-needed}
 \end{figure}

 The statement of Lemma \ref{lem:main} is best-possible in several ways. First of all, the dependence of $|\Q^*|$ on $\td(G)$ instead of $\tw(G)$ is essential.
  In Figure \ref{fig:assumptions-main-lemma-needed}~(left), a construction of a graph of treewidth $2$ together with a set $\Q$ is shown. In this graph, no optimal $\{K_2\}$-deletion (\textsc{Vertex Cover}) breaks all graphs in $\Q$. However, for any $Q \in \Q$ there is an optimal vertex cover breaking $\Q\setminus \{Q\}$. The example in Figure \ref{fig:assumptions-main-lemma-needed} can easily be extended to arbitrary $n$, showing that there is a set $\Q$ with $|\Q| = n$ such that no optimal vertex cover breaks $\Q$, yet there is no $\Q^* \subsetneq \Q$ such that no optimal vertex cover breaks $\Q^*$. Since $|\Q|$ is not bounded in terms of $\tw(G)=2$ and $\F = \{K_2\}$, this shows that $\td(G)$ cannot be replaced by $\tw(G)$.

 Secondly, the assumption that $\Q$ is $(\min_{H \in \F} |V(H)|)$-saturated cannot be avoided already for $\F = \{K_3\}$ (corresponding to \textsc{Feedback Vertex Set}). In Figure \ref{fig:assumptions-main-lemma-needed} (right) we show an example of a graph of treedepth $4$ and a set $\Q$ of size $2n+2$ that consist of single vertices of two labels each, where we again cannot properly bound the size of $\Q^*$. The example is shown for $n=4$ but can easily be generalized to arbitrary $n$, without increasing the treedepth. For any $\Q^* \subsetneq \Q$ there exists an optimal $\F$-deletion breaking $\Q^*$, while $|\Q|$ is not bounded in terms of $\td(G)$ and $\F$.

The proof of Lemma~\ref{lem:main} follows an inductive strategy that mimics how a recursive algorithm would solve \FDeletion on a bounded-treedepth graph~$C$. We pick a vertex~$v$ whose removal decreases the treedepth, and branch on whether~$v$ is part of the solution or not. If so, we remove~$v$ and recurse on a graph of smaller treedepth; if not, then we continue looking for solutions in which~$v$ is forbidden to be removed. The process builds up a set~$S$ with the property that removing~$S$ decreases the treedepth by~$|S|$, and we are only interested in solutions disjoint from~$S$. This proceeds while~$C-S$ remains connected; the branching depth is bounded since~$|S| \leq \td(C)$. When~$C-S$ becomes disconnected, we must take a more involved approach. We recurse on each of the connected components of~$C - S$ separately and find \FDeletion solutions there. But solutions for different components of~$C-S$ may not combine into a solution for~$C$, since various fragments of~$\F$-minors left behind in different components of~$C-S$, may be combined through their connections to~$S$ to form a forbidden minor. For this reason, when we recurse on connected components of~$C-S$ we place additional restrictions on the solutions chosen there, to ensure they also break \emph{fragments} of~\F-minors in such a way that the solutions can be properly combined. 

Our approach to bound the size of~$\Q^*$ is built on top of this inductive strategy. While branching over various ways to form an \FDeletion solution, we additionally branch on what fragments of labeled $\Q$-minors are left behind by the solution in the various components of~$C-S$. By exploiting the saturatedness of~$\Q$ in a crucial way, we obtain the desired bound on~$|\Q^*|$. The formalization of these ideas requires an extensive theory of how fragments of a forbidden minor in various components of~$C-S$ may combine to form a forbidden minor in~$C$, which is developed in Appendix~\ref{sec:prelims}.

\section{Kernelization for \texorpdfstring{\FDeletion}{F-Deletion}} \label{sec:kernel}
\newcommand{\rrgraphs}{\ensuremath{\mathcal{H}}\xspace} 
\newcommand{\qsize}{\ensuremath{\gamma}\xspace} 
\newcommand{\nummarked}{\ensuremath{\tau}\xspace} 
\newcommand{\qThreshold}{\ensuremath{\rho}\xspace}

In this section we describe the recursive approach to kernelize the \FDeletion problem using a constant-treedepth modulator. The correctness of this strategy will crucially depend on Lemma~\ref{lem:main}. Lemma~\ref{lem:kernel:recursive} identifies essential components in the input.

\begin{lemma} \label{lem:kernel:recursive}
Let~$\F$ be a finite set of connected graphs and let~$\eta \geq 1$ be a constant. There is a polynomial-time algorithm that, given a graph~$G$ along with a modulator~$X \subseteq V(G)$ such that~$\td(G-X) \leq \eta$, outputs an induced subgraph~$G'$ of~$G$ together with an integer~$\Delta$ such that~$\minFdel(G) = \minFdel(G') + \Delta$ and~$G' - X$ has at most~$|X|^{\Oh_{\F, \eta}(1)}$ connected components.
\end{lemma}

Before proving this lemma, we show how it implies Theorem~\ref{thm:main}.

\begin{thm:main:statement}
\mainthm
\end{thm:main:statement}
\begin{proof}
Consider an input~$(G,X,k)$ to \FDeletion. The proof is by induction on~$\eta$.

(\textbf{$\eta = 1$}) If~$\td(G-X) = 1$, then~$G-X$ is an independent set and any connected component of~$G-X$ contains one vertex. Apply Lemma~\ref{lem:kernel:recursive} to find an induced subgraph~$G'$ of~$G$ and integer~$\Delta$ such that~$\minFdel(G) = \minFdel(G') + \Delta$, which implies that~$(G,X,k)$ has answer \yes if and only if~$(G', X, k - \Delta)$ has answer \yes. Now~$G'-X$ has~$|X|^{\Oh_{\F, 1}(1)}$ single-vertex connected components. It follows that~$G'-X$ has at most~$|X| + |X|^{\Oh_{\F,1}(1)}$ vertices, which is polynomial in~$|X|$ for fixed~$\F$. Hence~$(G',X,k - \Delta)$ forms a polynomial kernel.

(\textbf{$\eta > 1$}) For~$\eta > 1$, we apply Lemma~\ref{lem:kernel:recursive} on the input~$(G,X,k)$ and find~$G'$ and~$\Delta$ as above. We will augment the modulator~$X$ into a superset~$X'$ to ensure that~$\td(G'-X') < \eta$. To this end, we consider each connected component~$C$ of~$G'-X$. If~$C$ consists of a single vertex then its treedepth is already smaller than~$\eta > 1$. Otherwise,~$C$ is a connected graph with more than one vertex, and by Definition~\ref{def:treedepth} there is a vertex~$x_C$ such that~$\td(C - \{x_C\}) < \td(C)$. Since the \textsc{Treedepth} problem parameterized by the target width is fixed-parameter tractable~\cite{ReidlRVS14}, and~$\eta$ is a constant, we can find such a vertex~$x_C$ by trying all options for~$x_C$ and computing the treewidth of the resulting graph in~$f(\eta) \cdot n^{\Oh(1)}$ time. (Alternatively, we can compute a treedepth-decomposition of~$C$ using the algorithm of Reidl et al.~\cite{ReidlRVS14} and take its root as~$x_C$.) We initialize~$X'$ as~$X$. For each component~$C$ of~$G'-X$ with treedepth larger than one, we add the corresponding treedepth-decreasing vertex~$x_C$ to~$X'$.

Since Lemma~\ref{lem:kernel:recursive} guarantees that the number of connected components of~$G'-X$ is polynomial in~$|X|$ for fixed~$\F$ and~$\eta$, the resulting modulator~$X'$ has size polynomial in~$|X|$. Moreover, it guarantees that~$\td(G'-X') < \eta$. Hence we now have an instance~$(G',X',k - \Delta)$ of \FDeletion parameterized by a treedepth-$(\eta-1)$ modulator, with the same answer as~$(G,X,k)$. We apply the kernel for the parameterization by a treedepth-($\eta-1$) modulator, which outputs an instance~$(G^*,X^*,k^*$) with the same answer as~$(G',X',k - \Delta)$ and therefore as~$(G,X,k)$. By induction, the size of~$G^*$ is bounded by some polynomial in~$|X'|$, which in turn is bounded by a polynomial in~$|X|$. Hence~$G^*$ has size~$|X|^{\Oh_{\F, \eta}(1)}$ for some suitably chosen constant, and we output~$(G^*,X^*,k^*)$ as the result of the kernelization.
\end{proof}

Now we prove Lemma~\ref{lem:kernel:recursive}.

\begin{proof}[Proof of Lemma \ref{lem:kernel:recursive}]
Let~$\mathcal{C}$ be the connected components of~$G-X$. To reduce their number, we have a single reduction rule stated in terms of labeled graphs. With each connected component~$C \in \mathcal{C}$, we naturally associate an $X$-labeled graph~$C_L$ by assigning a vertex~$v \in V(C)$ the labelset~$N_G(v) \cap X$. We are interested in which of these labeled graphs have optimal \FDeletion solutions that also hit certain fragments of potential \F-minor-models. We therefore define a set~$\rrgraphs$ which is a superset of the relevant fragments. 
We use~$\|\F\|$ as a shorthand for~$\max_{H \in \F} |V(H)|$. Let~$\rrgraphs$ consist of the connected $\|\F\|$-restricted $X$-labeled graphs that have at most~$m_\F := \max_{H \in \F} |E(H)|$ edges. We consider two $X$-labeled graphs to be identical if there is an isomorphism between them that respects the labelsets.

\begin{innerclaim}
$|\rrgraphs| \in |X|^{\Oh_{\F}(1)}$.
\end{innerclaim}
\begin{innerproof}
Graphs in~$\rrgraphs$ have at most~$m_\F + 1$ vertices. There are less than~$2^{(m_\F+1)^2}$ distinct choices for the graph structure of a member of~$\rrgraphs$, since there are less than~$2^{n^2}$ different $n$-vertex graphs. For each vertex, there are less than~$(|X|+1)^{\|\F\|}$ choices for a labelset of size at most~$\|\F\|$. Hence each graph structure~$H$ can appear with less than~$((|X|~+~1)^{\|\F\|})^{|V(H)|} \leq (|X|+1)^{\|\F\| \cdot (m_\F+1)}$ different choices of labeling function, giving an overall bound~$|\rrgraphs| \leq 2^{(m_\F+1)^2} \cdot (|X|+1)^{\|\F\| \cdot (m_\F + 1)}$ that is polynomial in~$|X|$.
\end{innerproof}

Choose~$\qsize \in \Oh_{\F, \eta}(1)$ such that Lemma~\ref{lem:main} guarantees that for this choice of~$\F$ and the treedepth bound~$\eta$, one can always find~$\Q^* \subseteq \Q$ of size at most~$\gamma$. Let~$\qThreshold := |X| + \max_{H \in \F} (|V(H)| + |E(H)|)$, and~$\nummarked := |X| + 1 + \qsize \cdot \qThreshold \in \Oh_{\F, \eta}(|X|)$. Consider the following marking procedure.

\begin{procedure}
For each set~$\mathcal{Q} \subseteq \rrgraphs$ of size at most~$\gamma$, do the following. Let
\begin{equation*}
\mathcal{C}_\Q := \left \{ C \in \mathcal{C} \mid \forall Y \in \minFdelsol(G[C])\colon C_L - Y \text{ has a graph from~$\mathcal{Q}$ as a labeled minor} \right \}.
\end{equation*}
Mark~$\nummarked$ arbitrarily chosen components from~$\mathcal{C}_\Q$, or mark all of them if there are fewer than~$\nummarked$.
\end{procedure}

Let~$\mathcal{C'} \subseteq \mathcal{C}$ denote the marked components,~$G' := G[X \cup \bigcup_{C \in \mathcal{C'}} C]$, and let~$\Delta := \sum _{C \in \mathcal{C} \setminus \mathcal{C'}} \minFdel(G[C])$. The procedure can be executed in polynomial time, using variants of Courcelle's theorem to find the sets~$\mathcal{C}_\Q$. We explain how this is done in Lemma~\ref{lem:compute:kernel}. Since~$\gamma \in \Oh_{\F, \eta}(1)$, the number of subsets of~$\rrgraphs$ over which we iterate is polynomial in~$|\rrgraphs|$ and therefore in~$|X|$. Since the graphs in~$\Q$ are $\|\F\|$-restricted, the number of labels involved is constant for fixed~$\F$ and~$\eta$, and therefore Lemma~\ref{lem:compute:kernel} guarantees a polynomial running time.

\begin{innerclaim}
$|\mathcal{C}'| \leq |X|^{\Oh_{\F, \eta}(1)}$.
\end{innerclaim}
\begin{innerproof}
The procedure loops over~$|X|^{\Oh_{\F, \eta}(1)}$ subsets~$\mathcal{Q}$. For each such set, we mark at most~$\nummarked = |X| + 1 + \qsize \cdot \qThreshold \in \Oh_{\F, \eta}(|X|)$ components.
\end{innerproof}

The pair~$(G', \Delta)$ is the desired outcome of Lemma~\ref{lem:kernel:recursive}. It remains to prove that~$\minFdel(G) = \minFdel(G') + \Delta$. This follows from Claim~\ref{claim:putback:component} by induction.

\begin{innerclaim} \label{claim:putback:component}
For any unmarked component~$C^* \in \mathcal{C} \setminus \mathcal{C'} \colon \minFdel(G) = \minFdel(G - V(C^*)) + \minFdel(G[C^*])$.
\end{innerclaim}
\begin{innerproof}
Let~$\widehat{G} := G - V(C^*)$. Clearly, any solution for the graph~$G$ can be partitioned into a solution for~$\widehat{G}$ and a solution for~$G[C^*]$, so that $\minFdel(G) \geq \minFdel(\widehat{G}) + \minFdel(G[C^*])$. We focus on proving the converse. Let~$\widehat{Y} \in \minFdelsol(\widehat{G})$ be an optimal solution on~$\widehat{G}$. Let~$X_0 := X \setminus \widehat{Y}$ and let~$\rrgraphs_0 \subseteq \rrgraphs$ contain those graphs for which the labelset of each vertex is contained in~$X_0$. Now define:
\begin{align}
\Q := \{H \in \rrgraphs_0 \mid & \text{~there are fewer than $\qThreshold$ components~$C$ of~$\widehat{G} - X$} \label{eq:def:q} \\
&\text{~whose $X$-labeled version~$C_L - \widehat{Y}$ contains~$H$ as $X$-labeled minor}\}. \nonumber
\end{align}
Intuitively, one may think of~$\Q$ as those labeled graphs (that represent potential fragments of forbidden $\F$-minors) that can be realized in only few ($\qThreshold \in \Oh_{\F}(|X|)$) components of~$\widehat{G}-X$ after removing the solution~$\widehat{Y}$. When lifting the solution~$\widehat{Y}$ in~$\widehat{G}$ to a solution in~$G$ by adding a solution in~$C^*$, it will be crucial to break all $X$-labeled minor models of~$\Q$ in~$C^*$; the fragments~$\rrgraphs_0 \setminus \Q$ that remain in \emph{many} different components turn out to be irrelevant.

For a subset~$X' \subseteq X_0$ of labels, let~$I_{X'}$ be the labeled graph consisting of a single vertex with labelset~$X'$. Let~$n_\F := \min_{H \in \F} |V(H)|$ and observe that~$n_\F \leq \qThreshold$. We prove:
\begin{equation} \label{eq:singletons:inq}
\forall X' \subseteq X_0, |X'| = n_\F \colon I_{X'} \in \Q.
\end{equation}
Suppose~$I_{X'} \notin \Q$ for suitable~$X'$. Then there are~$\qThreshold \geq n_\F$ components of~$\widehat{G} - X$ that have~$I_{X'}$ as labeled minor after removing the solution~$\widehat{Y}$. Take~$n_\F$ such components~$C_L^1, \ldots, C_L^{n_\F}$, and associate each one to a distinct vertex of~$X' \subseteq V(\widehat{G}) \setminus \widehat{Y}$. The fact that~$I_{X'}$ is a labeled minor of~$C_L^i - \widehat{Y}$ for each~$i$, implies that in each such component there is a connected vertex subset~$S_i \subseteq V(C^i_L) \setminus \widehat{Y}$ such that each label of~$X'$ appears at least once on a vertex of~$S_i$. Considering the corresponding vertex subset in~$\widehat{G} - \widehat{Y}$ and taking into account that the labeling of~$C^i_L$ represents adjacency to~$X$ in~$G$, this implies that we can contract each~$S_i$ into a single vertex~$s_i$ that becomes adjacent to all vertices of~$X'$. Then contract each~$s_i$ into a distinct vertex of~$X'$: these minor operations on graph~$\widehat{G}-\widehat{Y}$ turn~$X'$ into a clique of size~$n_\F$. Hence any graph on~$n_\F$ vertices is a minor of~$\widehat{G}-\widehat{Y}$, contradicting that~$\widehat{G}-\widehat{Y}$ is~$\F$-minor-free since~$\F$ has a graph on~$n_\F$ vertices. So \eqref{eq:singletons:inq} holds.

Now consider the unmarked component~$C^*$ in the statement of Claim~\ref{claim:putback:component}, and consider its labeled version~$C^*_L$. We say that a vertex set~$Y$ \emph{breaks} the minor models of the $X_0$-labeled graphs $\Q$ in $C^*_L$, or simply \emph{breaks} $\Q$ in $C^*_L$, if $C^*_L - Y$ does not contain any graph in $\Q$ as a labeled minor. We first show the following.
\begin{equation} \label{eq:break:q}
\exists Y^* \in \minFdelsol(G[C^*]) \colon \text{~$Y^*$ breaks $\Q$ in~$C^*_L$}.
\end{equation}
To establish~\eqref{eq:break:q}, assume that no solution of size~$\minFdel(G[C^*])$ in~$G[C^*]$ breaks~$\Q$. We will use Lemma~\ref{lem:main}, together with our marking scheme, to argue for a contradiction. Observe that~\eqref{eq:singletons:inq} implies that~$\Q$ is an~$n_\F$-saturated set of~$X_0$-labeled graphs. If no optimal solution on~$G[C^*]$ breaks~$\Q$, then by Lemma~\ref{lem:main} there is a set~$\Q^* \subseteq \Q$ of size at most~$\gamma$ such that no optimal solution on~$G[C^*]$ breaks~$\Q^*$. Since the assumption that~\eqref{eq:break:q} does not hold means that the unmarked~$C^*$ was eligible to be marked for the set~$\mathcal{C}_{\Q^*}$ in our procedure above, it has marked~$\nummarked$ other components~$C^1, \ldots, C^\nummarked \in \mathcal{C}_{\Q^*}$ of~$G-X$. For each~$i \in [\nummarked]$, there is no \FDeletion solution of size~$\minFdel(G[C^i])$ in~$G[C^i]$ that breaks~$\Q^*$ in the labeled version~$C^i_L$. Since~$\Q^* \subseteq \Q$, by~\eqref{eq:def:q} we have for each graph~$H \in \Q^*$ that there are fewer than~$\qThreshold$ components~$C^i$ among~$C^1, \ldots, C^\nummarked$ for which~$C^i_L - \widehat{Y}$ contains~$H$ as a labeled minor. Since~$|\Q^*| \leq \gamma$, it follows that there are at most~$\gamma \cdot \qThreshold$ indices~$i \in [\nummarked]$ for which~$C^i_L - Y$ contains some graph from~$\Q^*$ as a labeled minor. But since~$\nummarked = |X| + 1 + \gamma \cdot \qThreshold$, there are at least~$|X|+1$ components~$C^i_L$ in which all~$\Q^*$-minors are broken by~$\widehat{Y}$. Since no \emph{optimal} solution breaks~$\Q^*$ in the marked components, we have~$|\widehat{Y} \cap V(C^i)| > \minFdel(G[C^i])$ for at least~$|X|+1$ components. But this contradicts that~$\widehat{Y}$ is an optimal solution to \FDeletion on~$\widehat{G}$: since~$\F$ consists of connected graphs, we can form a solution~$\widehat{Y}'$ by taking~$X$ together with a set of size~$\minFdel(\widehat{G}[C])$ from each component~$C$ of~$\widehat{G}-X$. Since~$|\widehat{Y}' \cap V(C)| \leq |\widehat{Y} \cap V(C)|$ for all~$C \in \mathcal{C}$, with strict inequality for at least~$|X|+1$ components, we have~$|\widehat{Y}'| < |\widehat{Y}|$. This contradicts that~$\widehat{Y}$ is an optimal solution and establishes~\eqref{eq:break:q}.

Hence there exists a \Fdel $Y^*$ in $C^*_L$ breaking $\Q$ of size $\minFdel(G[C^*])$. We prove:
\begin{equation} \label{eq:updated:solution}
\widehat{Y} \cup Y^*\text{~is a solution to \FDeletion on~$G$.}
\end{equation}
This will complete the proof of Claim~\ref{claim:putback:component}, since~$|\widehat{Y} \cup Y^*| = \minFdel(\widehat{G}) + \minFdel(G[C^*])$. Assume for a contradiction that~$\widetilde{G} := G - (\widehat{Y} \cup Y^*)$ contains some graph~$H \in \F$ as a minor. Consider a minimal minor model of~$H$ in~$\widetilde{G}$, which is given by a vertex mapping~$\varphi \colon V(H) \to 2^{V(\widetilde{G})}$, and let~$\psi \colon E(H) \to E(\widetilde{G})$ be a corresponding edge mapping.

Out of all possible minimal minor models of~$H$ in~$\widetilde{G}$, select a model~$(\varphi, \psi)$ that minimizes the quantity~$|\varphi(V(H)) \cap V(C^*)|$. Observe that if~$\varphi(V(H)) \cap V(C^*) = \emptyset$, then~$\varphi$ is also a valid model in~$\widehat{G} - \widehat{Y}$, contradicting that~$\widehat{Y}$ is a solution to \FDeletion on~$\widehat{G}$. So in the remainder we consider the case that the minor model contains at least one vertex of~$C^*$. We will build a minimal minor model of~$H$ in~$G$ using strictly fewer vertices of~$C^*$, thereby contradicting the choice of~$(\varphi, \psi)$. 
\begin{figure}
\centering
\includegraphics{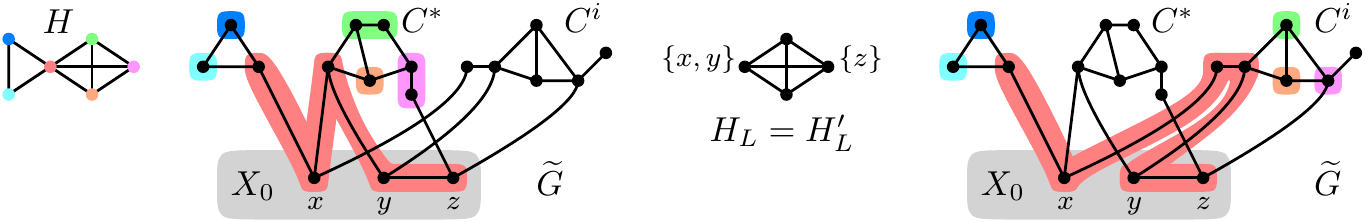}
\caption{This figure shows how to define $H_L$ based on $H$ and $\widetilde{G}$, and how to modify the minor model of $H$ in $\widetilde{G}$ such that it uses fewer vertices of $C^*$, in the proof of \eqref{eq:updated:solution} in Claim \ref{claim:putback:component}.}
\label{fig:replacement-argument}
\end{figure}
Consider the $X_0$-labeled subgraph~$H'_L$ of~$\widetilde{G}$ obtained by the following procedure, which is illustrated in Figure~\ref{fig:replacement-argument}:
\begin{enumerate}
	\item Start from the $X_0$-labeled subgraph of~$\widetilde{G}$ induced by~$\bigcup _{v \in V(H)} \varphi(v) \cap V(C^*)$, where each vertex~$u$ has labelset~$N_G(u) \cap X_0$. As observed above, this subgraph is not empty.
	\item Remove all edges from this subgraph, except those in the range of~$\psi$ and those that connect two vertices that belong to a common branch set under~$\varphi$.
	\item Contract every edge between two vertices that belong to a common branch set of~$\varphi$, obtaining an $X_0$-labeled graph~$H'_L$. (Recall that labelsets merge during edge contraction.)
\end{enumerate}
Observe that~$H'_L$ has at most~$|E(H)|$ edges, since each edge remaining in~$H'_L$ corresponds to an edge in the range of~$\psi$. We claim that~$H'_L$ is an $n_\F$-restricted graph: the labelset of each vertex has size less than~$n_\F$. To see this, observe that if some vertex of~$H'_L$ has a labelset~$X' \subseteq X_0$ of size at least~$n_\F$, then the pre-image of this vertex corresponds to a connected vertex subset~$A$ of~$\varphi(V(H)) \cap V(C^*)$ such that~$|N_G(A) \cap X_0| \geq n_\F$. Since~$(\varphi, \psi)$ is a minor model in~$\widetilde{G} = G - (\hat{Y} \cup Y^*)$, this would imply that~$C_L^* - Y^*$ has the one-vertex graph~$I_{X'}$ with labelset~$X'$ as a labeled minor. But~$I_{X'} \in \Q$ by \eqref{eq:singletons:inq}, while~$Y^*$ breaks all labeled $\Q$-minors in~$C^*_L$ by definition; a contradiction. Hence~$H'_L$ is indeed~$n_\F$-restricted.

Let~$H_L$ be an arbitrary connected component of~$H'_L$. Since~$H_L$ is connected,~$n_\F$-restricted, and contains at most~$|E(H)|$ edges, we have~$H_L \in \mathcal{H}_0$. As~$H_L$ clearly occurs as a labeled minor of~$C^*_L - Y^*$, while~$Y^*$ breaks~$\Q$ in~$C^*_L$, we have~$H_L \notin \Q$. By definition of~$\Q$, this implies there are at least~$\qThreshold$ connected components~$C^1, \ldots, C^\qThreshold$ of~$\widehat{G} - X$ such that~$C^i_L - \widehat{Y}$ contains~$H_L$ as $X_0$-labeled minor for each~$i \in [\qThreshold]$. By Lemma~\ref{lem:minormodel:countcomp}, the minimal model~$(\varphi, \psi)$ in~$\widetilde{G}$ intersects at most~$|X| + |V(H)| + |E(H)| \leq \qThreshold$ components of~$\widetilde{G} - X$ and therefore of~$G - X$. Since~$\varphi(V(H))$ also intersects~$C^* \notin \{C^1, \ldots, C^\qThreshold\}$, it follows that some~$C^i$ is disjoint from the range of~$(\varphi, \psi)$.

To finish the argument, fix~$C^i$ such that~$\varphi(V(H)) \cap V(C^i) = \emptyset$ and~$C^i_L - \widehat{Y}$ contains~$H_L$ as $X_0$-labeled minor. Let~$T$ denote the vertices of~$\varphi(V(H)) \cap V(C^*)$ whose contraction in the process above resulted in the connected component~$H_L$ of~$H'_L$. Then it is straightforward to verify that~$G[(\varphi(V(H)) \setminus T) \cup (C^i - \widehat{Y})]$ contains~$H$ as a minor. The role that vertices of~$T$ played in the minor model~$(\varphi, \psi)$ can be replaced by the vertices of~$C^i_L - \widehat{Y}$: each edge of~$\psi$ that was realized between vertices of~$T$ yielded an edge of~$H_L$ which is realized by a labeled $H_L$-minor in~$C^i_L - \widehat{Y}$; each fragment of a branch set that was realized within~$C^*$ yielded a vertex of~$H_L$ that is realized in the $H_L$-minor in~$C^i_L - \hat{Y}$; and finally the connectivity of the branch sets is ensured because the labeling ensures that for all fragments of branch sets in~$T$ that were adjacent to vertices of~$X - \widehat{Y} = X_0$, the branch set of the $H_L$-minor in~$C^i - \widehat{Y}$ realizing that fragment is also adjacent to all those vertices of~$X_0$. Hence there is a \emph{minimal} $H$-minor in~$\widetilde{G}$ whose range is a subset of~$(\varphi(V(H)) \setminus T) \cup (C^i - \widehat{Y})$. Since~$T \subseteq C^*$ is not empty, this contradicts our choice of~$(\varphi, \psi)$ as a minimal $H$-model minimizing the intersection with~$C^*$.
\end{innerproof}
This concludes the proof of Lemma~\ref{lem:kernel:recursive}.
\end{proof} 

\section{Conclusion} \label{sec:conclusion}
Our goal in this paper was to obtain polynomial kernelizations for a wide range of graph problems, in terms of a rich class of structural parameterizations. We obtained polynomial kernelizations for \FDeletion problems parameterized by a constant-treedepth modulator. The kernelization algorithm as presented here is only of theoretical interest. While the kernel size is polynomial for fixed~$\mathcal{F}$ and~$\eta$, the degree of the polynomial grows very quickly with~$\mathcal{F}$ and~$\eta$. It would be desirable to have a \emph{uniformly polynomial} kernel size, of the form~$f(\mathcal{F}, \eta) |X|^c$ for some constant~$c$ and function~$f$. Unfortunately, Theorem~\ref{thm:lowerbound} shows that even for the simplest choice of~$\mathcal{F}$, corresponding to the \textsc{Vertex Cover} problem, the degree of the polynomial must depend exponentially on~$\eta$ and no uniformly polynomial kernelization exists. The bad news also extends in the other direction: when taking the simplest choice for~$\eta$ and working with a treedepth-one modulator (a vertex cover), the degree of the polynomial in the kernel size for \FDeletion must depend on~$\mathcal{F}$~\cite[Theorem 1.1]{GiannopoulouJLS17} and a uniformly-polynomial kernel does not exist. 


\bibliography{H-min-free-paper}
\clearpage

\appendix

\section{Omitted proofs from Section \ref{sec:basic:prelims}} \label{sec:omitted:proofs}
\begin{proof}[Proof of Lemma \ref{lem:minormodel:countcomp}]
Consider a minimal minor model~$\varphi$ and let~$\psi \colon E(H) \to E(G)$ be a corresponding edge mapping. For each~$v \in V(H)$, the graph~$G[\varphi(v)]$ is connected by definition; let~$T_v$ be a spanning subtree of~$G[\varphi(v)]$. 

For each tree~$T_v$ that consists of more than one vertex, all leaves of~$T_v$ are incident on an edge in the range of~$\psi$: if~$u \in V(G)$ is a leaf of~$T_v$ not incident on an edge in the range of~$\psi$, then replacing~$\varphi(v)$ by~$\varphi(v) \setminus \{u\}$ preserves connectivity of the branch set and validity of the edge mapping~$\psi$. This yields a minor model of~$H$ in~$G$ of smaller range, contradicting the minimality of~$\varphi$.

We call a connected component~$C$ of~$G - X$ a \emph{terminal component for~$v$} if one of the following holds:
\begin{itemize}
	\item component~$C$ contains a vertex of~$T_v$ incident on an edge in the range of~$\psi$, or
	\item $T_v$ is a single-vertex tree contained in~$C$ (which occurs when~$v$ is isolated in~$H$).
\end{itemize}
A component~$C$ of~$G-X$ is a \emph{terminal component} if it is a terminal component for some~$v \in V(H)$. Observe that an edge~$\psi(e)$ cannot have endpoints in two different components of~$G - X$, as the presence of such an edge would mean that they are connected and form a single component. Hence each edge of~$H$ contributes at most one terminal component, implying that the total number of terminal components is at most~$|V(H)| + |E(H)|$.

Call a connected component~$C$ of~$G - X$ a \emph{nonterminal component for~$v$} if~$T_v$ contains a vertex of~$C$, but~$C$ is not a terminal component for~$v$. Intuitively, the minor model uses~$C$ to connect two vertices of~$\varphi(v) \cap X$. For~$v \in V(H)$ define~$X_v := X \cap T_v = X \cap \varphi(v)$. We bound the number of nonterminal components for~$v$ by~$|X_v| - 1$, as follows. 

Consider the graph~$T'_v$ on vertex set~$X_v$ obtained from~$T_v$ by repeatedly contracting any edge that has at most one endpoint in~$X_v$, which is possible since~$T_v$ is connected. If~$C$ is a nonterminal component for~$v$, then each component of~$T_v \cap C$ has at least two $T_v$-neighbors in~$X_v$ since~$T_v$ has no leaves in~$C$ by our observation above. Hence in the contraction process that turns~$T_v$ into~$T'_v$, the contraction of a nonterminal component~$C$ contributes at least one edge to~$T'_v$. No other component can contribute this same edge, as that would contradict the fact that~$T_v$ is acyclic. Hence the number of nonterminal components for~$v$ is bounded by the number of edges of~$T'_v$. As any contraction of an acyclic graph is acyclic, it follows that~$T'_v$ is an acyclic graph on vertex set~$X_v$. Hence it has at most~$|X_v| - 1$ edges, yielding the desired bound on the number of nonterminal components for~$v$. 

Since each vertex of~$X$ belongs to at most one branch set, we have $\sum_{v \in V(H)} |X_v| \leq |X|$ and hence the total number of nonterminal components is at most~$|X|$. As each component of~$G-X$ that intersects the range of~$\varphi$ is a terminal or nonterminal component for some~$v \in V(H)$, this proves Lemma~\ref{lem:minormodel:countcomp}.
\end{proof}

\begin{proof}[Proof of Lemma \ref{lem:compute:kernel}]
This essentially follows from the fact that graphs of bounded treedepth have bounded treewidth, together with known meta-theorems for graphs of bounded treewidth. The second task is the more interesting one; it can be achieved by comparing the quantity~$\minFdel(C)$ to the smallest size of a vertex subset of~$C$ that breaks all~$\F$-minors and all~$\Q$-minors. These quantities can be computed in linear time for fixed~$\F$ and~$\eta$ using the facts that~$C$ has constant treewidth, and that the questions can be phrased in the framework of linear Monadic Second Order Logic (MSOL) optimization extremum problems by Arnborg, Lagergren, and Seese~\cite[Theorem 5.6]{ArnborgLS91}. The existence of a labeled minor in a labeled graph can be formulated in MSOL on labeled graphs. This allows the problem of finding a smallest vertex subsets that removes all such minors to be formulated as finding the smallest vertex subset that satisfies an MSOL formula whose length depends only on~$\F$ and the total size of the graphs in~$\Q$, and where the number of labels in the instance depends only on the number~$L$ of labels used on graphs in~$\Q$; note that labels of~$C$ that do not occur on any graph in~$\Q$ can be safely forgotten.
\end{proof}

\section{Framework for boundaried labeled graphs} \label{sec:prelims}

In this section we set up a careful framework of definitions for working with labeled boundaried graphs. We establish several of their properties, which will be used extensively in the (technical) proof of the main lemma in Section~\ref{sec:main:lemma}.

\subsection{Labeled and boundaried graphs}

\begin{definition}[boundaried graph]
A \emph{$t$-boundaried graph} $G$ is a graph with boundary set $B \subseteq V(G)$ together with an injective boundary function $\bound_G \colon B \rightarrow [t]$. We denote the boundary $B$ of $G$ by $\delta(G)$. For $S \subseteq V(G)$, let $b_G(S) := \{b_G(u) \mid u \in (S \cap \delta(G))\}$ be the (possibly empty) set of boundary labels that are present in $S$. We use~$b_G(v)$ as a shorthand for~$b_G(\{v\})$.
\end{definition}
Note that by the above definition, every boundary vertex has a unique number between $1$ and $t$, but not every integer between $1$ and $t$ corresponds to a vertex in the boundary of $G$. Furthermore, every $t$-boundaried graph, is by definition also a $(t+1)$-boundaried graph.

For ease of presentation, we will often not define the boundary function of a $t$-boundaried graph explicitly. Only when relevant will we consider the values of~$\bound_{G}(v)$, but it should be understood that all $t$-boundaried graphs have such a boundary function. Edges between boundary vertices cannot be contracted. When contracting an edge incident with exactly one boundary vertex, the vertex resulting from the contraction takes over the boundary role.

An \emph{$X$-labeled boundaried graph} has both a label and a boundary function associated with it. As with the boundary function, we will only make the presence of the labeling function explicit when it is relevant to do so. When contracting an edge~$\{u,v\}$ in a labeled graph, the vertex resulting from the contraction is labeled by the union of the labelsets of~$u$ and~$v$.

\begin{definition}[isomorphism]\label{def:isomorphism} We extend the definition of \emph{graph isomorphism} to boundaried labeled graphs as follows. We say two $t$-boundaried $X$-labeled graphs $G$ and $G'$ are \emph{isomorphic} if there is an isomorphism $f \colon V(G) \rightarrow V(G')$ for which the following additional conditions hold.
\begin{itemize}
\item For all $v \in \delta(G)$, $\bound_G(f(v)) = \bound_{G'}(v)$, and
\item for all $v \in \delta(G')$, $\bound_{G}(f^{-1}(v)) = \bound_{G'}(v)$, and
\item for all $v \in V(G)$, $\lab_G(v) = \lab_{G'}(f(v))$.\qedhere
\end{itemize}
\end{definition}

Intuitively, labeled boundaried graphs are isomorphic if there is an isomorphism that preserves the labelsets and the indices of the boundary vertices.

\begin{definition}[universe of graphs]\label{def:graphs}
Let~$X$ be a finite set and let~$t \geq 0$ be an integer. We define the following sets of graphs, where two graphs are considered to be identical if they are isomorphic according to Definition~\ref{def:isomorphism}.
\begin{enumerate}
	\item Let $\graphs$ be the set of all (finite, undirected, simple) graphs.
\item Let $\congraphs$ be the set of all connected graphs.
	\item Let $\lgraphs$ be the set of all $X$-labeled graphs.
	\item Let $\bgraphs$ be the set of all $t$-boundaried graphs.
	\item Let $\attbgraphs$ (for \emph{attached}) be the subset of \bgraphs for which each connected component contains at least one boundary vertex.
\item Let $\mathcal{G}_{t,\textsc{def}(i)}$ be the set $t$-boundaried graphs, such that there exists a vertex $v$ with $\bound(v) = i$.
\end{enumerate}
We will also consider combinations of these identifiers in the obvious way. For example,~$\lbgraphs$ is the set of all $t$-boundaried $X$-labeled graphs.
\end{definition}
When convenient, we will interpret an unlabeled graph as an $X$-labeled graph, and a graph without boundary as a $0$-boundaried graph.

\begin{definition}[\setcomponents]
Let $G \in \lbgraphs$ be a graph, let \setcomponents (short for \emph{connected components}) be defined as
\begin{align*}
\setcomponents(G) &:= \{C \in \lbgraphs \mid C \text{ is a connected component of }G\}\text{ and}\\
\components(G) &:= |\setcomponents(G)|. \qedhere
\end{align*}
\end{definition}

Note that contrary to Definition~\ref{def:graphs}, if~$G$ contains multiple connected components that are isomorphic, then each of these appears as a distinct object in~$\setcomponents(G)$.

\begin{definition}[boundaried minor model] \label{def:minor:model:boundaried}
A \emph{boundaried minor model} of a $t$-boundaried graph $H$ in a $t$-boundaried graph $G$ is a mapping~$\varphi$ as in Definition~\ref{def:minor:model:plain}, that additionally satisfies the following for all $v \in V(H)$:
\begin{equation*}
    \bound_G(\varphi(u)) =
    \begin{cases}
      \emptyset, & \text{if } u \notin \delta(H) \\
      \{\bound_H(u)\}, & \text{otherwise.}
    \end{cases}
  \end{equation*}
\end{definition}
A \emph{boundaried labeled minor model} simultaneously satisfies the conditions of Definitions~\ref{def:minor:model:labeled} and~\ref{def:minor:model:boundaried}. Refer to Figure \ref{fig:minor} for an illustration of a boundaried labeled minor model.

\begin{figure}
\centering
\includegraphics{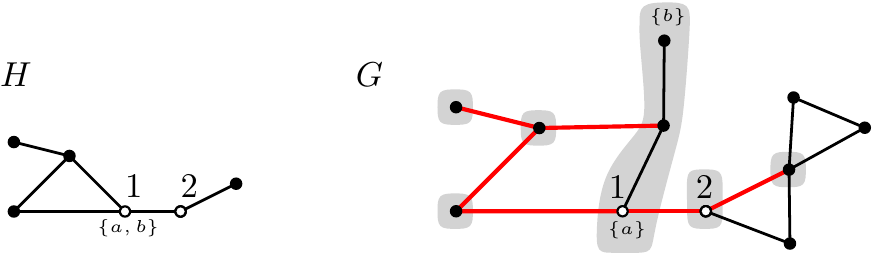}
\caption{Illustration of Definition \ref{def:minor:model:boundaried}, showing $2$-boundaried $\{a,b\}$-labeled graphs $H$ and $G$ such that $H \leqlb G$, together with a boundaried labeled minor model of $H$ in $G$ and an edge model (the red edges).}
\label{fig:minor}
\end{figure}

Using this notion of boundaried minors, we can define solutions to constrained versions of the \FDeletion problem.

\begin{definition}[\minFdel]
Let $G \in \bgraphs$ with boundary set~$S$, let~$\F \subseteq \graphs$, and let $\Pi$ be a set of $t$-boundaried graphs with boundary~$S$. Define
\begin{align*}
\minFdel(G,\Pi,S) := \min \{|Y| \mid \ &Y \subseteq V(G) \wedge Y \cap S = \emptyset, \text{ and}\\
& \text{$G-Y$ is \F-minor-free, and} \\
& G-Y \text{ has no graph in $\Pi$ as boundaried minor} \}.
\end{align*}
Define $\minFdel(G):= \minFdel(G,\emptyset, \emptyset)$, or simply the size of an optimal \F-minor free deletion in $G$. 
\end{definition}


When analyzing the structure of \FDeletion problems, it will be convenient to argue about all possible graphs that can be obtained as a minor. The following notion captures this concept.

\begin{definition}[\folio]
For $G \in \lbgraphs$, the \folio of $G$ consist of all minors of $G$:
\[\folio(G) := \{G' \in \lbgraphs \mid G' \leqlb G\}. \qedhere\]
The folio of an unlabeled graph, or an unboundaried graph, is defined analogously.
\end{definition}

\begin{definition}[\forget]
Let $G \in \lbgraphs$ and let $k \leq t$. Define
$\forget(G, k)$ as the $k$-boundaried $X$-labeled graph $G'$ obtained from~$G$ by setting $\bound_{G'}(v) = \bound_G(v)$ for all $v \in V(G)$ for which $\bound_G(v) \leq k$, and forgetting the boundary status of the higher-indexed boundary vertices. Define $\forget(G) := \forget(G,0)$.

For a set of graphs $S$, define $\forget(S,k) := \{\forget(G,k) \mid G \in S\}$ and $\forget(S) := \forget(S,0)$.
\end{definition}

Observe that the~$\forget$ operation is only used to forget the boundary status of a vertex; it is not used to omit labels from a labelset.

We now introduce the \extend  (for \emph{extend}) operation, which is illustrated in Figure~\ref{fig:extend}. It is useful when translating the question whether~$H$ is an ordinary minor of~$G$, into a question whether an extended version of~$H$ is a minor of a $t$-boundaried version of~$G$ in which $t$ distinct vertices have been selected as the boundary. The presence of the boundary can then be used in a divide-and-conquer type of approach.

\begin{figure}
\centering
\includegraphics{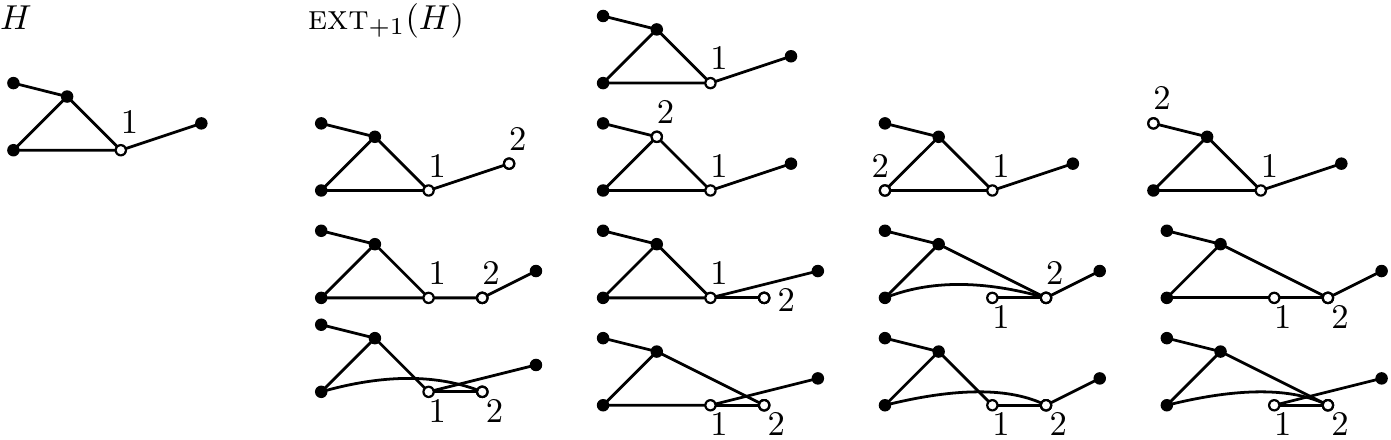}
\caption{Illustration of Definition~\ref{def:extend} (\extend). This figure shows $1$-boundaried graph $H$, together with $\extend_{+1}(H)$.}
\label{fig:extend}
\end{figure}

\begin{definition}[\extend] \label{def:extend}
Let $H \in \lbgraphs$ for some~$t \geq 0$. Let $\extend_{+1}(H)$ (short for \emph{extend}) be the set of all $(t+1)$-boundaried graphs $H'$ that can be obtained from $H$ by using exactly one of the following steps:
\begin{itemize}
\item Let $H'$ be equal to $H$, thereby forming~$H'$ as a $(t+1)$-boundaried graph in which there is no $t+1$'th boundary vertex.
\item Take a vertex $v \in V(H)\setminus \delta(H)$ and set $\bound_{H'}(v) := t+1$.
\item Split a vertex $u \in \delta(H)$ as follows. Let $V(H') := V(H) \cup \{v\}$. Let $\bound_{H'}(v) := t+1$. Add edge $\{u,v\}$ to $H'$.
    For any edge $\{u',u\} \in E(H)$ either keep it in $H'$ or replace it by edge $\{u',v\}$. For each label~$\ell$ on the labelset of~$u$, either keep it on~$u$ or move it to the labelset of~$v$.
\end{itemize}
Define $\extend_{+t'}(H)$ as the set of~$(t+t')$-boundaried graphs that can be obtained from~$H$ by applying exactly~$t'$ of such operations in a row. The extend operation for unlabeled graphs is defined analogously, with the exception that there are no labels to be divided in the uncontract step. For a set of graphs~$\Q$, define~$\extend_{+1}(\Q) := \bigcup _{Q \in \Q} \extend_{+1}(Q)$, and~$\extend_{+t}(\Q)$ analogously.
\end{definition}

Observe that, in terms of graph structure, the operations above can be reversed by minor operations that contract an edge (in the third case). This implies that, for all~$H \in \lbgraphs$ and~$H' \in \extend_{+t'}(H)$ we have that~$H \leqlb \forget(H', t)$: after forgetting the boundary status of the last~$t'$ boundary vertices, the resulting~$t$-boundaried labeled graph has~$H$ as a minor.

The next lemma follows quite easily from the above definitions. It shows that a $t$-boundaried graph~$H$ can be obtained from an unboundaried graph by~$t+1$ extend operations if and only if it can be obtained by~$t$ extend operations. It is later used to justify a technical step in the proof of the main lemma.

\begin{lemma}
\label{item:remove_boundary_vertex_reduce_extend} Let $G \in \lgraphs$, $H \in \lbgraphs$.  Then $H \in \extend_{+t+1}(G) \Leftrightarrow H \in \extend_{+t}(G)$.
\end{lemma}
\begin{proof}
If $H \in \extend_{+t}(G)$, it trivially follows that $H \in \extend_{+t+1}(G)$.

Suppose $H \in \extend_{+t+1}(G)$. Then by definition, there exists $H' \in \extend_{+t}(G)$ such that $H \in \extend_{+1}(H')$. Since the $H$ is a $t$-boundaried graph, it follows that the $(t+1)$'th boundary vertex is undefined and thus since $H \in \extend_{+1}(H')$, the only operation that leaves this vertex undefined is letting $H := H'$. It follows that $H \in \extend_{+t}(G)$.
\end{proof}

The following lemma gives a key property of the extend operation. It shows that if a labeled graph~$H$ with a boundary of size~$t'$ (possibly~$0$) is a minor of the labeled $t'$-boundaried graph~$\forget(G,t')$ obtained from~$G \in \lbgraphs$ by forgetting the boundary status of its last~$(t-t')$ boundary vertices, then there is an extension~$H' \in \extend_{+(t-t')}(H)$ that appears as a labeled $t$-boundaried minor in~$G$. To see that this extend operation is really necessary with our definition of boundaried minor models, consider Figure~\ref{fig:extend_needed}.

\begin{figure}
\centering
\includegraphics{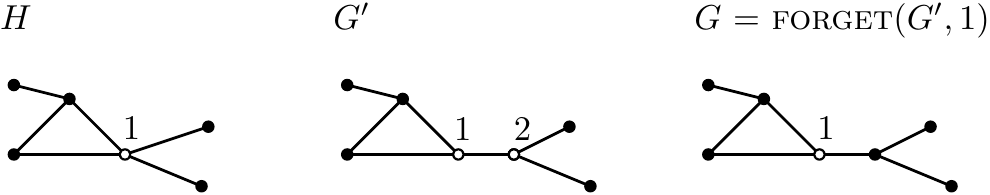}
\caption{This figure illustrates the reason for having an \extend operation, that can split vertices. We show graph $H$, and $2$-boundaried graph $G$. Here $H \leqlb G$, but $H$ cannot be obtained from graph $G'$ using minor operations (and adding a boundary vertex), since the edge between boundary vertices $1$ and $2$ cannot be contracted.}
\label{fig:extend_needed}
\end{figure}

\begin{lemma}\label{lem:extend:minor}
Let $H \in \graphs^X_{t'}$ for some integer~$t' \geq 0$ and let $G \in \lbgraphs$ with~$t \geq t'$ such that $H \leqlb \forget(G,t')$. Then there exists $H' \in \extend_{+(t-t')}(H) \subseteq \lbgraphs$ such that $H' \leqlb G$.
\end{lemma}
\begin{proof}
We prove the statement by induction on~$t-t'$. If~$t = t'$ then the statement is trivial since~$\forget(G,t) = G$ and~$H \in \extend_{+0}(H)$, so the base case~$t-t'=0$ holds.

For the induction step, suppose that~$t > t'$. Recall that~$\forget(G, t-1)$ is the $(t-1)$-boundaried graph obtained by omitting the $t$'th vertex from the boundary, but keeping it in the graph. By induction, there is a graph~$\widehat{H} \in \extend_{+(t-1-t')}(H)$ such that~$\widehat{H} \leqlb \forget(G, t-1)$; let~$(\varphi, \psi)$ be a corresponding labeled boundaried minor model. If boundary vertex $t$ is not defined in $G$, thus if there exists no vertex $x\in\delta(G)$ with $\bound_G(x) = t$, then trivially $\forget(G, t-1) = G$ and $(\varphi, \psi)$ is a boundaried labeled minor model of $\widehat{H}$ in $G$. Thereby, we can choose $H' := \widehat{H}$.

 Otherwise, let~$x \in \delta(G)$ such that $\bound_G(x) = t$ be the $t$'th boundary vertex of~$G$, which does not belong to the boundary of~$\forget(G, t-1)$. We consider the role of~$x$.

(\textbf{Not in range}) If~$x$ is not in the range of~$\varphi$, then let~$H'$ be~$\widehat{H}$. Then~$H' \in \extend_{+1}(\widehat{H})$, and since~$\widehat{H} \in \extend_{+(t-1-t')}(H)$ it follows that~$H' \in \extend_{+(t-t')}(H)$. A minor model of~$H'$ in~$G$ is given by~$(\varphi,\psi)$.
	
(\textbf{Part of a non-boundary branch set}) If~$x \in \varphi(u)$ for some~$u \notin \delta(\widehat{H})$, that is,~$x$ is part of the branch set of a non-boundary vertex of~$\widehat{H}$, then obtain~$H'$ from~$\widehat{H}$ by setting~$\bound_{H'}(u) = t$ so that~$u$ becomes the $t$'th boundary vertex of~$H'$. Model~$(\varphi, \psi)$ shows that~$H' \leqlb G$.
	
(\textbf{Part of a boundary branch set}) Finally, if~$x \in \varphi(u)$ for some~$u \in \delta(\widehat{H})$, that is,~$x$ is part of the branch set of one of the at most~$t-1$ boundary vertices of~$\widehat{H}$, then the procedure is somewhat more delicate.
Let~$y \in \delta(G)$ such that~$\bound_{G}(y) = \bound_{\widehat{H}}(u)$.
Consider the connected components~$C_1, \ldots, C_k$ of~$G[\varphi(u) \setminus \{x\}]$; we may have~$k=1$. By definition of a boundaried minor model, we know that~$y \in  \varphi(u)$. Since the boundary vertices of~$G$ are all distinct, we have~$y \neq x$. Hence~$y$ is contained in one of the connected components of~$G[\varphi(u) \setminus \{x\}]$. Assume without loss of generality that~$y \in C_1$. We obtain~$H'$ from~$\widehat{H}$ by the reverse of an edge contraction operation, as follows.
\begin{itemize}
	\item Initialize~$H'$ as a copy of~$\widehat{H}$, into which a new boundary vertex~$v$ with $\bound_{H'}(v) = t$  is inserted.
	\item Add the edge~$\{u, v\}$.
	\item For each edge~$\{u, w\} \in E(\widehat{H})$, we may re-attach it to~$v$ instead of~$u$, based on the following distinction.
	By Definition~\ref{def:minor:model:plain}, we know that~$\psi(\{u, w\})$ has exactly one endpoint~$b$ in~$\varphi(u)$.
	\begin{itemize}
		\item If~$w' \notin C_1$, then replace the edge~$\{u, w\}$ in~$\widehat{H}$ by the edge~$\{v, w\}$ in~$H'$.
		\item If~$w' \in C_1$, then the edge is preserved in~$H'$.
	\end{itemize}
	\item The labelsets of all vertices of~$H'$ match the labelsets of the corresponding vertices in~$H$, with the exception of possibly~$u$. In addition, we define a labelset for~$v$. For each label~$\ell$ on the labelset of~$u$, if~$C_1$ contains a vertex carrying that label, then add~$\ell$ to the labelset of~$u$; otherwise, one of the vertices in~$\{\bound_G(t)\} \cup \bigcup_{j=2}^k C_j$ carries label~$\ell$ and we add~$\ell$ to the labelset of~$v$.
\end{itemize}
This concludes the construction of~$H'$. We have~$H' \in \extend_{+((t-1-t')+1)}(H') = \extend_{+(t-t')}(H')$ since it was obtained from~$\widehat{H}$ by the reverse of an edge contraction; note that contracting the edge~$\{u, v\}$ recovers~$\widehat{H}$. It remains to prove that~$H' \leqlb G$. Towards that end, we build a model~$(\varphi',\psi')$ of~$H'$ in~$G$, as follows.
\begin{itemize}
	\item For~$w \in V(H') \setminus \{u, v\}$, set~$\varphi'(w) := \varphi(w)$.
	\item Set~$\varphi'(u) := C_1$.
	\item Set~$\varphi'(v) := \{x\} \cup \bigcup _{j=2}^k C_j$.
	\item For each edge~$\{w,w'\} \in E(H')$ that is not incident on~$v$, set~$\psi'(\{w,w'\}) := \psi(\{w,w'\})$.
	\item For each edge~$\{v, w\} \in E(H') \setminus \{v,u\}$, set~$\psi'(\{v, w\}) := \psi(\{u, w\})$.
	\item Set~$\psi'(\{u, v\}) := \{x, x'\}$ for some vertex~$x' \in C_1$ that is adjacent in~$G$ to~$x$.
\end{itemize}
To see that a vertex~$x'$ as required in the last step exists, observe that vertex~$x$ is adjacent to at least one vertex of~$C_i$ for all~$i \in [k]$. This follows from the fact that~$G[\varphi(u)]$ is connected by definition, and~$C_1, \ldots, C_k$ are the components that result from that connected graph by removing vertex~$x$. The same argument shows that~$G[\varphi'(v)]$ is connected. We trivially satisfy the requirement that the branch set of~$v$ contains~$x$, and have~$y \in C_1 = \varphi'(u)$. Using these facts, it is easy to verify that~$(\varphi',\psi')$ is a valid labeled boundaried minor model of~$H$ in~$G$, which concludes the proof.
\end{proof}

\subsection{Summing pieces of a boundaried graph}
We will now give some useful properties of boundaried graphs.
Boundaried graphs can be summed together using the following notion.

\begin{definition}[$\oplus$] \label{def:oplus}
Let~$G_1, G_2 \in \lbgraphs$. Then~$G_1 \oplus G_2$ is defined as the $t$-boundaried graph obtained from the disjoint union of~$G_1$ and~$G_2$ by identifying  vertices $u \in V(G_1)$ and $v \in V(G_2)$ whenever $\bound_{G_1}(u) = \bound_{G_2}(v)$. Let the labelset of the new vertex be $L_{G_1}(u) \cup L_{G_2}(v)$. We stress that no parallel edges are introduced in this step. 

For a set~$P = \{p_1, \ldots, p_k\}$ of $t$-boundaried graphs, define~$\bigoplus_{p \in P} p$ as~$p_1 \oplus p_2 \oplus \ldots \oplus p_k$. For~$P = \emptyset$, we define~$\bigoplus_{p \in P} p$ as the empty graph.
\end{definition}
The following notion will be instrumental to analyze how a large $t$-boundaried graph can be formed by gluing together \emph{pieces} of smaller $t$-boundaried graphs. Refer to Figure~\ref{fig:pieces} for an illustration.

\begin{figure}
\centering
\includegraphics{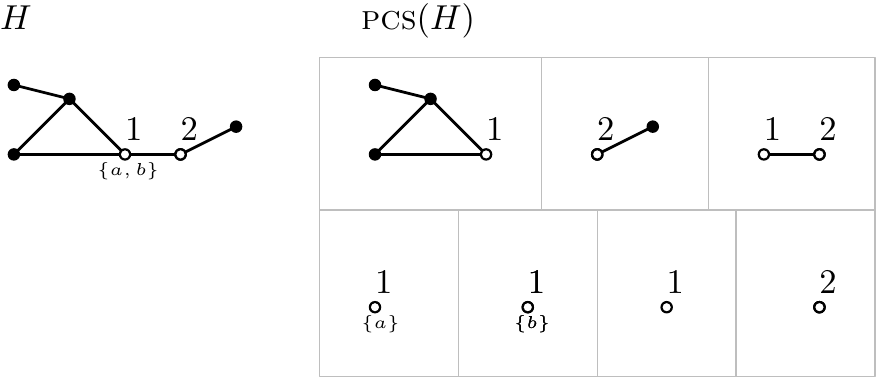}
\caption{This figure shows a $2$-boundaried $\{a,b\}$-labeled graph $H$, together with its set of pieces $\pieces(H)$.}
\label{fig:pieces}
\end{figure}

\begin{definition}[\pieces] \label{def:pcs}
Let $G \in \lbgraphs$. Let~$\pieces(G)$ (for \emph{pieces}) contain the following $t$-boundaried graphs.
\begin{itemize}
\item For all  vertices in $\delta(G)$, $\pieces(G)$ contains a graph $P$ consisting of a single vertex $u$ with $L_P(u) := \emptyset$ and $\bound_P(u):= \bound_G(v)$.
\item For all $v \in \delta(G)$, for all $x \in L_G(v)$, $\pieces(G)$ contains a graph $P$ consisting of a single vertex $u$ with $L_P(u) := \{x\}$ and $\bound_P(u):= \bound_G(v)$.
\item For every edge $\{u,v\} \in E(G)$ with $u,v\in\delta(G)$, $\pieces(G)$ contains a graph $P$ with vertices $x$ and $y$ and edge $\{x,y\}$. Define $\bound_P(x) := \bound_P(u)$, $\bound_P(y):= \bound_P(v)$, and $L_P(u) := L_P(v) = \emptyset$.
\item For every connected component $C$ of $G - \delta(G)$, define $C'$ as the set $C$ together with all vertices in $\delta(G)$ that are adjacent to $C$. Let $\pieces(G)$ contain a graph $P$ that is equal to $G[C']$ after removing all edges between boundary vertices. Remove all labels from the vertices in $\delta(P)$.
\end{itemize}
For unlabeled graphs,~$\pieces(G)$ is defined analogously by treating it as a $\emptyset$-labeled graph.
\end{definition}
From these definitions, it follows that~$\bigoplus_{p \in \pieces(G)} p = G$ for all~$G \in \lbgraphs$. The following lemma shows the key property of this definition of \pieces.

\begin{lemma}\label{lem:split_in_pieces}
Let $H$, $G_1$ and $G_2$ be $t$-boundaried $X$-labeled graphs.
\begin{align*}
H \leqlb G_1 \oplus G_2 \Leftrightarrow \exists P \subseteq &\pieces(H):
  \bigoplus_{p \in P} p \leqlb G_1 \text{ and }
  \bigoplus_{p \in \pieces(H)\setminus P} p \leqlb G_2.
 \end{align*}
\end{lemma}
\begin{proof}
Suppose
$\exists P \subseteq \pieces(H):
  \bigoplus_{p \in P} p \leqlb G_1 \text{ and}
  \bigoplus_{p \in \pieces(H)\setminus P} p \leqlb G_2$. Let $H_1 := \bigoplus_{p \in P} p$ and let $H_2 := \bigoplus_{p \notin P} p$.

  We start by showing $H_1 \oplus H_2 \leqlb G_1 \oplus G_2$ by defining a minor model $\varphi$ of $H_1 \oplus H_2$ in $G_1 \oplus G_2$. Since $H = H_1 \oplus H_2$ it will follow that $H \leqlb G_1 \oplus G_2$.
  Let $\varphi_1$ be a minor model of $H_1$ in $G_1$ and let $\varphi_2$ be a minor model of $H_2$ in $G_2$. For any non-boundary vertex $v \in V(H_1 \oplus H_2)$ that was originally in $H_x$ for $x\in \{1,2\}$, it is easy to see that we may define $\varphi(v) := \varphi_x(v)$. Similarly, for any boundary vertex $v \in \delta(H)$ that only occurs in $H_x$ for $x \in \{1,2\}$, we define $\varphi(v) := \varphi_x(v)$.
  For any other boundary vertex $v \in \delta(H)$, let $v_1 \in \delta(H_1)$ and $v_2 \in \delta(H_2)$ such that $\bound_{H}(v) = \bound_{H_1}(v_1) = \bound_{H_2}(v_2)$.
  Define   $\varphi(v) : = \varphi_1(v_1) \cup \varphi_2(v_2)$. Verify that this branch set remains connected, since branch set $\varphi(v_x)$  for $x\in[2]$ contains the vertex in $G_x$ with boundary label $b_{H_1}(v_1)=b_{H_2}(v_2)$ by definition, and in $G$ these vertices were identified.

Suppose $H \leqlb G_1\oplus G_2$. Find a minor model $\varphi$ with edge mapping $\psi$ of $H$ in $G_1 \oplus G_2$. 
We now show how to define $P$. Let $p \in \pieces(H)$.
We consider four options:
\begin{itemize}
\item $p$ consists of an unlabeled, isolated boundary vertex $v\in\delta(H)$. If $\varphi(v)\cap \delta(G_1)\neq \emptyset$, add $p$ to $G_1$.
\item $p$ consists of one vertex $v \in \delta(H)$ with $\lab_H(v) = \{x\}$ for some $x\in X$. Let $u \in \varphi(v)$ be a vertex with $x \in \lab_G(u)$. If $u \in V(G_1)$, add $p$ to $P$.
\item $p$ consists of boundary vertices $u$ and $v$ connected by edge $\{u,v\}$. Let $\{u',v'\} := \psi(\{u,v\})$. If $u',v' \in V(G_1)$ and $\{u',v'\} \in E(G_1)$, add $p$ to $P$. Else, $\{u',v'\} \in E(G_2)$ and we do not add $p$ to $P$.
\item $p$ contains some non-boundary vertex  $v \notin \delta(H)$. Note that $\varphi(v)$ is either completely contained in $G_1$ or in $G_2$. If $\varphi(v) \subseteq V(G_1)$, add $p$ to $P$. Else, do nothing.
\end{itemize}
Let $H_1 := \bigoplus_{p \in P} p$ and let $H_2 := \bigoplus_{p \notin P} p$, we show how to find boundaried minor models of $H_i$ in $G_i$ for $i \in \{1,2\}$. We only show that $H_1 \leqb G_1$, the proof that $H_2 \leqb G_2$ is symmetric. Let $\varphi_1$ be defined as $\varphi_1(v) = \varphi(v) \cap V(G_1)$ for $v \in V(H_1)$. We show that $\varphi_1$ is a valid boundaried minor model of $H_1$.
\begin{itemize}
\item Connected -- Let $v \in V(H_1)$. If $v$ is not a boundary vertex, $\varphi_1(v) =\varphi(v)$ and all properties follow. Else, if $\varphi_1(v)$ is not connected, there are at least two parts in $G_1$, that were connected in $G$. But this is not possible, since then $\varphi(v)$ contained more than one boundary vertex.
\item Non-intersecting -- This follows immediately from the definition.
\item Edges -- Let $\{u,v\} \in E(H_1)$. If $u,v \notin \delta(H)$,  the result follows. Suppose $u,v \in \delta(H)$, then this edge exists because the part $p$ consisting of the two boundary vertices $u$ and $v$ connected by an edge (namely, $\{u,v\}$) was added to $P$. Thereby, $\psi(\{u,v\}) = \{u',v'\}$ with $\{u',v'\} \in E(G_1)$ and by $u',v' \in V(G_1)$ and the definition of $\varphi_1$ it follows that $u' \in \varphi_1(u)$ and $v' \in \varphi_1(v)$ as required.

    Suppose $u \notin \delta(H)$ and $v \in \delta(H)$. Consider $\psi(\{u,v\}) = \{u',v'\}$. Clearly, $\{u',v'\}$ is an edge in $G$. It remains to show that $u',v' \in V(G_1)$, then it follows by definition that $u'\in\varphi_1(u)$ and $v' \in \varphi_1(v)$. Clearly, $u' \in V(G_1)$ as $\varphi_1(u) \subseteq V(G_1)$ since $u$ is not a boundary vertex. Furthermore, $v' \in V(G_1)$ as $\{u',v'\}$ is an edge, and no edges go from $V(G_1)$ to $G \setminus V(G_1)$. Thereby, $u'$ and $v'$ are in  $V(G_1)$ as required.
\item Boundary -- Let $u \in \delta(H_1)$. Let $u' \in G$ be the vertex with the same boundary label, by definition $u' \in \varphi(u)$. Since $u \in \delta(H_1)$, it follows from the choice of $H_1$ that $u' \in V(G_1)$, and thereby $u' \in \varphi_1(u)$. The other properties follow immediately from $\varphi_1(u) \subseteq \varphi(u)$ for all $u \in V(H_1)$.
\item Labels -- Let $u \in V(H)$ such that $x \in \lab_H(u)$. If $u \notin \delta(H_1)$ the result follows. Suppose $u \in \delta(H_1)$. By point three of the construction, there exists $u' \in \varphi(u)$ such that $\lab_G(u) = \{x\}$ and $u' \in V(G_1)$, thereby $u' \in \varphi_1(u)$.\qedhere
\end{itemize}
\end{proof}

In combination with Lemma \ref{lem:split_in_pieces}, the above lemma shows that if we are given graphs $H \in \lgraphs$ and $G \in \lbgraphs$ such that $G = G_1 \oplus G_2$ with $H \leql \forget(G)$, then there exists $H' \in \extend_{+t}(G)$ such that $H' \leqlb G$. Furthermore, there exists $P \subseteq \pieces(H')$ such that $\bigoplus_{p\in P} \leqlb G_1$ and $\bigoplus_{p \notin P } p\leqlb G_2$. An example is shown in Figure \ref{fig:extend_needed_2}, where we see that the extend of $H$ is really necessary.

\begin{figure}
\centering
\includegraphics{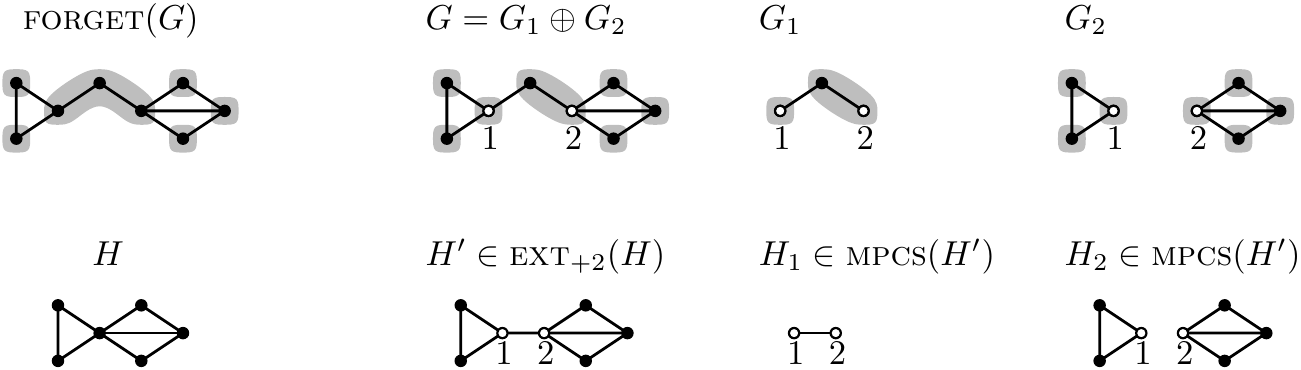}
\caption{
This figure illustrates how the \extend and \multipieces operations combine. It shows $G = G_1 \oplus G_2$ and $H$, the choice of $H' \in \extend_{+2}(H)$ such that $H' \leqlb G$, and the choice of $H_1, H_2 \in \multipieces_{+2}(H)$ such that $H_1 \leqlb G_1$ and $H_2 \leqlb G_2$. Minor models are shown for all the claimed minor relations. 
}
\label{fig:extend_needed_2}
\end{figure}

\begin{definition}[\multipieces] \label{def:mpcs}
Let $G \in \lbgraphs$. For this definition, let two graphs be equal if they are isomorphic, as defined in Definition \ref{def:isomorphism}. Define the \emph{multi pieces}, abbreviated as \multipieces as
\[\multipieces(G) := \{\bigoplus_{p \in P} p  \mid  P \subseteq \pieces(G) \wedge P \neq \emptyset\}.\]
For a set of graphs $\Q \subseteq \lbgraphs$ define
\[\multipieces(\Q) := \bigcup_{Q \in \Q} \multipieces(Q).\]
Let $\multipieces_{+t}(\Q) := \multipieces(\extend_{+t}(\Q))$.
\end{definition}
By these definitions, the set $\pieces(\Q)$ may contain graphs that are isomorphic to each other, but the set $\multipieces(\Q)$ can not.

The following observation follows immediately from the definition of \multipieces.

\begin{observation} \label{obs:alternative-definition-mpcs}
Let $G,H \in \lbgraphs$. Then $H \in \multipieces(G)$ if and only if $H$ can be obtained from $G$ by a sequence of the following operations.
\begin{enumerate}
\item \label{mpcs:remove-component} Remove a connected component of $G - \delta(G)$ with its edges to the boundary.
\item \label{mpcs:remove-edge} Remove an edge between two boundary vertices.
\item \label{mpcs:remove-label} Remove a label from the labelset of a boundary vertex.
\item \label{mpcs:remove-boundary-vertex} Remove an isolated boundary vertex.
\end{enumerate}
\end{observation}

The following lemma says that if a $(t+1)$-boundaried graph can be obtained by first extending a (ordinary) graph in~$\F$ for~$t$ times, then restricting the result to the sum of a subset of its pieces, and extending that sum once more to obtain~$H'$, then one can also obtain~$H'$ by first extending~$(t+1)$-times and then restricting to a subset of the pieces. Note that it applies to both labeled and unlabeled graphs, by using an empty label set for $X$.

\begin{lemma} \label{lem:mpcs:plusone}
Let~$\Q \subseteq \lgraphs$. If~$H \in \multipieces_{+t}(\Q)$ and~$H' \in \extend_{+1}(H)$, then~$H' \in \multipieces_{+(t+1)}(\Q)$.
\end{lemma}
\begin{proof}
By Definition~\ref{def:mpcs} we have~$H \in \multipieces(Q)$ for some~$Q \in \extend_{+t}(\Q)$. Consider the operation that extends~$H$ to~$H' \in \extend_{+1}(H)$. We will mimic this operation in~$Q$ to find some~$Q' \in \extend_{+(t+1)}(\Q)$ that forms a supergraph of~$H'$ respecting the boundary and labelsets; then we will argue that~$H'$ can be obtained as the sum of a subset of pieces of~$Q'$ and therefore~$H' \in \multipieces(Q') \subseteq \multipieces_{+(t+1)}(\Q)$.
\begin{itemize}
	\item If~$H'$ was simply a copy of $H$, then obtain let~$Q'$ equal~$Q$.
	\item If~$H'$ was obtained by setting~$\bound_{H'}(v) = t+1$ for some~$v \in V(H)$, then since~$H$ is the sum of some pieces of~$Q$ we have~$v \in V(Q)$. Obtain~$Q'$ from~$Q$ by setting~$\bound_{Q'}(v) = t+1$.
	\item Otherwise,~$H'$ was obtained by the reverse of an edge contraction operation on some boundary vertex~$\bound_{H}^{-1}(i)$ by adding a new boundary vertex~$v$ with $\bound_{H'}(v)=t+1$, adding the edge~$\{\bound_{H'}^{-1}(i), \bound_{H'}^{-1}(t+1)\}$, moving some of the labels~$L$ on~$\bound^{-1}_{H'}(i)$ to~$\bound_{H'}^{-1}(t+1)$, and finally changing some edges that were incident on~$\bound_{H'}^{-1}(i)$ to now be incident on~$\bound_{H'}^{-1}(t+1)$ instead; let~$T$ be the set of vertices that are adjacent to~$\bound_{H}^{-1}(i)$ in~$H$ but not adjacent to~$\bound_{H'}^{-1}(i)$ in~$H'$.
	
	We can apply the same operations on~$Q$. We initialize~$Q'$ as~$Q$, insert a new vertex~$v$ with $\bound_{Q'}(v) = t+1$, add the edge~$\{\bound_{Q'}^{-1}(i), \bound_{Q'}^{-1}(t+1)\}$, move the labels~$L$ from~$\bound_{Q'}^{-1}(i)$ to~$\bound_{Q'}^{-1}(t+1)$, remove the edges from~$T$ to~$\bound_{Q'}^{-1}(i)$ and make those vertices adjacent to~$\bound_{Q'}^{-1}(t+1)$ instead.
\end{itemize}
In all cases, we find~$Q' \in \extend_{+1}(Q)$ such that~$Q'$ is a supergraph of~$H'$ with the same boundary function and the same labelsets. 
By Observation \ref{obs:alternative-definition-mpcs} and $H \in\multipieces(Q)$, it follows $H$ can be obtained from $Q$ by a sequence of the described removal operations.
We show that~$H'$ can be obtained from~$Q'$ by a similar sequence of operations, implying $H'\in\multipieces(Q')$ by Observation \ref{obs:alternative-definition-mpcs}. This is easy to see for operations~\ref{mpcs:remove-edge}, \ref{mpcs:remove-label}, and~\ref{mpcs:remove-boundary-vertex}. If~$H$ was obtained from~$Q$ by (amongst others) removing the vertices~$C$ of a connected component of~$Q - \delta(Q)$, then no vertex of~$C$ was chosen as~$\bound_{H'}^{-1}(t+1)$ and we claim that~$C$ is also the vertex set of a connected component of~$Q' - \delta(Q')$. Observe that~$H$ being a component of~$Q - \delta(Q)$ implies that all neighbors of~$C$ in~$Q$ belong to the boundary~$\delta(Q)$. Therefore, the operation that took~$Q$ to~$Q'$ was splitting a boundary vertex (causing all neighbors of~$C$ to lie in the new boundary),  or turning a vertex not in~$C$ into a boundary vertex (which again does not affect the neighborhood of~$C$). Hence the neighborhood of~$C$ in~$Q'$ is a subset of the boundary, and all such~$C$ are connected components of~$Q' - \delta(Q')$. They therefore correspond to pieces of~$Q'$ that can be omitted from the sum of pieces if desired. It follows that~$H' \in \multipieces(Q')$, which concludes the proof.
\end{proof}

\begin{lemma}
\label{item:remove_boundary_vertex_reduce_mpcs} Let $G \in \lgraphs$ and $H \in \lbgraphs$. Then $H \in \multipieces_{+t+1}(G) \Leftrightarrow H \in \multipieces_{+t}(G)$.
\end{lemma}
\begin{proof}
Suppose $H \in \multipieces_{+t}(G)$. Then there exists $G'$ such that $G' \in \extend_{+t}(G)$ and $H \in \multipieces(G')$. If is easy to see that $G' \in \extend_{+t+1}(G)$, and thereby $H \in \multipieces_{+t+1}(G)$.

\begin{figure}
\centering
\includegraphics{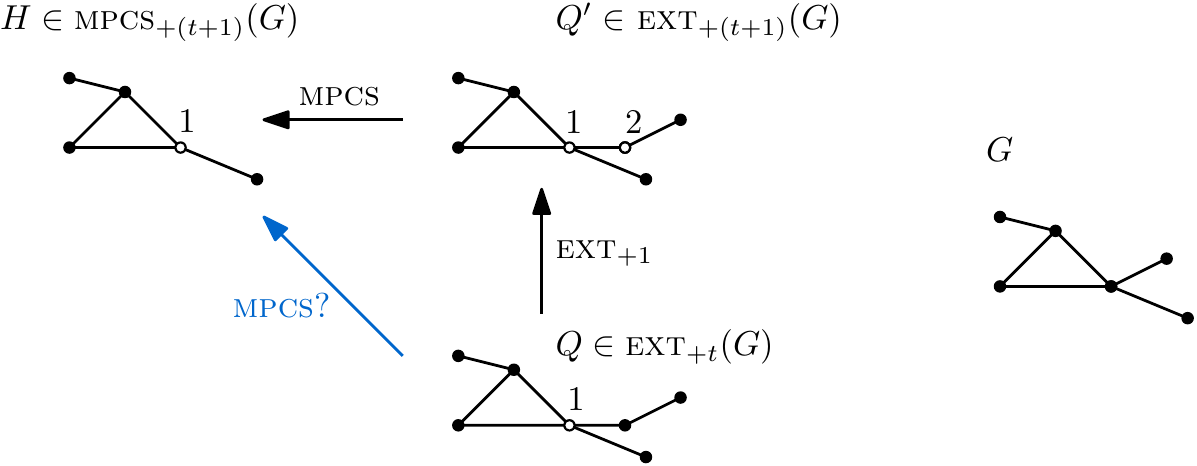}
\caption{A figure showing an example of $G$  and $H\in \multipieces_{+t+1}(G)$ for Lemma \ref{item:remove_boundary_vertex_reduce_mpcs} and the graphs $Q$ and $Q'$ used in the proof for $t=1$, together with the relations between them.}
\label{fig:mpcs-properties}
\end{figure}

For the opposite direction, suppose $H \in \multipieces_{+t+1}(G)$. By definition, we can find $Q' \in \extend_{+t+1}(G)$ such that $H \in \multipieces(Q')$. Now let $Q$ such that $Q \in \extend_{+t}(G)$ and $Q' \in \extend_{+1}(Q)$. See Figure \ref{fig:mpcs-properties} for an illustration. We will show that $H \in \multipieces(Q)$, to conclude the proof.  To show this, we will use the alternative definition for \multipieces, as given in Observation \ref{obs:alternative-definition-mpcs}.  We do a case distinction on the extend operation applied to obtain $Q'$ from $Q$.
\begin{itemize}
\item $Q' = Q$. Since $H \in \multipieces(Q')$, the result follows.
\item $Q'$ was obtained from $Q$ by setting $\bound_{Q'}(v) := t+1$ for an existing vertex $v$.  To obtain $H$ from $Q$, start by removing the connected component of $Q -\delta(Q)$ containing $v$ and its edges to the boundary. Since $H$ is a $t$-boundaried graph, $v$ is not a vertex of $H$ and all components of $Q' - \delta(Q')$ that contain a neighbor of $v$, were indeed removed from $H'$ in the process of obtaining $Q'$. Furthermore, $v$ itself and edges to other boundary vertices were indeed removed.

    If a connected component $S$ of $Q'-\delta(Q')$ was removed to obtain $H'$ that does not have connections to $v$, it is also a component of $Q - \delta(Q)$ and we remove it.

    Furthermore we forget any edges between boundary vertices, labels within the boundary, and (now isolated) boundary vertices as needed.
\item $Q'$ was obtained from $Q$ by splitting vertex $u\in \delta(Q)$ into $u$ and $v$ (such that $b_{Q'}(v) = t+1$) and redistributing labels and edges. Consider the steps used to obtain $H$ from $Q'$. Suppose component $S$ of $Q' - \delta(Q')$ was removed, then we also remove this component from $Q$ to obtain $H$. Note that $S$ is a component in $Q - \delta(Q)$, because $N_{Q'}(S) \subseteq \delta(Q')$ and all edges connecting to $v$ in $Q'$ connect to $u$ in $Q$, thus $N_Q(S) \subseteq \delta(Q)$.
    Note that the only components of $Q'-\delta(Q')$ that differ from those in $Q-\delta(Q)$ are the ones containing a neighbor of $v$ and these are always removed because $v$ is removed to obtain $H$, thus $v$ must then be isolated in $Q'$.

    If a boundary edge not incident on $v$ is removed, the same edge can be removed in $Q$. If a boundary edge $\{v,w\}$ is removed, then either $w=u$ and we do nothing, or $v\neq u$ and we remove edge $\{u,w\}$ from $Q$.

    Similarly, if a label is removed from a boundary vertex other than $v$, it can simply be removed from this vertex in $Q$ as well. If a label is removed from the labelset of $v$ in $Q'$, remove it from the labelset of $u$ in $Q$.

    If isolated boundary vertices (other than $v$) are removed, we can also removed them from $Q$ to obtain $H$.\qedhere
\end{itemize}
\end{proof}

The next lemma shows that when starting from a set of connected graphs~$\Q$, any multi piece that is obtained by extending a graph~$Q \in \Q$ is either attached, or consists of a real extension of~$Q$ (rather than a subset of its pieces).

\begin{lemma} \label{lem:mpcs:plust:attached:or:extend}
Let~$\Q \subseteq \lcongraphs$ and~$H \in \multipieces_{+t}(\Q)$. Then~$H \in \attlbgraphs$ or~$H \in \extend_{+t}(\Q)$.
\end{lemma}
\begin{proof}
Since~$H \in \multipieces_{+t}(\Q)$, there is a graph~$Q \in \Q$ and a graph~$Q' \in \extend_{+t}(Q)$, such that~$H$ is obtained as~$\bigoplus_{p \in P} p$ for some subset~$P \subseteq \pieces(Q')$. Observe that if all connected components of~$Q'$ contain a boundary vertex, then the same is true for all pieces of~$Q'$ and hence for each multi piece of~$Q'$. So if~$Q' \in \attlbgraphs$, then~$H \in \attlbgraphs$ and we are done.


Suppose $Q' \notin \attlbgraphs$, it follows that $Q'$ was obtained from $Q$ by simply copying $Q$ in every step of \extend. Thereby, $Q' =Q$, $\delta(Q') = \emptyset$ and it is easy to see that $\multipieces(Q') = \{Q'\} =\{Q\}$. Thus, $H = Q$ and thereby $H \in \extend_{+t}(\Q)$.
\end{proof}

\subsection{Splitting and merging families of boundaried graphs}
The definitions in this section will be useful when analyzing how the task of breaking the models of a certain family~$\Pi$ of boundaried minors in a graph~$G_1 \oplus G_2$, can be divided into breaking certain minor models~$\Pi_1$ in~$G_1$ and breaking~$\Pi_2$ in~$G_2$.

\begin{definition}[\splitPi]
Let $\Pi \subseteq \bgraphs$. Define
$\splitPi_{\mathcal{F}}(\Pi)$ as the set of all pairs $(\Pi_1,\Pi_2)$ with~$\Pi_1, \Pi_2 \subseteq \bgraphs$ that can be obtained by the following procedure. Initialize $\Pi_1 := \Pi$ and $\Pi_2 := \Pi$. Continue as follows.
\begin{itemize}
\item For each $G \in \Pi$, consider each set $Y \subseteq \pieces(G)$, let $Y' := \pieces(G)\setminus Y$. Add $\bigoplus_{y \in Y} y$ to $\Pi_1$ or add $\bigoplus_{y \in Y'} y$ to $\Pi_2$ (or both).
\item For each $i \in \{1,2\}$ and $G \in \multipieces_{+t}(\mathcal{F})$, if some minor of $G$ is in $\Pi_i$, then  add $G$ to $\Pi_i$.\qedhere
\end{itemize}
\end{definition}

The subscript of $\splitPi_{\mathcal{F}}(\Pi)$ will be omitted when it is clear from the context.
Observe that by the above definition $\Pi_1 \supseteq \Pi$ and $\Pi_2 \supseteq \Pi$ for any $(\Pi_1,\Pi_2)\in\splitPi(\Pi)$.

\begin{lemma}\label{lem:splitPi}
Let $\Pi$ be a set of $t$-boundaried graphs and let $G_1$ and $G_2$ be $t$-boundaried graphs.  Let $G := G_1\oplus G_2$. Then $G$ does not have any graph in $\Pi$ as a minor if and only if there exist $(\Pi_1,\Pi_2) \in \splitPi(\Pi)$ such that $G_i$ has no graph in $\Pi_i$ as a minor for $i \in [2]$.
\end{lemma}
\begin{proof}
Suppose there exist $(\Pi_1,\Pi_2) \in \splitPi(\Pi)$ such that $G_i$ has no graph in $\Pi_i$ as a minor for $i \in [2]$. Suppose for contradiction that there exists $H \in \Pi$ such that $H \leqb G$. By Lemma \ref{lem:split_in_pieces}, this implies that there exists $P \subseteq \pieces(H)$ such that $\bigoplus_{p \in P} p \leqb G_1$ and $\bigoplus_{p \notin P} p \leqb G_2$. But by the definition of \splitPi, either $\bigoplus_{p \in P} p \in \Pi_1$ or $\bigoplus_{p \notin P} p  \in \Pi_2$, which is a contradiction.

Suppose $G$ does not have any graph in $\Pi$ as a minor.  We show how to define $\Pi_1$ and $\Pi_2$. Initialize $\Pi_i := \Pi$ for $i \in [2]$. Consider every graph $H\in \Pi$, and all $P \subseteq \pieces(H)$. If $\bigoplus_{p \in P} p \not\leqb G_1$, add $\bigoplus_{p \in P} p$ to $\Pi_1$. If $\bigoplus_{p \notin P} p \not\leqb G_2$, add $\bigoplus_{p \notin P} p$ to $\Pi_2$. Finally, for any $H' \in \multipieces_{+t}(\mathcal{F})$ such that $\Pi_i$ contains a minor of $H'$, add $H'$ to $\Pi_i$ for $i \in [2]$.
%
%
It remains to show that $(\Pi_1,\Pi_2) \in \splitPi(\Pi)$, since  by definition, $G_1$ has no minors in $\Pi_1$ and $G_2$ has no minors in $\Pi_2$. We need to verify that for all $H \in \Pi$ and all $P \subseteq \pieces(H)$, either $\bigoplus_{p \in P} p \in \Pi_1$ or $\bigoplus_{p \notin P} p \in \Pi_2$. Suppose for contradiction that there exists such $P$ with
$\bigoplus_{p \in P} p \notin \Pi_1$ and $\bigoplus_{p \notin P} p \notin \Pi_2$. By the construction of $\Pi_1$ and $\Pi_2$, this implies that
$\bigoplus_{p \in P} p \leqb G_1$ and $\bigoplus_{p \notin P} p \leqb G_2$. By Lemma \ref{lem:split_in_pieces}, it follows that $\bigoplus_{p \in \pieces(H)} p \leqb G_1 \oplus G_2 = G$. Since  $\bigoplus_{p \in \pieces(H)}p = H$, it follows that $H \leqb G$. Since $H \in \Pi$, this contradicts that $G$ has no graph in $\Pi$ as a minor.
\end{proof}

The following operation acts as an inverse to~$\splitPi$.

\begin{definition}[$\odot$] \label{def:merge}
Let $\Pi_1, \Pi_2 \subseteq \bgraphs$. Let $\F \subseteq \graphs$, and define:
\begin{align*}
\Pi_1 \odot_\F \Pi_2 := \{G & \in \multipieces_{+t}(\F) \mid
\forall G_1, G_2 \in \bgraphs \colon G_1 \oplus G_2 = G \wedge E(G_1) \cap E(G_2) = \emptyset \Rightarrow \\
& \Pi_1 \text{ contains a minor of $G_1$ or } \Pi_2 \text{ contains a minor of $G_2$}
\}.
\end{align*}
We omit the subscript from~$\odot_\F$ when it is clear from the context.
\end{definition}

\begin{lemma}\label{lem:merge-vs-split}
Let $\Pi,\Pi_1,\Pi_2 \subseteq \multipieces_{+t}(\F)$ be sets of $t$-boundaried graphs, such that $(\Pi_1,\Pi_2) \in \splitPi(\Pi)$. Then $\Pi_1 \odot \Pi_2 \supseteq \Pi$.
\end{lemma}
\begin{proof}
Let $H \in \Pi$, suppose for contradiction that $H \notin \Pi_1\odot \Pi_2$. Find $H_1,H_2$ such that $\Pi_1$ contains no minor of $H_1$ and $\Pi_2$ contains no minor of $H_2$, and furthermore $H_1 \oplus H_2 = H$ and the two graphs share no edges. These must exist by the definition of $\odot$. By definition, $H \leqb H_1 \oplus H_2$ and by Lemma \ref{lem:split_in_pieces} there exists $P \subseteq \pieces(H)$ such that $H_1' := \bigoplus_{p \in P}p\leqb H_1$ and $H_2' := \bigoplus_{p \in \pieces(H)\setminus P}p\leqb H_2$ and $H_1' \oplus H_2' = H$. But since $(\Pi_1,\Pi_2) \in \splitPi(\Pi)$, this implies either $H_1'\in \Pi_1$ or $H_2' \in \Pi_2$ and thus $\Pi_1$ contains a minor of $H_1$ (namely $H_1'$) or $\Pi_2$ contains a minor of $H_2$, which is a contradiction. Thus, $H \in \Pi_1 \odot \Pi_2$.
\end{proof}

\begin{lemma}\label{lem:mergePi}
Let $G := G_1 \oplus G_2$ be a $t$-boundaried graph and let $\mathcal{F} \subset \graphs$. Let $\Pi_1,\Pi_2 \subseteq \multipieces_{+t}(\F)$. If $G_1$ has no $\Pi_1$-minors and $G_2$ has no $\Pi_2$-minors, then $G$ has no $\Pi_1 \odot \Pi_2$-minors.
\end{lemma}
\begin{proof}
Suppose for contradiction that $G$ has a boundaried $\Pi_1 \odot \Pi_2$-minor $H$. By Lemma \ref{lem:split_in_pieces}, there exist $P \subseteq \pieces(H)$ such that $H_1 := \bigoplus_{p \in P}p\leqlb G_1$ and $H_2 := \bigoplus_{p \in \pieces(H)\setminus P}p\leqlb G_2$. But then no minor of $H_1$ can be in $\Pi_1$, no minor of $H_2$ is in $\Pi_2$ and $H_1$ and $H_2$ share no edges, and $H_1 \oplus H_2 = H$. But this contradicts the fact that $H \in \Pi_1\odot \Pi_2$. Thereby, $G$ has no $\Pi_1 \odot \Pi_2$-minors.
\end{proof}

\begin{lemma}\label{lem:merge_associative}
Let $\piaone, \piatwo,\pib \subseteq \multipieces_{+t}(\F)$ be three sets of boundaried graphs, then
\[(\piaone \odot \piatwo) \odot \pib = \piaone \odot (\piatwo \odot \pib).\]
\end{lemma}
\begin{proof}
It is easy to see that
\begin{align*}
(\piaone \odot &\piatwo) \odot \pib = \\
&\{G \in \multipieces_{+t}(\F) \mid \\
&\quad \forall G_1,G_2,G_3 \in \bgraphs \text{ that share no edges, such that }  G_1 \oplus G_2 \oplus G_3 = G  : \\
&\qquad\text{$\exists i \in [3]$ such that $\Pi_i$ contains a minor of $G_i$}\} =\\
\piaone \odot &(\piatwo \odot \pib)\qedhere
\end{align*}
\end{proof}

It is easy to see from the definition that $\odot$ is commutative, so that it is both associative and commutative. Hence,  reordering or parenthesizing an expression of the form $\Pi_1\odot\Pi_2\odot\ldots\odot \Pi_n$, will not change the result.

\begin{lemma}\label{lem:merge:large}
Let~$\Pi_1, \Pi_2, \Pi_3 \subseteq \multipieces_{+t}(\F)$ and~$H \in \multipieces_{+t}(\F)$, for some~$\F \subseteq \graphs$. Then~$H \in \Pi_1 \odot \Pi_2 \odot \Pi_3$ if and only if for all partitions of~$\pieces(H)$ into~$P_1, P_2, P_3$, there exists~$i \in [3]$ such that~$\Pi_i$ contains a minor of~$\bigoplus_{p \in P_i} p$.
\end{lemma}
\begin{proof}
This follows from Lemma~\ref{lem:split_in_pieces} and Definition~\ref{def:merge}.
\end{proof}

The following technical statement is used in the inductive proof of the main lemma. There, we will be considering \FDeletion solutions in a graph~$G$ that is decomposed as~$G = \forget(G_A \oplus G_B \oplus G_C)$ for $t$-boundaried graphs~$G_A, G_B, G_C$. If~$G_A - \delta(G_A)$ is connected, we progress in the induction by selecting a vertex of~$G_A - \delta(G_A)$ whose removal decreases the treedepth of~$G_A - \delta(G_A)$ and using it as the $t+1$'th boundary vertex. This turns~$G_A$ into a $(t+1)$-boundaried graph~$G'_A$, and we interpret~$G_B$ and~$G_C$ also as~$(t+1)$-boundaried graphs~$G'_B,G'_C$ in which the $(t+1)$'th boundary vertex is undefined. The next lemma shows that if breaking the $t$-boundaried graphs~$\Pi_1,\Pi_2,\Pi_3$ in~$G_A,G_B,G_C$ respectively was sufficient to obtain an \FDeletion solution in~$\forget(G_A \oplus G_B \oplus G_C)$, then breaking the $(t+1)$-boundaried graphs~$\Pi'_1,\Pi'_2,\Pi'_3$ in~$G'_A,G'_B,G'_C$ yields an \FDeletion solution in~$\forget(G'_A \oplus G'_B \oplus G'_C)$.

\begin{lemma}\label{lem:prohibitions:extendplusone}
Let $\Pi_1,\Pi_2,\Pi_3 \subseteq \multipieces_{+t}(\F)$ for some~$\F \subseteq \graphs$, let $\Pi_1' := \extend_{+1}(\Pi_1)$, and let $\Pi_i' := \extend_{+1}(\Pi_i) \cup (\nonisobgraphs{t+1} \cap \multipieces_{+(t+1)}(\F))$ for $i \in \{2,3\}$. If $\Pi_1\odot\Pi_2 \odot \Pi_3 \supseteq \extend_{+t}(\F)$, then we have~$\Pi_1'\odot\Pi_2' \odot \Pi_3' \supseteq \extend_{+(t+1)}(\F)$.
\end{lemma}
\begin{proof}
Assume~$\Pi_1\odot\Pi_2 \odot \Pi_3 \supseteq \extend_{+t}(\F)$. Consider some graph~$F' \in \extend_{+(t+1)}(\F)$. By Lemma~\ref{lem:merge:large}, to show that~$F' \in \Pi_1'\odot\Pi_2' \odot \Pi_3'$ it suffices to prove that for all partitions~$P'_1, P'_2, P'_3$ of~$\pieces(F')$, there exists~$i \in [3]$ such that~$\Pi'_i$ contains a minor of~$\bigoplus_{p \in P'_i} p$. So consider an arbitrary partition of~$\pieces(F')$. Define~$F'_i := \bigoplus_{p \in P'_i} p$ for all~$i \in [3]$; by Definition~\ref{def:mpcs} it follows that~$F'_i \in \multipieces_{+(t+1)}(\F)$ if~$P'_i \neq \emptyset$. Since the~$P'_i$ partition the pieces, we have~$F'_1 \oplus F'_2 \oplus F'_3 = F'$ and these three graphs are edge-disjoint. We consider the status of the $(t+1)$'th boundary vertex in these graphs.

If~$F'_j \in \nonisobgraphs{t+1}$ for some~$j \in \{2,3\}$, then we have~$F'_j \in \multipieces_{+(t+1)}(\F) \cap \nonisobgraphs{t+1}$ and therefore~$F'_j \in \Pi'_j$, by assumption. Hence~$\Pi'_j$ contains a minor of~$F'_j$ and we are done. In the remainder of the proof, we consider the case that the $(t+1)$'th boundary vertex is left undefined in both~$F'_2$ and~$F'_3$.

Since~$F' \in \extend_{+(t+1)}(\F)$, there is some graph~$F \in \extend_{+t}(\F)$ such that~$F' \in \extend_{+1}(F)$. Following Definition~\ref{def:extend}, we consider three cases of how~$F$ was extended to form~$F'$. In the first case, the proof is trivial. In the second and third case, we will show that there exists some~$i \in [3]$ and $t$-boundaried graph~$F_i$ such that~$F'_i \in \extend_{+1}(F_i)$ and such that~$\Pi_i$ contains a $t$-boundaried minor~$H_i$ of~$F_i$. This will later be shown to prove the lemma.

(\textbf{No change:} $\mathbf{F' := F}$) Suppose~$F'$ equals $F$ because the $(t+1)$'th boundary vertex is left undefined in~$F'$. It follows that $\pieces(F') = \pieces(F)$ and thereby $P_1',P_2',P_3'$ form a partition of the pieces of $F$. It follows from Lemma \ref{lem:merge:large}, that there exists an $i \in [3]$ such that $\Pi_i$ contains a minor $\pi$ of $F_i'$. By definition of extend, $\pi \in \Pi_i'$, which concludes the proof.

(\textbf{Boundary status for existing vertex}) Next, consider the case that~$F'$ was obtained from~$F$ by assigning~$\bound_{F'}(v) := t+1$ for some~$v \in V(F) \setminus \delta(F)$. Recall that~$F'_2, F'_3$ do not define a $(t+1)$'th boundary vertex. Since~$F'_1 \oplus F'_2 \oplus F'_3 = F'$, this means~$F'_1$ contains a vertex~$v_1$ that acts as the $t+1$'th boundary vertex. Define~$F_i := \forget(F'_i, t)$, i.e., by forgetting the boundary status of the $(t+1)$'th boundary vertex. It follows that~$F = F_1 \oplus F_2 \oplus F_3$, since boundary vertex~$t+1$ is not defined in~$F'_j$ for~$j \in \{2,3\}$. As these operations invert the extend operation, we have~$F'_i \in \extend_{+1}(F_i)$ for~$i \in [3]$. Since~$F$ is an unlabeled graph and~$F$ is the sum of the edge-disjoint graphs~$F_1 \oplus F_2 \oplus F_3$, there exists a partition~$P_1, P_2, P_3$ of~$\pieces(F)$ such that~$F_i = \bigoplus_{p \in P_i} p$ for all~$i \in [3]$. As~$F \in \extend_{+t}(\F)$, we have~$F \in \Pi_1\odot\Pi_2 \odot \Pi_3$ by assumption. By Lemma~\ref{lem:merge:large}, these last two facts imply that there exists~$i \in [3]$ such that~$\Pi_i$ contains a $t$-boundaried minor~$H_i$ of~$F_i$.

(\textbf{Splitting a boundary vertex}) Finally, we consider the case that~$F'$ was obtained from~$F$ by splitting the  boundary vertex $\bound_{F}^{-1}(i)$ for some $i \in [t]$, by the reverse of an edge contraction. This resulted in the edge~$\{\bound_{F'}^{-1}(i), \bound_{F'}^{-1}(t+1)\}$ in~$F'$, and contracting this edge recovers~$F$. Since~$\{\bound_{F'}(i),\bound_{F'}(t+1)\}$ is an edge of~$F' = F'_1 \oplus F'_2 \oplus F'_3$, while the $(t+1)$'th boundary vertex is undefined in~$F'_2$ and~$F'_3$, this edge is contained in~$F'_1$. Obtain~$F_1$ from~$\forget(F'_1,t)$ by contracting this edge. Let~$F_j$ equal~$F'_j$ for~$j \in \{2,3\}$. By construction,~$F'_i \in \extend(F_i)$ for each~$i \in [3]$. Similarly as before we find~$F = F_1 \oplus F_2 \oplus F_3$ and there exists a partition~$P_1, P_2, P_3$ of~$\pieces(F)$ such that~$F_i = \bigoplus_{p \in P_i} p$ for all~$i \in [3]$. Just as in the previous case, this implies that there exists~$i \in [3]$ such that~$\Pi_i$ contains a $t$-boundaried minor~$H_i$ of~$F_i$.

(\textbf{Concluding the proof}) In the last two cases, we have found some~$i \in [3]$ and~$F_i$ such that~$F'_i \in \extend_{+1}(F_i)$ and~$H_i \leqb F_i$ for some~$H_i \in \Pi_i$. By Definition~\ref{def:extend} and the fact that~$F'_i \in \extend_{+1}(F_i)$, we have~$F_i \leqb \forget(F'_i, t)$. Hence~$H_i \leqb \forget(F'_i, t)$, which implies by Lemma~\ref{lem:extend:minor} applied for~$X=\emptyset$ that there exists~$H'_i \in \extend_{+1}(H)$ such that~$H'_i \leqb F'_i$. We have~$H'_i \in \extend_{+1}(H_i) \subseteq \extend_{+1}(\Pi_i) \subseteq \Pi'_i$, and therefore~$\Pi'_i$ contains a minor of~$F'_i$, as desired. This concludes the proof.
\end{proof}

\subsection{Starred folio}

Often, we will be interested in whether graphs of a certain kind appear in the folio. We therefore introduce the following abbreviation.

\begin{definition}[$\folio^*$]
For a set~$\Q \subseteq \lcongraphs$, an integer~$t$, and~$G \in \lbgraphs$, define \[\foliotstar(G) := \folio(G) \cap \multipieces_{+t}(\Q).\] Let $S$ be a set of graphs, define $\foliotstar(S) := \bigcup_{G \in S}\foliotstar(G)$.
\end{definition}



Intuitively,~$\folio^*_{\Q,t}(G)$ for a $t$-boundaried $X$-labeled graph~$G$ gives the set of fragments of $\Q$-minors that can be formed in~$G$, so that different  \FDeletion solutions leaving behind the same folio are interchangeable. The following lemma states that if we have three $t$-boundaried labeled graphs~$G_\alpha, G_{\beta_1}, G_{\beta_2}$ such that the set of $\Q$-minor fragments that can be formed in~$G_{\beta_1}$ is a subset of the fragments that can be formed in~$G_{\beta_2}$, then when we glue on a third graph~$G_{\alpha}$, the set of fragments that can be made in~$G_{\alpha} \oplus G_{\beta_1}$ is a subset of those available in~$G_{\alpha} \oplus G_{\beta_2}$.

\begin{lemma} \label{lem:replacehalf:preservefolio}
Let~$\Q \subseteq \lcongraphs$ and let~$G_\alpha, G_{\beta_1}, G_{\beta_2} \in \lbgraphs$ such that~$\folio^*_{\Q,t}(G_{\beta_1}) \subseteq \folio^*_{\Q,t}(G_{\beta_2})$. Then:
\[\folio^*_{\Q,t}(G_\alpha \oplus G_{\beta_1}) \subseteq \folio^*_{\Q,t}(G_\alpha \oplus G_{\beta_2}).\]
Moreover, if~$\folio^*_{\Q,t}(G_{\beta_1}) = \folio^*_{\Q,t}(G_{\beta_2})$ then~$\folio^*_{\Q,t}(G_\alpha \oplus G_{\beta_1}) = \folio^*_{\Q,t}(G_\alpha \oplus G_{\beta_2})$.
\end{lemma}

\begin{proof}
We first prove the $\subseteq$-part of the statement, from which we will later derive the statement about equality. So assume that~$\folio^*_{\Q,t}(G_{\beta_1}) \subseteq \folio^*_{\Q,t}(G_{\beta_2})$ and consider an arbitrary graph~$H \in \folio^*_{\Q,t}(G_\alpha \oplus G_{\beta_1})$. We will show that~$H \in \folio^*_{\Q,t}(G_\alpha \oplus G_{\beta_2})$.

By definition of~$\folio^*_{\Q,t}$ we have~$H \leqlb G_\alpha \oplus G_{\beta_1}$ and~$H \in \multipieces_{+t}(\Q)$. By Lemma~\ref{lem:split_in_pieces} there is a subset of pieces~$P \subseteq \pieces(H)$ such that~$H_1 := (\bigoplus_{p \in P} p)$ is a minor of~$G_\alpha$, and~$H_2 := (\bigoplus_{p \in \pieces(H)\setminus P} p)$ is a minor of~$G_{\beta_1}$. If either $H_1$ or $H_2$ is empty, it is easy to see that $H \in \foliotstar(G_\alpha \oplus G_{\beta_2})$, as needed. Otherwise, since~$H \in \multipieces_{+t}(\Q)$ and~$H_1, H_2 \in \multipieces(H)$, it follows that~$H_1,H_2 \in \multipieces_{+t}(\Q)$. 
 Hence the fact that~$H_2 \leqlb G_{\beta_1}$ implies that~$H_2 \in \folio^*_{\Q,t}(G_{\beta_1})$. Since~$H_2 \in \folio^*_{\Q,t}(G_{\beta_1}) \sqsubseteq \folio^*_{\Q,t}(G_{\beta_2})$ there exists~$H'_2 \in \folio^*_{\Q,t}(G_{\beta_2})$ with~$H_2 \leqlb H'_2$, which in turn implies that~$H_2 \in \folio^*_{\Q,t}(G_{\beta_2})$. Since~$H_1 \oplus H_2 = H$ and~$H_1 \leqlb G_{\alpha}$, Lemma~\ref{lem:split_in_pieces} now implies that~$H \leqlb G_\alpha \oplus G_{\beta_2}$. Thus~$H \in \folio^*_{\Q,t}(G_\alpha \oplus G_{\beta_2})$, as desired.

As the given argument applies to arbitrary~$H \in \folio^*_{\Q,t}(G_\alpha \oplus G_{\beta_2})$, we have~$\folio^*_{\Q,t}(G_\alpha \oplus G_{\beta_1}) \subseteq \folio^*_{\Q,t}(G_\alpha \oplus G_{\beta_2})$. Now we prove the statement about equality. If~$\folio^*_{\Q,t}(G_{\beta_1}) = \folio^*_{\Q,t}(G_{\beta_2})$, then by applying the above argument once we find~$\folio^*_{\Q,t}(G_\alpha \oplus G_{\beta_1}) \subseteq \folio^*_{\Q,t}(G_\alpha \oplus G_{\beta_2})$. Applying it once more with the roles of~$G_{\beta_1}$ and~$G_{\beta_2}$ reversed, we obtain the containment in the other direction, and conclude that~$\folio^*_{\Q,t}(G_\alpha \oplus G_{\beta_1}) = \folio^*_{\Q,t}(G_\alpha \oplus G_{\beta_2})$.
\end{proof}



%


\section{Main lemma} \label{sec:main:lemma}

In this section we prove our main lemma, Lemma \ref{lem:main}, by stating an inductive version of the lemma in Lemma \ref{lem:subset_Q}. We use this result to prove our main lemma at the end of this section.
To state the inductive version of our main lemma, we need the following additional definition.
%

\begin{definition}\label{def:minFdelsolwith}
For~$G_A,G_B,G_C \in \lbgraphs$ with boundary set~$S:= \delta(G_A\oplus G_B \oplus G_C)$, a given family~$\F \subseteq \congraphs$ of graphs, sets~$\Pi_A, \Pi_B,\Pi_C \subseteq \bgraphs$ (we will call these \emph{prohibitions}), set \Q of $X$-labeled graphs, and~$R_B \subseteq \lbgraphs$, define \minFdelsolwith (for \emph{opt. solution such that}) as:
\begin{align*}
\minFdelsolwith(G_A \oplus &G_B \oplus G_C, \Pi_A, \Pi_B, \Pi_C, R_B) :=\\ \{& Y \in \minFdelsol(G_A \oplus G_B \oplus G_C) \mid \\
 &Y \cap S = \emptyset \text{, and } \\
&G_i - Y^* \text{ has no boundaried } \Pi_i\text{-minor for any $i \in \{A,B,C\}$, and}\\
&R_B = \foliotstar(G_B - Y) \}.\qedhere
\end{align*}
\end{definition}
%
The prohibitions introduced above will be used to keep track of how minor models for \F are broken in graph $G_A \oplus G_B \oplus G_C$, by keeping track of which pieces are broken in each of the three graphs $G_A$, $G_B$, and $G_C$. Which graphs in $\F$ are minors of $G_A \oplus G_B \oplus G_C$, can be completely determined from the folios of $G_A$, $G_B$, and $G_C$. In particular, using Lemma \ref{lem:split_in_pieces}, we only need to consider minors in $\multipieces_{+t}(\F)$.

Similarly, to keep track of which minor models for graphs in \Q are broken by optimal solutions in $G$, we only need to keep track of which $\multipieces_{+t}(\Q)$ remain in $G_A$, $G_B$, and $G_C$, meaning that we only keep track of \foliotstar of these graphs after removing a solution.
 In the inductive application, we will prescribe a certain $\foliotstar$ for $G_B$ and we will query for the $\foliotstar$ of $G_A$ and leave the solution within $G_C$ unrestrained. 

\begin{definition}
Let two graphs be distinct when they are not isomorphic (see Definition \ref{def:isomorphism}).
Define $\numberof(\ell, t, n, \theta)$ as the number of distinct $\theta$-restricted $t$-boundaried $[\ell]$-labeled graphs on at most~$n$ vertices.
\end{definition}
We can now state the version of our main lemma that we will use to do induction. Recall that  $\|\F\|$ is defined as $\max_{H \in \F} |V(H)|$. For a graph~$G$, define~$\iscon(G) = 1$ if~$G$ is connected and~$0$ otherwise.
\begin{lemma}[{Main lemma (inductively)}]\label{lem:subset_Q}
Let~$X$ be a finite set, let $G_A, G_B, G_C \in \lbgraphs$ with $E(G_A)$, $E(G_B)$, $E(G_C)$ disjoint,  let $G:=G_A \oplus G_B \oplus G_C$, define $S:= \delta(G)$ and let $G$ satisfy $\td(G) \geq \td(G_A-S) + |S|$.


Let $\F \subseteq \congraphs$, let $\Q \subseteq \lcongraphs$ such that each graph in \Q has at most~$\max_{H \in \F} |E(H)|+1$ vertices and \Q is $\min_{H \in \F}|V(H)|$-saturated. Let $\Pi_A,\Pi_B,\Pi_C \subseteq \multipieces_{+t}(\F)\subseteq \bgraphs$, such that  $\Pi_A \odot \Pi_B \odot \Pi_C \supseteq \extend_{+t}(\F)$. Let $R_B \subseteq \multipieces_{+t}(\Q)$ be a set of $X$-labeled graphs.
%
Define
\begin{align*}
\mathcal{R}_0 := \{ & \foliotstar((G_A \oplus G_B)-Y) \mid
 Y \in \minFdelsolwith(G_A \oplus G_B \oplus G_C, \Pi_A, \Pi_B,\Pi_C, R_B)
\}.
\end{align*}
We also say that $Y$ is the \emph{corresponding \Fdel} for a remainder $R$ in $\mathcal{R}_0$. 
Let $\mathcal{R}$ be defined as the set of $\subseteq$-minimal elements from $\mathcal{R}_0$.
Let $\mathcal{R}_Q \subseteq \mathcal{R}$ (the remainders that leave a~$\Q$-minor) be defined as
\[\mathcal{R}_Q := \{R \in \mathcal{R} \mid \exists q \in \Q, \exists r \in R: q \leql \forget(r)\}.\]
Let $\mathcal{R}_N := \mathcal{R} \setminus \mathcal{R}_Q$ be the rest of the remainders, which leave no~$\Q$-minor. Let
\begin{align*}
\nu(\Pi_A) &:= | \multipieces_{+t}(\F) \setminus \Pi_A|\\
\xi(R_B) &:= \numberof \bigl(t \cdot \min _{H \in \F} |V(H)|, t, t + \max _{H \in \Q}|V(H)|, \min _{H \in \F} |V(H)| \bigr) - |R_B|\\
\mu(G_A,\Pi_A, S) &:= \minFdel(G_A,\Pi_A,S)  -\shrink \sum_{C \in \setcomponents(G_A-S)} \shrink \minFdel(C).
\end{align*}
Then there exist functions~$f$ and~$g$ (that do \emph{not} depend on~$|V(G)|$ or $|\Q|$) such that
\begin{enumerate}
\item $|\mathcal{R}_N| \leq f(\td(G_A-S), \iscon(G_A - S), \mu(G_A,\Pi_A,S), \nu(\Pi_A),\xi(R_B), \|\F\|, |S|)$ and
\item $\exists \Q^*\subseteq \Q$ such that
\[|\Q^*| \leq g(\td(G_A-S), \iscon(G_A - S), \mu(G_A,\Pi_A,S), \nu(\Pi_A),\xi(R_B), \|\F\|, |S|),\]
and for each $R\in \mathcal{R}_Q$ there exist $q \in \Q^*$, $r \in R$ with $q \leql \forget(r)$.
\end{enumerate}

\end{lemma}


\noindent Before giving the proof, we will explain the lemma statement and some of the main ideas of the proof. The above statement is a generalization of Lemma \ref{lem:main} that is needed in order to be able to do induction.
Initially, $G_A$ should be the entire graph  and $G_B$ and $G_C$ should be empty. When inductively using the lemma, parts from $G_A$ will be moved to $G_B$ and $G_C$. We use $\Pi_A$, $\Pi_B$ and $\Pi_C$ to keep track of which parts of graphs in \F are broken   in $G_A$, $G_B$, and $G_C$ by the solutions we are considering. A solution breaking $\Pi_X$ in $G_X$ for all $X \in \{A,B,C\}$ is guaranteed to break $\F$ in $G$, by the requirement that $\Pi_A\odot \Pi_B \odot \Pi_C \supseteq \extend_{+t}(\F)$. Since the graphs in \F are unlabeled, the graphs in $\Pi_A$, $\Pi_B$, and $\Pi_C$ are also unlabeled.

Graph $G_A$ is the ``query'' part, i.e., we want the lemma to tell us what $\folio^*$'s may remain in $G_A$ after removing a solution. Then $G_B$ is a ``prescribed'' part; we only want to know $\folio^*$'s in $G_A$ for solutions that have a prescribed effect in $G_B$. The behavior of $G_C$ is unrestrained, but not part of the output; we want to know all possible $\folio^*$'s that can be generated within $G_A$ by global solutions that do exactly as prescribed in $G_B$ and are free in $G_C$.

To see which decisions we made for $G_B$, remainder $R_B$ is introduced. Informally speaking, this set contains all fragments of graphs in \Q that solutions for $G$ that we are currently considering, will leave behind in $G_B$. This is relevant, since combined with graph $G_A$ these parts may form a graph in \Q.

$\mathcal{R}_0$ will then contain sets of graphs that we call remainders. Each remainder set corresponds to the parts of \Q that are left behind after deleting some optimal solution from $G$ (one remainder may be generated by more than one optimal solution). Intuitively, if some remainder in $\mathcal{R}_0$ is not minimal with respect to $\subseteq$, the corresponding minimal remainder is preferable: if no \FDeletion solution breaks all \Q-minors, then any minimal remainder contains a graph from \Q.
Thereby, will only consider the $\subseteq$-minimal remainders from this set, which form $\mathcal{R}$. We then promise that there exists a small subset $\Q^* \subseteq \Q$ that covers all remainders that have some \Q-minor, which will give us Lemma \ref{lem:main}. Furthermore we promise that the set of remainders with no \Q-minor, called $\mathcal{R}_N$, is appropriately bounded. This last claim is only needed for the induction, as these remainders can later be combined into bigger remainders that contain $\Q$-minors, and we need the number of options to consider to be sufficiently small.

The general proof strategy is to consider graph $G_A-S$. If it is connected, we find a vertex $s$ whose removal from~$G_A - S$ reduces its treedepth. We consider two types of optimal solutions, namely those containing $s$ and those not containing $s$. For solutions containing $s$, we obtain remainders by removing $s$ from the graph and applying the induction hypothesis. For solutions avoiding $s$, we obtain remainders by adding $s$ to $S$ and applying induction. We then take the union of both.

If $G_A - S$ is not connected, we split off one connected component $C$. We show how to combine each remainder obtained from $C$ with the remainders obtained from what remains of $G_A$, when $C$ is moved to $G_B$ and $R_B$ is adapted correctly.  This requires some careful analysis, as the size bounds are not allowed to depend on the number of connected components of $G_A-S$.

We bound $|\Q'|$ and $|\mathcal{R}_N|$ by $f$ and $g$ respectively, that depend on a number of parameters introduced for the induction. $\iscon(G_A-S)$ is used to denote if $G_A-S$ is connected, since this case is ``easier'' than the disconnected case, where we need to carefully combine remainders from all different connected components, without breaking the size bounds. In the case that $G_A - S$ is not connected, we split off a component $C$ from $G_A - S$. It can happen that a globally optimal solution, is not locally optimal in $C$. However, this should not be the case for too many components, because otherwise removing $S$ and taking an optimal solution in each component gives a smaller solution. We will use $\mu$ to keep track of the difference between the size of a globally optimal solution in $G_A$, and sizes of the locally optimal solutions in all components of $G_A - S$, and by the reasoning above $\mu$ should be bounded by $|S|$.  We use $\nu$ to keep track of how restrictive prohibition $\Pi_A$ is, the idea is that as we split $G_A$ further, smaller bits of $\F$ should be broken in each component and thus $\nu$ will decrease. Furthermore, $\xi(R_B)$ takes care of the complexity of $R_B$. At some point you can add nothing new to $R_B$ without introducing $\Q$-minors, which will help to bound $\xi(R_B)$. The bounds depend on $|S|$ but the size of this set will be bounded by the treedepth of $G$, by the precondition that $\td(G_A - S) + |S| \leq \td(G)$.

\begin{proof}[Proof of Lemma \ref{lem:subset_Q}.]
We will construct sets $\mathcal{R}'_N$ and $\Q' \subseteq \Q$ such that
\begin{align}\label{req:RprimeN}
\mathcal{R}_N \subseteq \mathcal{R}'_N
\end{align}
and
\begin{align}\label{req:Qprime}
\text{for all $R \in \mathcal{R}_{\Q}$ there exist graphs $q \in \Q', r\in R$ such that $q \leql \forget(r)$.}
\end{align}
We will then show that $\mathcal{R}'_N$ and $\Q'$ satisfy the required size bounds.
We do this by induction on the number of vertices of $G_A - S$.

\paragraph*{Base cases} We have a number of different base cases.
\begin{enumerate}
\item \label{base-case:empty-graph} Suppose $|V(G_A - S)| = 0$. Consider any two solutions $Y, Y' \in \minFdelsolwith(G_A \oplus G_B \oplus G_C, \Pi_A,\Pi_B,\Pi_C,R_B)$. By definition \ref{def:minFdelsolwith}, $\foliotstar(G_B - Y)$ and $\foliotstar(G_B - Y')$ are required  to always equal $R_B$  and no vertex from $G_A$ may be removed (as it only contains vertices in $S$).  Thus, it holds that $G_A - Y = G_A - Y' = G_A$ and $\foliotstar(G_B - Y) = \foliotstar(G_B-Y') = R_B$. It follows from Lemma \ref{lem:replacehalf:preservefolio}, that $\foliotstar((G_A \oplus G_B) - Y) =\foliotstar((G_A \oplus G_B) - Y')$.

     Since this holds for any two solutions, there is in fact at most one remainder $R \in \mathcal{R}$. If $R$ has no $\Q$-minor, $\mathcal{R}_N$ has size one, $\Q^*$ is empty and we are done. If $R$ does have a $\Q$-minor, $\mathcal{R}_N$ is empty and we add one such $\Q$-minor to $\Q^*$ such that $|\Q^*| = 1$.
\item \label{base-case:mu-large} $\mu(G_A,\Pi_A,S) > |S|$. 
As~$\F$ only contains connected graphs,~$\sum_{C \in \setcomponents(G_A-S)} \minFdel(C) = \minFdel(G_A - S)$ and there is an \FDeletion solution~$Y^*$ in~$G$ of size~$|S| + \minFdel(G_A - S) + \minFdel(G_B - S) + \minFdel(G_C - S)$ that removes all of~$S$ together with optimal solutions in the three subgraphs. If~$\mu(G_A,\Pi_A,S) > |S|$, then any \FDeletion solution in~$G_A$ that breaks~$\Pi_A$ and is disjoint from~$S$ is strictly larger than~$Y^*$, since it contains~$\minFdel(G_A,\Pi_A,S) > |S| + \minFdel(G_A - S)$ vertices from~$G_A$. Consequently, the set~$\minFdelsolwith(G_A \oplus G_B \oplus G_C, \Pi_A, \Pi_B,\Pi_C, R_B)$ is empty and therefore the statement is vacuously true.
\item \label{base-case:Q-in-RB} $R_B \cap \extend_{+t}(\Q) \neq \emptyset$. Consider a graph~$q' \in R_B \cap \extend_{+t}(\Q)$. Then by Definition~\ref{def:minFdelsolwith}, the set~$\mathcal{R}_0$ we consider only contains remainders corresponding to solutions~$Y$ in which~$q' \in \foliotstar(G_B - Y)$. Let~$q \in \Q$ such that~$q' \in \extend_{+t}(\Q)$. It follows that~$q \leql \forget(q')$ and therefore that~$q$ is a minor of~$\forget(q')$ with $q' \in R$ for all~$R \in \mathcal{R}_0$. Thereby, $\mathcal{R}_N = \emptyset$ and if we define $\Q^* := \{q\}$, it is easy to see that every remainder in $\mathcal{R}$ has a $\Q^*$-minor and $|\Q^*| = 1$.
\item \label{base-case:xi-large} $\xi(R_B) < 0$ and $R_B \cap \extend_{+t}(\Q) = \emptyset$. This means that $|R_B| > N(t \cdot \min_{H \in \mathcal{F}}|V(H)|,t,t+\|\Q\|,\min_{H \in \F}|V(H)|)$ (recall Definition \ref{def:graphs} for the meaning of $N(\ell,t,n,\theta)$). Since $R_B$ is a subset of $\multipieces_{+t}(\Q)$, it follows that  all graphs in $R_B$ have at most $t+\|\Q\|$ vertices. Furthermore, they are $t$-boundaried. So, if $|R_B| > N(t \cdot \min_{H \in \mathcal{F}}|V(H)|,t,t+\|\Q\|,\min_{H \in \F}|V(H)|)$, either some graph has more than $\min_{H \in \F}|V(H)|$ labels on a single vertex, or more than $t \cdot \min_{H \in \mathcal{F}}|V(H)|$ labels are used in total. Suppose some vertex in $R_B$ has more than $\min_{H \in \F}|V(H)|$ labels. Then $\mathcal{R}_N = \emptyset$, since every remainder $R\in \mathcal{R}$ satisfies $R_B\subseteq R$, and $R_B$ contains a graph with at least $\min_{H \in \F}|V(H)|$ labels on one vertex. Since $\Q$ is $\min_{H \in \F}|V(H)|$-saturated, the vertex with the first $\min_{H \in \F}|V(H)|$ of these labels is in $\Q$. Let $\Q^*$ contain only this graph, it is easy to see that every remainder in $\mathcal{R}$ has a $\Q^*$-minor and $|\Q^*| = 1$.

    Suppose more than $t \cdot \min_{H \in \mathcal{F}}|V(H)|$ labels are used in total. For a vertex $s \in S$, let $\ell(s) := \{x \in X \mid \exists r \in R_B \colon \exists v \in V(r) \colon x \in \lab_r(v)$  and $v$ is in the same connected component as $s$ in graph $r\}$. Thus, $\ell(s)$ contains all labels that occur in the same connected component as $s$ in at least one graph in $R_B$. By the assumption that $R_B \cap \extend_{+t}(\Q) = \emptyset$, together with the precondition that~$R_B \subseteq \multipieces_{+t}(\Q)$ and Lemma~\ref{lem:mpcs:plust:attached:or:extend}, we find that $R_B \subseteq \attlbgraphs$: all connected components of all graphs in~$R_B$ contain a boundary vertex. Hence every label occurring on a vertex in~$R_B$ appears in a common component with some boundary vertex, which implies~$|\bigcup_{s \in S} \ell(s)| > t \cdot \min_{H \in \mathcal{F}}|V(H)|$. By the pigeonhole-principle, there exists~$s \in S$ such that $|\ell(s)| > \min_{H \in \mathcal{F}}|V(H)|$. 
		
		Let~$X'$ be any subset of~$\ell(s)$ of size~$\min_{H \in \mathcal{F}}|V(H)|$. We show that the single-vertex graph with labelset~$X'$ is a minor of~$G - Y$ for any~$Y \in \minFdelsolwith(G_A \oplus G_B \oplus G_C, \Pi_A,\Pi_B,\Pi_C,R_B)$ and thus we can use a~$\Q^*$ of size one and~$\mathcal{R}_N =\emptyset$. Since for each~$x \in X'$ there is a graph in~$R_B$ that has a vertex with label~$x$ in the same component as~$s$, in~$G-Y$ vertex~$s$ is in the same component as one vertex with label~$x$, for all~$x \in X'$. If we contract this component to~$s$ and delete all other components in the graph, we get exactly the desired (unboundaried) minor. Furthermore, since~$G-Y$ always has a~$\Q^*$-minor, any remainder in~$\mathcal{R}$ will contain a graph from~$\extend_{+t}(\Q^*)$ by Lemma~\ref{lem:extend:minor}.
\end{enumerate}

\paragraph*{Step} Having covered the base cases, we continue by describing the induction step. Suppose the lemma statement holds for all cases where $|V(G_A-S)|$ is smaller. We do a case distinction on whether $G_A - S$ is connected or not. Refer to Figure \ref{fig:sketch} for a sketch of the situation in the graph.
\begin{figure}[t]
\centering
\includegraphics{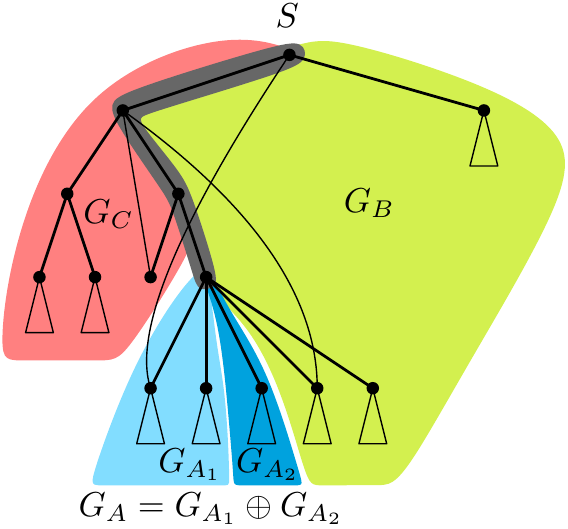}
\caption{A sketch of $G$ and its treedepth decomposition, where $G_A$, $G_B$, and $G_C$ are indicated, as they could be when the lemma is applied inductively, starting from $G_A = G$ and both other graphs empty. The 4-vertex set~$S$ is the common boundary of the three subgraphs.}
\label{fig:sketch}
\end{figure}


\paragraph*{$\mathbf{G_A-S}$ disconnected} Suppose $G_A - S$ consists of multiple components. Order them arbitrarily, and let $C$ be the last component. Define $G_{A_2}$ as the $X$-labeled $t$-boundaried subgraph of $G_A$ induced by $S \cup C$, from which all edges between vertices of $S$ have been removed. Let $G_{A_1} := G_A[V(G_A)\setminus V(C)]$. By this definition, $G_{A_1} \oplus G_{A_2}$ equals $G_A$, and all edges between vertices in $S$ are in $G_{A_1}$. 

In the following procedures and proofs, for $(\piaone,\piatwo)\in\splitPi(\Pi_A)$, let
$\Rtilde$ (respectively $\Rtilde_N$,$\Rtilde_0$) denote the set $\mathcal{R}$ (respectively $\mathcal{R}_0$, $\mathcal{R}_N$) obtained by applying the lemma inductively to the graph $G_{A_2} \oplus G_B \oplus (G_C \oplus G_{A_1})$, with $G_{A'} := G_{A_2}$ and the new $G_{C'}$ as $G_C \oplus G_{A_1}$ and with prohibitions $\piatwo$ and $\Pi_B$ and $\piaone \odot \Pi_C$ and remainder $R_B$. Let $\Qtilde \subseteq \Q$ be a set $\Q^*$ of bounded size as described by the lemma applied to the stated parameters.

Let \Rhat{R} (respectively $ \Rhat{R}_0$, $ \Rhat{R}_N$) denote the set $\mathcal{R}$ (respectively $\mathcal{R}_0$, $\mathcal{R}_N$) obtained in graph $G_{A_1} \oplus (G_{A_2} \oplus G_B)\oplus G_C$ with $G_{A'} := G_{A_1}$ and with prohibitions $\piaone$, $\piatwo \odot \pib$, and $\Pi_C$ and remainder $R$. Let $\Qhat{R} \subseteq \Q$ be a set $\Q^*$ obtained with the same parameters, that has suitably bounded size. It can be found by the induction hypothesis. 

See Figure \ref{fig:division-hat-tilde} for a sketch of the graphs in which \Rtilde and \Rhat{R} are obtained. Using these inductively acquired objects, we define a set $\mathcal{R}'_N \subseteq 2^{\lbgraphs}$ by the following procedure. An intuitive description is provided below its definition.
\begin{figure}
\centering
\includegraphics{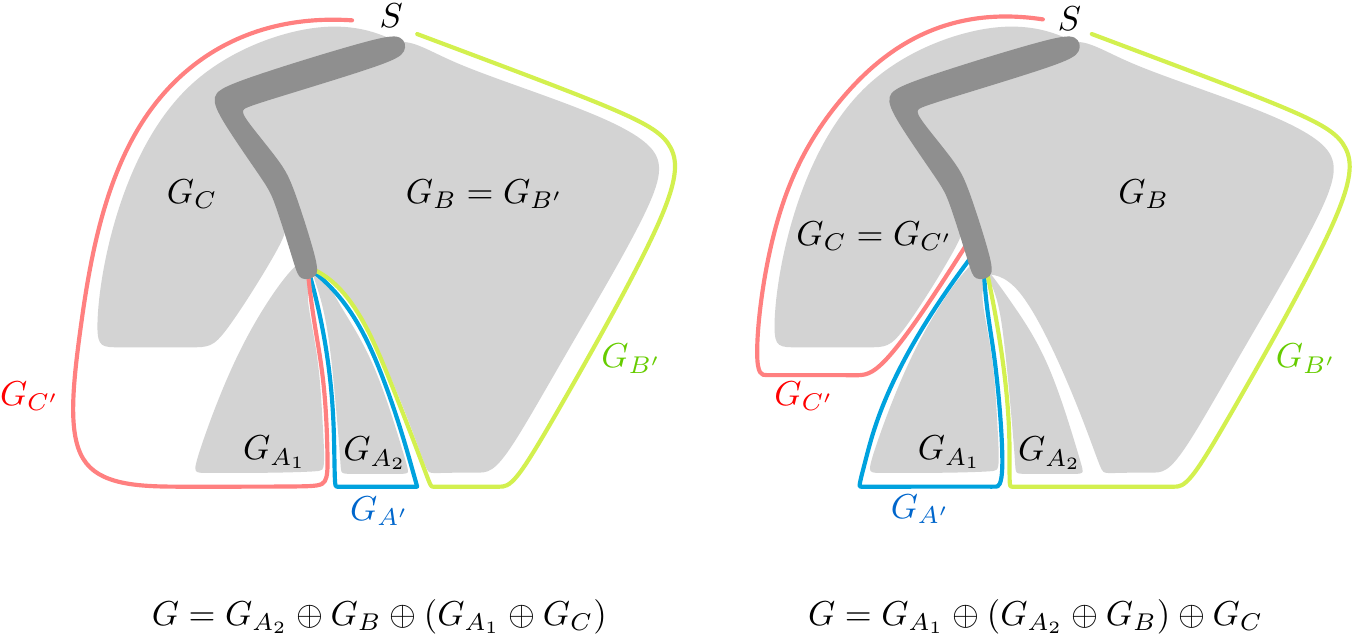}
\caption{On the left the graph in which \Rtilde is obtained is shown, on the right the graph where \Rhat{R} is obtained is given.}
\label{fig:division-hat-tilde}
\end{figure}
%
\begin{enumerate}
\setlength\itemsep{1mm}
\item[]
\item[]  \textsc{Find-Remainders} (disconnected case)
\item \label{step:emptypi}\textbf{if} there exists $\Pi_{A_2}$ such that $(\Pi_A,\Pi_{A_2}) \in \splitPi(\Pi_A)$ \textbf{and} $\minFdel(G_{A_2}, \Pi_{A_2},S) = \minFdel(G_{A_2})$ \textbf{and} $R_B \in \widetilde{\mathcal{R}}^{\piatwo,\Pi_A}$ \textbf{then}
\item \label{step:if-Rhat} \quad Let $\mathcal{R}'_N := \hat{\mathcal{R}}_N^{\Pi_A,\piatwo,R_B}$ and
\item \label{step:if-Qhat}\quad Let $\mathcal{\Q'} := \hat{\mathcal{\Q}}^{\Pi_A,\piatwo,R_B}$
\item \label{step:else}\textbf{else}
\item \label{step:split}\quad \textbf{for each pair} $(\piaone,\piatwo) \in \splitPi(\Pi_A)$, such that there exists
\item[] \quad $Y \in \minFdelsolwith(G_A \oplus G_B \oplus G_C, \Pi_A,\Pi_B,\Pi_C,R_B)$  s.t. $G_{A_1} - Y$ is \piaone-minor free and
\item[]\quad $G_{A_2} - Y$ is \piatwo-minor free \textbf{do}
\item\label{step:tilde} \qquad Add \Qtilde to $\Q'$. Consider \Rtilde
\item \qquad \label{step:foralltilde} \textbf{for each } $R \in \Rtilde_N$ \textbf{do}
\item\label{step:hat} \quad\qquad Add \Rhat{R} to $\mathcal{R}_N'$ and add \Qhat{R} to $\Q'$
\end{enumerate}

It is easy to verify that when the lemma is used inductively in the above procedure, all preconditions for the lemma are satisfied.

Let us try to give some intuition behind this procedure.
Let $R \in \mathcal{R}$ be a remainder. We want to make sure that if $R \in \mathcal{R}_N$, it is added to $\mathcal{R}_N'$ and otherwise some \Q-minor of $R$ is added to $\mathcal{Q}'$. The idea is that there exists a  solution $Y \in \minFdelsolwith(G_A \oplus G_B \oplus G_C, \Pi_A,\Pi_B,\Pi_C, R_B)$
 that corresponds to $R$, such that we can do the following.

Clearly,
 there exist $(\Pi_{A_1},\Pi_{A_2}) \in \splitPi(\Pi_A)$ such that $Y$ breaks $\Pi_{A_1}$ in $G_{A_1}$ and $\Pi_{A_2}$ in $G_{A_2}$. We try all such pairs in line~\ref{step:split}. 
 To do induction it is needed to decrease the number of vertices in $G_A$. To achieve this, we ``split off'' one connected component of $G_A - S$, namely $G_{A_2}$.  The idea is that when $G_{A_2}-Y$ has a $\Q$-minor, then it is added to $\Q'$ in line \ref{step:tilde} and we are done. Otherwise, we  move  $G_{A_2}$ to $G_B$, by considering the graph $G_{A_1} \oplus (G_{A_2} \oplus G_B) \oplus G_C$ with prohibitions $\Pi_{A_1}$, $\Pi_{A_2} \odot \Pi_B$, and $\Pi_C$. Consider $R_B':= \foliotstar((G_{A_2} \oplus G_B) - Y)$. In line \ref{step:foralltilde}, we  ``guess'' this $R_B'$. Then, the remainder $R$ should be a remainder of  $G_{A_1} \oplus (G_{A_2} \oplus G_B) \oplus G_C$, with respect to $R_B'$ and the new prohibitions (note that $G_A \oplus G_B = G_{A_1} \oplus (G_{A_2} \oplus G_B)$), which can be shown by the existence of $Y$. Thus, $R \in \hat{\mathcal{R}}^{\Pi_{A_1},\Pi_{A_2},R_B'}.$ If $R \in \mathcal{R}_N$, then $R$ is correctly added to $\mathcal{R}_N'$ in line \ref{step:hat}. Otherwise, a $\Q$-minor of $R$ is added to $\Q'$ in line \ref{step:hat}. We will use that the induction hypothesis is applied to situations with ``smaller'' parameters to bound the sizes of $\Q'$ and $\mathcal{R}'_N$.

At this point, one might wonder why the if-part of the above procedure (line~\ref{step:emptypi} to~\ref{step:if-Qhat}) exists. The point is that in this case, we cannot guarantee that the parameters for which we apply the induction hypothesis will indeed be strictly ``smaller''. Luckily, in this specific situation we can ``copy'' all remainders from a single application of the induction hypothesis, whose parameters are not worse than the current parameters, which is enough to get the appropriate size bounds.
When the if-condition does not hold and the else-part (starting in line~\ref{step:else}) is executed, we can guarantee that the induction hypothesis is applied to ``smaller'' parameters than the current ones, and use this to argue that their combination will satisfy the required size bound.

The following two claims show that $\mathcal{R}'_N$ and $\Q'$ satisfy the requirements given in~\eqref{req:RprimeN} and~\eqref{req:Qprime}. Claim~\ref{claim:QandRsmall} proves that they satisfy the desired size bound.

\begin{innerclaim}\label{claim:main-lemma-disconnected}
Suppose $G_A - S$ is not connected, and that there exists $\piatwo$ such that $(\Pi_A,\piatwo) \in \splitPi(\Pi_A)$, $\minFdel(G_{A_2}, \piatwo,S) = \minFdel(G_{A_2})$, and  $R_B \in \widetilde{\mathcal{R}}^{\piatwo,\Pi_A}$. Then $\mathcal{R}_N \subseteq \mathcal{R}'_N$ and for each $R \in \mathcal{R}_Q$, there exist $q \in \Q', r \in R$ such that $q\leql\forget(r)$.
\end{innerclaim}
\noindent  Before giving the proof, let us give an outline.
%
This claim considers the situation that a solution exists in $G_{A_2}$ of size $\minFdel(G_{A_2})$ that does not intersect $S$ and has a very large effect in breaking $\Pi_A$ in $G_A$, namely by breaking $\piatwo$ in $G_{A_2}$ such that $\piatwo \odot \Pi_A \supseteq \Pi_A$. The second condition $R_B \in \widetilde{\mathcal{R}}^{\piatwo,\Pi_A}$ states that there exists a minimum $\F$-deletion in $G$ whose restriction to $G_{A_2}$ has these desired properties, and such that the remainder it leaves in $G_{A_2} \oplus G_B$ is exactly the same as the remainder $R_B$ (which equals the remainder that is made in $G_B$ by the solutions we are currently considering). The main idea is that under these conditions, the remainders in $\mathcal{R}$ can all be generated by solutions that break $\piatwo$ in $G_{A_2}$ and leave remainder $R_B$ in $G_{A_2} \oplus G_B$. Hence, a remainder in $\mathcal{R}_N$ is also a remainder in
$\hat{\mathcal{R}}^{\piatwo,\Pi_A,R_B}$ and thereby $R\in \mathcal{R}_N'$ as promised. Furthermore a remainder $R \in \mathcal{R}_{\Q}$ has a $\Q'$-minor since it is contained in
$\hat{\mathcal{R}}^{\piatwo,\Pi_A,R_B}_{\Q}$ and we added set $\hat{\mathcal{Q}}^{\piatwo,\Pi_A,R_B}$ to $\Q'$.

To actually show these properties, let $R \in \mathcal{R}$ and let $Y$ be its corresponding solution, such that $\foliotstar((G_A \oplus G_B) - Y) = R$. Let $Y'$ be the solution that shows that $R_B \in \widetilde{\mathcal{R}}^{\piatwo,\Pi_A}$, such that $\foliotstar((G_{A_2} \oplus G_B) - Y') = R_B$. We combine these solutions to obtain a solution $\hat{Y}$, that equals $Y$ in $G_{A_1}$ and $G_C$ and equals $Y'$ in $G_{A_2}$ and $G_B$. We denote the corresponding remainder by $\hat{R} = \foliotstar((G_A \oplus G_B) - \hat{Y})$ and show that $\hat{R} \subseteq R$.

 We then show that $\hat{R}$ is contained in $\hat{\mathcal{R}}_0^{\Pi_A,\piatwo,R_B}$  in \eqref{eq:hatR_in_Rhat} in the proof. Since we do not necessarily know if $\hat{R}$ is minimal for this set, we replace it by minimal remainder $\hat{R}'$. For this remainder, we have added elements to $\mathcal{R}_N'$ and $\Q'$ as needed. We then show that $\hat{R}' \in \mathcal{R}_0$ in  \eqref{eq:hatRprime_in_R0}, in order to conclude that either $\hat{R}'$ is smaller than $R$ implying $R$ is not minimal (which contradicts that $R \in \mathcal{R}$), or $\hat{R}'$ equals $R$ and we are done.

We regularly use Lemma \ref{lem:mergePi}, saying that if a solution breaks $\Pi_{A_1}$ in $G_{A_1}$ and $\Pi_{A_2} $ in $G_{A_2}$, it breaks $\Pi_{A_1} \odot \Pi_{A_2}$ in $G_{A_1} \oplus G_{A_2}$. Furthermore, we regularly use that if a set breaks $\Pi_A \odot \Pi_B \odot \Pi_C$ in $G$, it breaks all occurrences of $\F$, by $\Pi_A \odot \Pi_B\odot \Pi_C \supseteq \extend_{+t}(\F)$ together with Lemma \ref{lem:extend:minor}. 
We proceed by formalizing these ideas into a rigorous proof.
\begin{innerproof}[{Proof of Claim \ref{claim:main-lemma-disconnected}.}]
Let $R \in \mathcal{R}$, let $Y \in \minFdelsolwith(G_A \oplus G_B \oplus G_C, \Pi_A,\Pi_B,\Pi_C,R_B)$ be a corresponding \Fdel. 
Let $\piatwo$ be given, such that $(\Pi_A,\piatwo) \in \splitPi(\Pi_A)$, $R_B \in \widetilde{R}^{\piatwo,\Pi_A}$ and $\minFdel(G_{A_2},\piatwo,S) = \minFdel(G_{A_2})$.

Let $Y' \in \minFdelsolwith(G_{A_2} \oplus G_B \oplus (G_{A_1} \oplus G_C), \piatwo, \Pi_B,\Pi_C \odot \Pi_A, R_B)$ be a solution showing that $R_B \in \widetilde{\mathcal{R}}^{\piatwo,\Pi_A}$, such that $R_B = \foliotstar((G_{A_2} \oplus G_B)-Y')$. Define $\hat{Y} := (Y' \cap V(G_{A_2} \oplus G_B)) \cup (Y \cap V(G_{A_1} \oplus  G_C))$. Let $\hat{R}$ be the remainder corresponding to this solution, thus $\hat{R} := \foliotstar((G_A \oplus G_B)-\hat{Y})$.

We start by showing that
\begin{equation}\label{eq:hatR_in_Rhat}
\hat{R} \in \hat{\mathcal{R}}_0^{\Pi_A,\piatwo,R_B},
\end{equation}
by showing $\hat{Y} \in \minFdelsolwith(G_{A_1} \oplus (G_{A_2} \oplus G_B)\oplus G_C, \Pi_A, \piatwo\odot \Pi_B,\Pi_C, R_B)$. Since $\hat{Y}$ breaks $\piatwo$ in $G_{A_2}$, $\Pi_B$ in $G_B$, $\Pi_C$ in $G_C$ and $\Pi_A$ in $G_{A_1}$, $\hat{Y}$ breaks $\Pi_A \odot \piatwo \odot \Pi_B \odot \Pi_C$ in $G_A \oplus G_B \oplus G_C$ by Lemma \ref{lem:mergePi}. Since $(\piatwo,\Pi_A) \in \splitPi(\Pi_A)$ by assumption, it follows by Lemma \ref{lem:merge-vs-split} that $\piatwo\odot \Pi_A \supseteq \Pi_A$.
Since $\Pi_A \odot \Pi_B \odot \Pi_C \supseteq \extend_{+t}(\F)$, it follows that $\hat{Y}$ is a solution  in $G$ (meaning that it is an $\F$-minor free deletion) by Lemma \ref{lem:extend:minor}.

By definition, $\hat{Y} \cap S =\emptyset$ and $G_{A_1} -\hat{Y}$ has no $\Pi_A$-minor and $G_C-\hat{Y}$ has no $\Pi_C$-minor. Furthermore, $(G_{A_2} \oplus G_B)-\hat{Y}$ has no $(\piatwo \odot \Pi_B)$-minor, since $(G_{A_2} \oplus G_B)-\hat{Y} = (G_{A_2}-\hat{Y}) \oplus (G_B - \hat{Y})$ and $G_{A_2}-\hat{Y}$ has no $\piatwo$-minor and $G_B-\hat{Y}$ has no $\Pi_B$-minor, this follows from Lemma \ref{lem:mergePi}. The fact that $R_B = \foliotstar((G_{A_2} \oplus G_B) - \hat{Y})$ is clear from $\hat{Y} \cap V(G_{A_2} \oplus G_B) = Y' \cap V(G_{A_2} \oplus G_B)$ and  by definition of $Y'$ that $R_B  = \foliotstar((G_{A_2} \oplus G_B) - Y')$. It remains to show that $\hat{Y}$ has the required size. We use that $|Y| = |Y'| = \minFdel(G)$. Suppose $|\hat{Y}| > |Y|$, then $|Y' \cap V(G_{A_2} \oplus G_B)| > |Y \cap V(G_{A_2} \oplus G_B)|$. But then $Y' \notin \minFdelsol(G)$, since a smaller solution can be obtained
by replacing $Y' \cap V(G_B)$ with $Y \cap V(G_B)$
and replacing $Y' \cap V(G_{A_2})$ by a solution of size $\minFdel(G_{A_2})$ breaking $\piatwo$ in $G_{A_2}$. This replacement yields a solution in $\minFdelsol(G)$, since $\piatwo \odot \Pi_A \supseteq \Pi_A$ by Lemma~\ref{lem:merge-vs-split}, so that all of $\Pi_A \odot \Pi_B \odot \Pi_C \supseteq \extend_{+t}(\F)$ is broken, implying by Lemma~\ref{lem:mergePi} that all $\F$-minors are broken in $G$.

 This is a valid solution, since it gives the same prohibitions, and it is smaller than $Y'$, which is a contradiction. 
Thereby, $\hat{R} \in \hat{\mathcal{R}}_0^{\Pi_A,\piatwo,R_B}$, which establishes \eqref{eq:hatR_in_Rhat}.

By \eqref{eq:hatR_in_Rhat}, there exists a corresponding minimal remainder $\hat{R}' \in \hat{\mathcal{R}}^{\Pi_A,\piatwo,R_B}$ such that $\hat{R}' \subseteq \hat{R}$.
We now show that:
\begin{equation}\label{eq:hatRprime_in_R0}
\hat{R}' \in \mathcal{R}_0.
\end{equation}
Let $Y''$ be a solution corresponding to $\hat{R}'$, thus $Y'' \in \minFdelsolwith(G_{A_1} \oplus (G_{A_2} \oplus G_B) \oplus G_C, \Pi_A,\piatwo\odot \Pi_B,\Pi_C,R_B)$. Define $\hat{Y}' := (Y' \cap V(G_{A_2} \oplus G_B) \cup (Y'' \cap V(G_{A_1} \oplus G_C))$. We show that $\hat{R}' = \foliotstar((G_A \oplus G_B)-\hat{Y}')$  and that $\hat{Y}'$ satisfies the required properties.
Since $\foliotstar((G_{A_2} \oplus G_B)-\hat{Y}') = \foliotstar((G_{A_2} \oplus G_B)-Y') = R_B$ and
$G_{A_1} - Y'' = G_{A_1} - \hat{Y}'$, it follows from Lemma \ref{lem:replacehalf:preservefolio} that $\hat{R}'$ is a remainder corresponding to $\hat{Y}'$ in $G_A \oplus G_B$.

Now we show that $\hat{Y}' \in \minFdelsolwith(G_A \oplus G_B \oplus G_C ,\Pi_A,\Pi_B, \Pi_C,R_B)$, such that $\hat{R}' \in \mathcal{R}_0$ follows.
By definition, $\hat{Y}' \cap S = \emptyset$. Furthermore, $\hat{Y}'$ breaks $\piatwo$ in $G_{A_2}$ since $\hat{Y}' \cap V(G_{A_2}) =  Y' \cap V(G_{A_2})$ which breaks $\piatwo$ in $G_{A_2}$ by definition. Also, $\hat{Y}'$ breaks $\Pi_A$ in $G_{A_1}$, since $\hat{Y}' \cap V(G_{A_1}) = Y'' \cap V(G_{A_1})$ which breaks $\Pi_A$ in $V(G_{A_1})$ by definition. We conclude $\hat{Y}'$ breaks $\Pi_A$ in $G_A$. Furthermore $\hat{Y}'$ breaks $\Pi_B$ in $G_B$, since $\hat{Y}' \cap V(G_B) = Y' \cap G_B$ which breaks $\Pi_B$ by definition. Similarly, $\hat{Y}' \cap V(G_C) = Y''\cap V(G_C)$, which breaks $\Pi_C$ in $G_C$ by definition.
The requirement for $R_B$ is satisfied since $\hat{Y}' \cap V(G_B) = Y' \cap V(G_B)$ and $Y'$ satisfies the requirement for $R_B$.
That $\hat{Y}'$ is a solution follows from the fact that $\hat{Y}'$ gives the required prohibitions in $G_A$, $G_B$ and $G_C$. It remains to show that $\hat{Y}'$ has minimal size, we will use that  $|Y''| = |Y'|= \minFdel(G_A \oplus G_B \oplus G_C)$. Suppose $|\hat{Y}'| > |Y''|$, then $|Y' \cap V(G_{A_2} \oplus G_B)| > |Y'' \cap V(G_{A_2} \oplus G_B)|$. But then $Y'$ is not minimal, as a smaller solution can be obtained by taking $Y''$ in $G_B$ and  a solution breaking $\piatwo$ in $G_{A_2}$ of size $\minFdel(G_{A_2})$. This solution breaks $\piatwo \odot \Pi_B$ in $G_{A_2} \oplus G_B$ and $\Pi_A \odot \Pi_C$ in $G_{A_1} \oplus G_C$, thus it is a valid solution. Thus, $|\hat{Y}'| = \minFdel(G_A \oplus G_B \oplus G_C)$.
Thereby, $\hat{R}' \in \mathcal{R}_0$, which proves \eqref{eq:hatRprime_in_R0}.

Furthermore, we prove
\begin{equation}\label{eq:hatRsmallerthanR}
\hat{R}' \subseteq \hat{R} \subseteq R.
\end{equation}
$\hat{R}' \subseteq \hat{R}$ is true by definition.
Since $G_{A_1} - Y = G_{A_1} - \hat{Y}$ and $\foliotstar((G_{A_2} \oplus G_B) - \hat{Y}) = R_B \subseteq \foliotstar((G_{A_2} \oplus G_B) - Y)$, it follows that $\hat{R} \subseteq R$ from Lemma \ref{lem:replacehalf:preservefolio}.

We can now show that
\begin{equation}\label{eq:RN}
R \in \mathcal{R}_N \Rightarrow R \in \mathcal{R}_N'.
\end{equation}
If $R \in \mathcal{R}_N$, it follows that $\hat{R}'$ has no $\Q$-minor, thus $\hat{R}' \in \hat{\mathcal{R}}^{\Pi_A,\piatwo,R_B}_N$ and it was added to $\mathcal{R}'_N$ in line~\ref{step:if-Rhat}. It follows from~\eqref{eq:hatRprime_in_R0} and~\eqref{eq:hatRsmallerthanR} that either $R$ is not minimal, which is a contradiction, or $R \in \mathcal{R}_N'$.

Finally, we prove
\begin{equation}\label{eq:RQ}
R \in \mathcal{R}_{\Q} \Rightarrow \exists q \in \Q',r\in R: q \leqlb \forget(r).
\end{equation}
If $R \in \mathcal{R}_{\Q}$, and $\hat{R}'$ has no $\Q$-minor, it follows from~\eqref{eq:hatRprime_in_R0} and~\eqref{eq:hatRsmallerthanR} that $R$ was not minimal in $\mathcal{R}_0$, which is a contradiction with $R \in \mathcal{R}$.
If $\hat{R}'$ does have a $\Q$-minor, one of these minors was added to $\Q'$ in line \ref{step:if-Qhat}, and this minor is also a minor of some graph in $R$.

The lemma statement now follows from  \eqref{eq:RN} and \eqref{eq:RQ}.
\end{innerproof}

\begin{innerclaim}\label{claim:main-lemma-disconnected-else}
Suppose $G_A - S$ is not connected, and furthermore that there exists no $\piatwo$ such that $(\Pi_A,\piatwo) \in \splitPi(\Pi_A)$, $\minFdel(G_{A_2}, \piatwo,S) = \minFdel(G_{A_2})$, and  $R_B \in \widetilde{\mathcal{R}}^{\piatwo,\Pi_A}$. Then $\mathcal{R}_N \subseteq \mathcal{R}'_N$ and for each $R \in \mathcal{R}_Q$, there exist $q \in \Q', r \in R$ such that $q \leql \forget(r)$.
\end{innerclaim}
\noindent Before stating the proof, we will first give an outline.
The precondition for this claim that there exists no $\piatwo$ such that $(\Pi_A,\piatwo) \in \splitPi(\Pi_A)$, $\minFdel(G_{A_2}, \piatwo,S) = \minFdel(G_{A_2})$, and  $R_B \in \widetilde{\mathcal{R}}^{\piatwo,\Pi_A}$, is essential to show that we perform lines \ref{step:else} to \ref{step:hat} of the procedure, but is not needed for the proof of Claim \ref{claim:main-lemma-disconnected-else}; it will be used later when bounding the size of $\mathcal{R}'_N$ and $\Q'$. The claim then states that $\mathcal{R}_N'$ and $\Q'$ satisfy requirements~\eqref{req:RprimeN} and~\eqref{req:Qprime}. Let $R \in \mathcal{R}$ be a remainder in $G_A \oplus G_B \oplus G_C$ corresponding to some solution~$Y$, such that~$\foliotstar(G_B - Y) = R_B$. We show how to find remainder $\widetilde{R}$, such that $\widetilde{R}$ is a remainder for $G_{A_2} \oplus G_B \oplus (G_{A_1} \oplus G_C)$ and furthermore $R$ is a remainder in $G_{A_1} \oplus (G_{A_2} \oplus G_B)\oplus G_C$ for some corresponding solution whose removal leaves exactly $\widetilde{R}$ in $G_{A_2} \oplus G_B$.

We define $\widetilde{R} = \foliotstar((G_{A_2} \oplus G_B)-Y)$ as the remainder left by $Y$ in $G_{A_2} \oplus G_B$. Then, we show $\widetilde{R} \in \Rtilde_0$ in  \eqref{eq:tildeR_in_Rtilde0}. We find a corresponding minimal remainder $\widetilde{R}'\subseteq \widetilde{R}$, let $Y'$ be its corresponding solution. If this remainder now has a $\Q$-minor, such a minor was added to $\Q'$ and we are done. Otherwise, we find a \Fdel $\hat{Y}$, that is equal to $Y$ in $G_{A_1}$ and $G_C$, but equal to $Y'$ in $G_{A_2} \oplus G_B$. The idea is that $\hat{Y}$ leaves the same minors in $G_{A_1} \oplus G_C$ as $Y$, but leaves a ``smaller'' remainder in $G_{A_2} \oplus G_B$.
Let $\hat{R} := \foliotstar((G_A \oplus G_B) - \hat{Y})$ be its corresponding remainder. We then show $\hat{R} \in \Rhat{\widetilde{R}'}_0$ in  \eqref{eq:hatR_in_RhattildeRprime0}. Thus there exist a minimal $\hat{R}' \in \Rhat{\widetilde{R}'}$ with $\hat{R}' \subseteq \hat{R}$, for which we added elements to $\mathcal{R}_N'$ and $\Q'$. We then show that $\hat{R}' \in \mathcal{R}_0$ in \eqref{eq:Rprime_in_R0_loop}.  We use this to conclude that either $R$ is not minimal, or it is in $\mathcal{R}_N'$, or a minor of it is in $\Q'$. We (sometimes implicitly) use the same set of lemmas as in the proof of Claim \ref{claim:main-lemma-disconnected}. We now give the formal proof.
\begin{innerproof}[Proof of Claim \ref{claim:main-lemma-disconnected-else}.]
Let $R \in \mathcal{R}$, let $Y \in \minFdelsolwith(G_A \oplus G_B \oplus G_C, \Pi_A,\Pi_B,\Pi_C,R_B)$ be its corresponding \Fdel.
Since $G_A-Y$ has no $\Pi_A$-minor and $G_A - Y = (G_{A_1} - Y) \oplus (G_{A_2} - Y)$, by Lemma \ref{lem:splitPi} there exist $(\piaone, \piatwo) \in \splitPi(\Pi_A)$ such that $G_{A_1}-Y$ has no $\piaone$-minor and $G_{A_2}-Y$ has no \piatwo-minor. By the existence of $Y$, there is an optimal solution in $G_A \oplus G_B \oplus G_C$ avoiding $S$ and breaking $\piatwo$ in $G_{A_2}$ and $\piaone$ in $G_{A_1}$. Since furthermore the condition of the if-statement in line \ref{step:emptypi} was false, we applied Step~\ref{step:tilde} of the algorithm for $(\piaone, \piatwo)$.

Let $\widetilde{R} := \foliotstar((G_{A_2} \oplus G_B) - Y)$. We show that
\begin{equation}\label{eq:tildeR_in_Rtilde0}
\widetilde{R} \in \Rtilde_0,
 \end{equation}
  by showing that it has corresponding solution $Y \in \minFdelsolwith(G_{A_2} \oplus G_B \oplus (G_C \oplus G_{A_1}), \piatwo,\allowbreak \Pi_B, \allowbreak \Pi_C \odot \Pi_{A_1}, R_B)$. By definition of $Y$, $Y \cap S = \emptyset$, $Y$ is a solution and $Y$ breaks $\piatwo$ in $G_{A_2}$ and $\Pi_B$ in $G_B$. Furthermore it breaks $\Pi_C \odot \piaone$ in $G_C \oplus G_{A_1}$ by Lemma \ref{lem:mergePi}, since it breaks $\piaone$ in $G_{A_1}$ and $\Pi_C$ in $G_C$. The size bound on $Y$ and property of $R_B$ follow from the fact that $Y$ is given to be a corresponding solution for $R$. Thereby, \eqref{eq:tildeR_in_Rtilde0} follows.

We have just proven that $\widetilde{R} \in \Rtilde_0$. Thereby, there exists $\widetilde{R}' \in \Rtilde$ such that $\widetilde{R}' \subseteq \widetilde{R}$.
Suppose $\widetilde{R}'$ has a $\Q$-minor. Then we added one such minor to $\Q'$ in line \ref{step:tilde}. Since $\widetilde{R}' \subseteq \widetilde{R}$ this implies $\widetilde{R}$ also has this minor and thus $R \in \mathcal{R}_Q$ has this minor, as required.

Suppose $\widetilde{R}'$ has no $\Q$-minor, thereby $\widetilde{R}' \in \Rtilde_N$.
Consider a \Fdel $Y' \in \minFdelsolwith(G_{A_2} \oplus G_B \oplus (G_{A_1} \oplus G_C),\piatwo, \Pi_B,\piaone \odot \Pi_C, R_B)$ corresponding to $\widetilde{R}'$.

Define $\hat{Y} := (Y' \cap V(G_{A_2} \oplus G_B )) \cup (Y\cap V(G_{A_1}\oplus G_C))$. Let $\hat{R} := \foliotstar((G_A \oplus G_B) - \hat{Y})$. Then we show that
\begin{equation}\label{eq:hatRsmallerthanR2}
\hat{R} \subseteq R.
\end{equation}
By definition,  $\folio(G_{A_1} - \hat{Y}) = \folio(G_{A_1} - Y)$ and $\foliotstar((G_{A_2} \oplus G_B) - \hat{Y}) = \widetilde{R}' \subseteq \widetilde{R}= \foliotstar((G_{A_2} \oplus G_B) - Y)$. Now $\hat{R} \subseteq R$  follows from Lemma \ref{lem:replacehalf:preservefolio}.

We applied Step \ref{step:hat} of the algorithm with remainder $\widetilde{R}'$. We now show that
\begin{equation}\label{eq:hatR_in_RhattildeRprime0}
\hat{R} \in \Rhat{\widetilde{R}'}_0,
\end{equation}
by proving  $\hat{Y} \in \minFdelsolwith(G_{A_1} \oplus (G_{A_2} \oplus G_B) \oplus G_C ,\piaone,\piatwo \odot \Pi_B,\Pi_C,\widetilde{R}')$. By definition, $\hat{Y}$ breaks $\piatwo$ in $G_{A_2}$ and $\Pi_B$ in $G_B$. Furthermore, $\hat{Y}$ breaks \piaone in $G_{A_1}$ by definition of $Y$, and $\Pi_C$ in $G_C$. It follows that $\hat{Y}$ breaks $\Pi_B \odot \piatwo$ in $G_{A_2} \oplus G_B$ as needed, and that $\hat{Y}$ is a solution as it breaks $\Pi_A \odot \Pi_B \odot \Pi_C \supseteq \extend_{+t}(\mathcal{F})$.

It is easy to see that $|\hat{Y}| \leq |Y|$, otherwise $|Y'\cap V(G_{A_2} \oplus G_B)| > |Y \cap V(G_{A_2} \oplus G_B)|$. Since $Y$ gives prohibition $\piatwo$ in $G_{A_2}$ and $\Pi_B$ in $G_B$, just like $Y'$, this would contradict the optimality of $Y'$. Thereby $|\hat{Y}| = \minFdel(G_A \oplus G_B \oplus G_C)$. Since $\hat{Y} \cap (G_{A_2} \oplus G_B) = Y' \cap (G_{A_2} \oplus G_B)$, it follows that $\widetilde{R}' = \foliotstar((G_{A_2} \oplus G_B) -\hat{Y})$ as required.
Thereby, $\hat{R} \in \Rhat{\widetilde{R}'}_0$.

Thus there exists $\hat{R}' \subseteq  \hat{R}$ with $\hat{R}' \in \Rhat{\widetilde{R}'}$. We show that
\begin{equation}\label{eq:Rprime_in_R0_loop}
\hat{R}' \in \mathcal{R}_0.
\end{equation}
Let $Y''$ be a solution corresponding to $\hat{R}'$, thus $Y'' \in \minFdelsolwith(G_{A_1} \oplus (G_{A_2} \oplus G_B) \oplus G_C, \piaone, \piatwo \odot \Pi_B, \Pi_C    ,\widetilde{R}')$. Let $\hat{Y}' := (\hat{Y} \cap V(G_{A_2} \oplus G_B)) \cup (Y'' \cap V(G_{A_1} \oplus G_C))$. We then show that $\hat{Y}' \in \minFdelsolwith(G_A \oplus G_B \oplus G_C, \Pi_A,\Pi_B,\Pi_C,R_B)$ and that $\hat{R}'$ is its corresponding remainder in $G_A \oplus G_B$. Since $\folio(G_{A_1} - \hat{Y}') = \folio(G_{A_1} - Y'')$ and $\foliotstar((G_{A_2} \oplus G_B) - \hat{Y}') = \foliotstar((G_{A_2} \oplus G_B) - \hat{Y})  = \widetilde{R}'$, it follows that $\hat{R}'  = \foliotstar((G_A \oplus G_B) - \hat{Y}')$ by applying Lemma \ref{lem:replacehalf:preservefolio}. It remains to show that $\hat{Y}'$ has the required properties. Clearly, $\hat{Y}' \cap S = \emptyset$. Furthermore, $\hat{Y'}$ breaks $\Pi_A$ in $G_A$, since $V(G_{A_2}) \cap \hat{Y}' = V(G_{A_2}) \cap \hat{Y}$ and $\hat{Y}$ breaks $\piatwo$ in $G_{A_2}$ and $V(G_{A_1})\cap \hat{Y}' = V(G_{A_1}) \cap Y''$ and $Y''$ breaks $\piaone$ in $V(G_{A_1})$, and $(\piaone, \piatwo) \in \splitPi(\Pi_A)$. Furthermore, $\hat{Y}'$ breaks $\Pi_B$ in $G_B$, since $\hat{Y}' \cap V(G_B) = \hat{Y} \cap V(G_B) = Y'\cap V(G_B)$. Furthermore, $\hat{Y}'$ breaks $\Pi_C$ in $G_C$, since $Y''$ breaks $\Pi_C$ in $G_C$. Since $\hat{Y}'$ breaks the required prohibitions, it follows that it is indeed a valid $\F$-minor free deletion of $G_A \oplus G_B \oplus G_C$. The requirement for $R_B$ follows from $\hat{Y}' \cap V(G_B) = \hat{Y} \cap V(G_B) = Y'\cap V(G_B)$ and the definition of $Y'$. It remains to show that $\hat{Y}'$ has the required size, we will use that $|\hat{Y}| = |Y''| = \minFdel(G_A \oplus G_B \oplus G_C)$. Suppose $|\hat{Y}'| > |Y''|$. But then $|\hat{Y} \cap V(G_{A_2} \oplus G_B)| > |Y'' \cap V(G_{A_2} \oplus G_B)|$. But since $Y''$ gives the same prohibition in $G_{A_2} \oplus G_B$ as $\hat{Y}$, we can replace the part of $\hat{Y}$ in $G_{A_2} \oplus G_B$ by $Y''$ to obtain a better solution, which is a contradiction. Thus $|\hat{Y}'| = \minFdel(G_A \oplus G_B \oplus G_C)$.
We can now conclude that $\hat{R}' \in \mathcal{R}_0$.

If $R \neq \hat{R}'$, it follows from $\hat{R}' \subseteq \hat{R}$ (by definition), $\hat{R} \subseteq R$ \eqref{eq:hatRsmallerthanR2}, and that $\hat{R}' \in \mathcal{R}_0$ \eqref{eq:Rprime_in_R0_loop} that $R$ is not minimal and thus $R \notin \mathcal{R}$, which is a contradiction.

Otherwise, we have $\hat{R}' = \hat{R} = R \in \Rhat{\widetilde{R}'}$, so for that choice of $\widetilde{R}' \in \Rtilde_N$ in line \ref{step:foralltilde}, we have added a set of remainders to $\mathcal{R}'$ in line \ref{step:hat} that includes $\hat{R}' = R$ (if $R$ has no $\Q$-minor, thus $R \in \mathcal{R}_N$) or we have added one of its $\Q$-minors to $\Q'$ (if $R$ does have a $\Q$-minor, meaning $R \in \mathcal{R}_Q$).
\end{innerproof}

\paragraph*{$\mathbf{G_A-S}$ connected}

Having concluded the case of the induction step that~$G_A - S$ is disconnected, while the base case covered the setting that~$G_A - S$ is empty, we are left with the case that~$G_A - S$ is connected and non-empty. Let $s$ be a vertex whose removal from~$G_A - S$ decreases its treedepth, which exists by Definition~\ref{def:treedepth}. To find all remainders in the case that $G_A-S$ is connected, we will combine the remainders corresponding to optimal solutions that contain $s$, with those corresponding to solutions that avoid $s$. The first case is easy: to find remainders corresponding to solutions that contain $s$, we will simply apply induction on the graph $G-\{s\}$. The second case is slightly more involved, where we do induction by adding $s$ to $S$ (as these solutions avoid using $s$). Some additional work needs to be done here, to find the correct remainders for the original graph, that was only $t$-boundaried, while the graph to which we apply induction is $(t+1)$-boundaried. Let $\mathcal{R}'_N$ and $\mathcal{Q}'$ be defined by the following procedure.
\begin{enumerate}
\setlength\itemsep{1mm}
\item[]
\item[]  \textsc{Find-Remainders} (connected case)
\item \textbf{if} There exists $Y \in \minFdelsolwith(G_A \oplus G_B \oplus G_C , \Pi_A,\Pi_B, \Pi_C, R_B)$ with $s \in Y$
 \item \quad Let $G_{A'}$ be $G_A  - \{s\}$
\item\label{step:minus} \quad Let $\mathcal{R}_N^-$ and $\Q^-$ be sets obtained by applying Lemma \ref{lem:subset_Q} inductively on $G_{A'} \oplus G_B \oplus G_C$
\item[] \quad with $S$ and $R_B$ unchanged
\item \textbf{if} There exists $Y \in \minFdelsolwith(G_A \oplus G_B \oplus G_C , \Pi_A,\Pi_B, \Pi_C, R_B)$ with $s \notin Y$
\item \quad Let $\Pi_{A'} := \extend_{+1}(\Pi_A)$
\item \quad Let $\Pi_{B'}:= \extend_{+1}(\Pi_B)\cup (\nonisobgraphs{t+1} \cap \multipieces_{+t+1}(\F))$ \emph{(recall Definition \ref{def:graphs})}
\item \quad Let $\Pi_{C'}:= \extend_{+1}(\Pi_C)\cup (\nonisobgraphs{t+1} \cap \multipieces_{+t+1}(\F))$
\item \quad Let $G_{A'}$ be $G_A$ with boundary label $t+1$ assigned to vertex $s$, thus $\bound_{G_{A'}}(s) = t+1$.
\item \quad Let $G' := G_{A'} \oplus G_{B}\oplus G_{C}$
\item []\emph{(Observe that $\forget(G',t) = G$ and recall that any $t$-boundaried graph can be interpreted as being $(t+1)$-boundaried)}
\item\label{step:plus} \quad Compute $\mathcal{R}'^+_N,\Q^+$ for $\Pi_{A'},\Pi_{B'},\Pi_{C'}, R_B$ on graph $G'$ with $S'$ as the boundary of $G'$
\item \label{item:connected:15} \quad Let \[\mathcal{R}^+_N := \bigcup_{R' \in \mathcal{R}'^+_N}\foliotstar(\forget(R',t))\]
\item \label{item:connected:16} Let $\mathcal{R}'_N := \mathcal{R}^+_N \cup \mathcal{R}^-_N$
\item Let $\Q' := \Q^- \cup \Q^+$
\end{enumerate}



We start by arguing  that in every step where the induction hypothesis is applied, all preconditions are  satisfied. In Step~\ref{step:minus}, this is trivially true, Step~\ref{step:plus} is more interesting. We argue:
\begin{itemize}
\item $\td(G) \geq \td(G_{A'}-S') + |S'|$. We know that $|S'| = |S| +1$ and $s$ was chosen such that $\td(G_{A'}-S') = \td(G_A - S) - 1$. The statement now follows from the precondition that $\td(G) \geq \td(G_{A}-S) + |S'|$.
\item $\Pi_{A'}\odot\Pi_{B'}\odot\Pi_{C'} \supseteq \extend_{+t+1}(\F)$ by Lemma \ref{lem:prohibitions:extendplusone} and the fact that $\Pi_{A}\odot\Pi_{B}\odot\Pi_{C} \supseteq \extend_{+t}(\F)$.
\item $\Pi_{A'} \subseteq \multipieces_{+t+1}(\F)$, this immediately follows from Lemma \ref{lem:mpcs:plusone}.
\item $\Pi_{B'} \subseteq \multipieces_{+t+1}(\F)$ and $\Pi_{C'} \subseteq \multipieces_{+t+1}(\F)$. This immediately follows from the definition together with Lemma \ref{lem:mpcs:plusone}.
\item $R_{B} \subseteq \multipieces_{+t+1}(\Q)$.  By definition of $R_B$, $R_B \subseteq \multipieces_{+t}(\Q)$, it follows from the definition of $\multipieces_{+t+1}$ that $R_B \subseteq \multipieces_{+t+1}(\Q)$. 
\end{itemize}

\begin{innerclaim}\label{claim:main-lemma-connected-case}
Suppose $G_A - S$ is connected, then $\mathcal{R}_N \subseteq \mathcal{R}'_N$ and for each $R \in \mathcal{R}_Q$ there exists $q \in \Q'$ and $r \in R$ such that $q \leql \forget(r)$.
\end{innerclaim}
\noindent Before giving the proof, we give a short outline. Given a remainder $R$, we will do a case distinction. If a corresponding solution of $R$ contains $s$, it is easy to see that $R$ is a remainder for $G-\{s\}$, thereby $R \in \mathcal{R}^-$ and the result follows. If a corresponding solution of $R$ avoids $s$, the proof is a bit more involved. The idea is to show that there exists a minimal remainder $R''$ in $\mathcal{R}'^+$, such that if $R$ has a $\Q$-minor then $R''$ also has a $\Q$-minor. Furthermore if $R''$ has no $\Q$-minor, we show that $R'' \in \mathcal{R}'^+_N$, such that in line \ref{item:connected:15} of the procedure we add $R$ to $\mathcal{R}^+_N$. In the proof we will use implicitly that $G_{B}$ and $G_{C}$ contain no $\nonisobgraphs{t+1}$-minors, which follows trivially from the definitions: boundary vertex $t+1$ is not defined in $G_{B}$ and $G_{C}$ as these are $t$-boundaried graphs, so any graph in which the $t+1$'th boundary vertex is defined cannot be obtained from $G_{B}$ or $G_{C}$ by minor operations.
\begin{innerproof}[Proof of Claim \ref{claim:main-lemma-connected-case}]
Let $R \in \mathcal{R}$ be a remainder with corresponding solution $Y \in \minFdelsolwith(G_A \oplus G_B \oplus G_C, \Pi_A,\Pi_B,\Pi_C,R_B)$. Suppose $s \in Y$, it is easy to see that thereby, we applied line \ref{step:minus} of the procedure.  Let $Y' := Y \setminus \{s\}$. It is easy to see that $R$ is a remainder for $Y'$ in graph $G - \{s\}$. We show $R \in \mathcal{R}^-_0$, by showing that $Y'$ has the required properties. The only property it does not inherit from $Y$ is that it has optimal size. However, if there exists a solution in $G_A \oplus G_B \oplus G_C - \{s\}$ that is smaller than $|Y|-1$, we could use this solution and extend it with $s$ to obtain a solution in the entire graph that is strictly smaller than $Y$, which is a contradiction. Thus, $Y' \in \minFdelsolwith((G_A-\{s\})\oplus G_B \oplus G_C, \Pi_A,\Pi_B,\Pi_C,R_B)$. It is easy to see that if $R$ is not minimal for $\mathcal{R}^-_0$, then it is not minimal for $\mathcal{R}_0$ and then $R \notin \mathcal{R}$, which is a contradiction.
Thereby if $R \in \mathcal{R}_N$, it follows $R \in \mathcal{R}_N^-$ and if $R\in \mathcal{R}_{\Q}$, then there exists $q \in \Q^- \subseteq \Q'$ such that $q \leql \forget(r)$ for some $r \in R$.

Suppose $s \notin Y$, it follows that we applied Step $\ref{step:plus}$, as $Y$ is an example of such a solution avoiding~$s$. Let $G' := G_{A'}\oplus G_{B}\oplus G_{C}$. We show
\begin{equation}\label{eq:Y_in_minFdelsolwith}
Y \in \minFdelsolwith(G_{A'} \oplus G_{B} \oplus G_{C}, \Pi_{A'},\Pi_{B'},\Pi_{C'},R_{B}).
\end{equation}
 It is easy to see that $Y$ is a solution of optimal size, as $\forget(G',t) = G$ by definition and graphs in $\mathcal{F}$ are not boundaried. Furthermore, since $s \notin Y$ and $Y \cap S = \emptyset$ it follows that $Y \cap S' = \emptyset$. Then suppose $G_{A'} - Y$ has a $\Pi_{A'}$-minor. Since $G_{A'}$ was obtained from $G_A$ by assigning a boundary label to an existing vertex it is easy to see that this implies $G_{A}$ has a $\Pi_A$-minor, which is a contradiction. Furthermore, suppose $G_{B}-Y$ has a  $\pi'$-minor for $\pi'\in\Pi_{B'}$. Since $G_{B}$ has no \nonisobgraphs{t+1}-minors, $\pi' \in \extend_{+1}(\Pi_B)$. However since $G_{B}$ has no \nonisobgraphs{t+1}-minors, it follows from the definition of $\extend$ that $\pi' \in \Pi_B$, which contradicts the choice of $Y$. The same reasoning implies that $G_{C}$ has no $\Pi_{C'}$-minors. It remains to show the property of $R_{B}$:
 \begin{equation}\label{eq:RBprime-folioplusonestar}
 R_{B} = \foliotplusonestar(G_{B} - Y).
  \end{equation}
  Let $r \in R_{B}$, hereby $r \in \folio(G_B - Y)$ and $r \in \multipieces_{+t}(\Q)$. It follows from Lemma \ref{item:remove_boundary_vertex_reduce_mpcs} that $r \in \multipieces_{+t+1}(\Q)$. Thereby, $r \in \foliotplusonestar(G_{B}-Y)$.

   Suppose $r \in \foliotplusonestar(G_{B}-Y)$, it follows that $r \in \folio(G_{B} - Y)$ and $r \in \multipieces_{+t+1}(\Q)$. By Lemma \ref{item:remove_boundary_vertex_reduce_mpcs}, it follows that $r \in \multipieces_{+t}(\Q)$. Since we know that $R_B = \foliotstar(G_B-Y)$, it follows that $r \in R_B$.
   This shows \ref{eq:RBprime-folioplusonestar}, which was the last step towards proving \eqref{eq:Y_in_minFdelsolwith}. 

Let $R' :=  \foliotplusonestar((G_{A'} \oplus G_{B})-Y)$. It follows from \eqref{eq:Y_in_minFdelsolwith} that $R' \in \mathcal{R}'^+_0$.
Thereby, there exists $R'' \subseteq R'$ such that $R'' \in \mathcal{R}'^+$. Let $Y' \in \minFdelsolwith(G_{A'}\oplus G_{B} \oplus G_{C}, \Pi_{A'},\Pi_{B'},\Pi_{C'},R_{B})$ be its corresponding solution, hence $R'' = \foliotplusonestar((G_{A'}\oplus G_{B}) - Y')$. Let $\hat{R} := \foliotstar((G_A \oplus G_B) - Y')$.
We show that
\begin{equation}\label{eq:connected:hat_R_smaller_R}
\hat{R} \subseteq R.
\end{equation}
Let $r \in \hat{R}$, then there exists $r' \in \extend_{+1}(r)$ such that $r' \in \folio(G_{A'} \oplus G_{B} - Y')$ by Lemma~\ref{lem:extend:minor}.
We know that $r \in \multipieces_{+t}(\Q)$ by definition, it follows from Lemma~\ref{lem:mpcs:plusone} that $r'\in\multipieces_{+t+1}(\Q)$. Thereby $r' \in \foliotplusonestar(G_{A'} \oplus G_{B} - Y') = R''$. Since $R'' \subseteq R' = \foliotplusonestar(G_{A'}\oplus G_{B} -Y)$, it follows that $r' \leqlb G_{A'} \oplus G_{B}-Y$. But since $r \leqlb \forget(r',t)$, it follows that $r \leqlb \forget(G_{A'}\oplus G_{B} - Y,t) = G_A\oplus G_B - Y$. Since $r \in \multipieces_{+t}(\Q)$  by definition, it follows that $r \in \foliotstar((G_A \oplus G_B) - Y)=R$, which establishes \eqref{eq:connected:hat_R_smaller_R}. We now prove
\begin{equation}\label{eq:connected:Rhat-in-R0}
\hat{R} \in \mathcal{R}_0,
\end{equation}
by showing that $Y' \in \minFdelsolwith(G_{A}\oplus G_B\oplus G_C,\Pi_A,\Pi_B,\Pi_C,R_B)$.  It is easy to see that $Y'\cap S = \emptyset$, and that $Y'$ has optimal size, since $\forget(G') = \forget(G)$. Furthermore since $Y'$ breaks $\Pi_{A'}$ in $G_{A'}$, by Lemma \ref{lem:extend:minor} it breaks $\Pi_A$ in $G_A$. Furthermore it breaks $\Pi_B$ and $\Pi_C$ in $G_B$ and $G_C$ respectively, by Lemma \ref{lem:extend:minor} and because it breaks $\Pi_{B'}$ and $\Pi_{C'}$ in $G_{B}$ and $G_{C}$. Furthermore, we show $R_B = \foliotstar(G_B-Y')$.  
By definition of $Y'$, we know $R_B = \foliotplusonestar(G_B-Y')$. Since it follows from Lemma \ref{item:remove_boundary_vertex_reduce_mpcs}, that for any $t$-boundaried graph $r$ it holds that $r \in \multipieces_{+t+1}(\Q)$ if and only if $r \in \multipieces_{+t}(\Q)$, it is easy to see that $R_B = \foliotstar(G_B-Y')$.
It follows that  $\hat{R} \in \mathcal{R}_0$, which establishes \eqref{eq:connected:Rhat-in-R0}.

If $\hat{R} \neq R$,the combination of \eqref{eq:connected:hat_R_smaller_R} and \eqref{eq:connected:Rhat-in-R0} shows that $R$ was not minimal, which is a contradiction. Otherwise, $R=\hat{R}$ and thereby $R$ has corresponding solution $Y'$ in $G$.

To conclude the proof, we consider two cases depending on whether $R''$ contains an $\extend_{+(t+1)}(\Q)$-minor or not. Suppose first that it does not, meaning that $R'' \in \mathcal{R}'^+_N$. This implies that $R$ has no $\extend_{+t}(\Q)$ minor, by Lemma \ref{lem:extend:minor}, thus we conclude $R \in \mathcal{R}_N$. We show that
\begin{equation}\label{eq:connected:R-is-foliostar}
R = \foliotstar(\forget(R'',t)),
\end{equation}
 which will establish that $R$ was added to $\mathcal{R}_N^+$ in Step \ref{item:connected:15} of the procedure.
Let $r \in R$, we show that $r \in \foliotstar(\forget(R'',t))$. Since $r \leqlb (G_A \oplus G_B) - Y'$, there exists $r' \in \extend_{+1}(r)$ such that $r' \leqlb (G_{A'} \oplus G_{B}) - Y$ by Lemma~\ref{lem:extend:minor}.
Since $r \in \multipieces_{+t}(\Q)$, it follows that $r' \in \multipieces_{+t+1}(\Q)$ by Lemma~\ref{lem:mpcs:plusone}. Thereby, $r' \in R''$. It is easy to see that $r \in \folio(\forget(R'',t))$ since $r \leqlb \forget(r',t)$. Thereby, $r \in \foliotstar(\forget(R'',t))$.

To prove the reverse direction of \eqref{eq:connected:R-is-foliostar}, let $r \in \foliotstar(\forget(R'',t))$. We show that $r \in R$. Since $r \in \folio(\forget(R'',t))$, it follows that $r \in \folio((G_A \oplus G_B) - Y')$, as needed. By definition, $r \in  \multipieces_{+t}(\Q)$, thereby it follows that $r \in R$, which establishes~\eqref{eq:connected:R-is-foliostar}.
 This concludes the proof for the case where $R'' \in \mathcal{R}'^+_N$, since $R$ has been added to $\mathcal{R}_N^+$ in line \ref{item:connected:15} of the procedure by~\eqref{eq:connected:R-is-foliostar}, and thereby it was added to $\mathcal{R}'_N$ in line \ref{item:connected:16}.

Suppose $R''$ does have a \Q-minor, implying $R'' \in \mathcal{R}'^+_{\Q}$. Then some \Q-minor $q$ of $R''$ was added to $\Q^+$. Then $q\leql \forget(\folio(G_{A'}\oplus G_{B} - Y'))=\forget(\folio(G_{A}\oplus G_{B} - Y'))$. By Lemma \ref{lem:extend:minor}, there is a $t$-boundaried graph $q'\in\extend_{+t}(\Q)$ such that $q' \leqlb \folio(G_{A}\oplus G_{B}-Y'))$. It is easy to see that $q'\in\multipieces_{+t}(\Q)$, thereby $q' \in \foliotstar(G_A \oplus G_B-Y')$, implying that $R$ has also has this $q$-minor, concluding the proof.
\end{innerproof}

\paragraph*{Size bound} In order to prove a bound on the size of $\mathcal{R}_N'$ and $\Q'$, we first bound the size of all used parameters.
\begin{innerclaim}\label{claim:parameter_bounds}
If none of the base cases apply, then the following bounds are satisfied:
\begin{itemize}
\item $0 \leq \mu(G_A,\Pi_A,S) \leq |S|$,
\item $0 \leq \nu(\Pi_A) \leq \numberof(0,t,\|\F\|+t,0)$,
\item $0 \leq \xi(R_B) \leq \numberof(t\cdot \min_{H \in \F}|V(H)|,t,t+\max_{H \in \Q}|V(H)|,\min_{H\in\F}|V(H)|)$, and
\item $0 \leq |S| \leq \td(G)$.
\end{itemize}
\end{innerclaim}
\begin{innerproof}
We show that the bounds always hold, or that we are in one of the base cases.
\begin{itemize}
\item For any \Fdel $Y$ of $G_A$ that breaks $\Pi_A$ and does not remove any vertex of $S$, for a connected component $C$ of $G-S$, $|C \cap Y| \geq \minFdel(C)$. Thereby, $\mu(G_A,\Pi_A,S) \geq 0$. If $\mu(G_A,\Pi_A,S) > |S|$, base case \ref{base-case:mu-large} applies. Thereby, $0 \leq \mu(G_A,\Pi_A,S) \leq |S|$.
\item Clearly, the size of a set is always as least zero. Furthermore, $\Pi_A \subseteq \multipieces_{+t}(\F)$, thus $\nu(\Pi_A) \leq |\multipieces_{+t}(\F)| \leq N(0,t,\|\F\|+t,0)$, as these are $t$-boundaried graphs of size at most $\|\F\|+t$.
\item By definition, $\xi(R_B) \leq \numberof(t\cdot \min_{H \in \F}|V(H)|,t,t+\max_{H \in \Q}|V(H)|,\min_{H\in\F}|V(H)|)$. If $\xi(R_B) < 0$, we are in base case~\ref{base-case:Q-in-RB} or~\ref{base-case:xi-large}.
\item Trivially, $|S| \geq 0$. It follows from the precondition to the lemma that $\td(G) \geq \td(G_A-S) + |S|$. Thereby, $|S| \leq \td(G)$. \qedhere
\end{itemize}
\end{innerproof}

To show the size bound in the case that $G_A-S$ is not connected, we will use the following claim. It will be used to show that the sizes of the sets obtained in lines \ref{step:tilde} and \ref{step:hat} of the procedure for the disconnected case, are upper bounded by $f$ and $g$ with  lexicographically smaller parameter values. Recall the definition of \Rtilde, which was used for \textsc{Find-Remainders} (disconnected case): let
$\Rtilde$ denote the set $\mathcal{R}$ obtained by applying the lemma inductively to the graph $G_{A_2} \oplus G_B \oplus (G_C \oplus G_{A_1})$, with $G_{A'} := G_{A_2}$ and the new $G_{C'}$ as $G_C \oplus G_{A_1}$ and with prohibitions $\piatwo$ and $\Pi_B$ and $\piaone \odot \Pi_C$ and remainder $R_B$.

\begin{innerclaim}\label{claim:at_least_one_parameter_decreases}
Let $(\piaone,\piatwo) \in \splitPi(\Pi_A)$ such that there exists $Y \in \minFdelsolwith(G_A \oplus G_B \oplus G_C,\Pi_A,\Pi_B,\Pi_C,R_B)$ such that $Y$ breaks $\piaone$ in $G_{A_1}$ and $\piatwo$ in $G_{A_2}$. Let $R \in \Rtilde$ and
suppose there does not exist $\piatwo'$ such that $(\Pi_A, \piatwo') \in \splitPi(\Pi_A)$, $\minFdel(G_{A_2},\piatwo',S) = \minFdel(G_{A_2})$, and $R_B \in \widetilde{\mathcal{R}}^{\piatwo',\Pi_A}$.
 Then $\nu(\piaone) < \nu(\Pi_A)$, or $\xi(R) < \xi(R_B)$, or $\mu(G_{A_1},\piaone,S) < \mu(G_A,\Pi_A,S)$.
\end{innerclaim}
\begin{innerproof}
Suppose for contradiction that $\nu(\piaone) = \nu(\Pi_A)$, $\xi(R) = \xi(R_B)$, and $\mu(G_{A_1},\piaone,S) = \mu(G_A,\Pi_A,S)$.
Since $\nu(\piaone) = \nu(\Pi_A)$ and $\piaone$ is from $\splitPi(\Pi_A)$ it follows  that $\piaone = \Pi_A$.  From $\mu(G_A,\Pi_A,S) = \mu(G_{A_1},\piaone,S)$ and~$\Pi_{A_1} = \Pi_A$, it follows that
\begin{align}
&\label{eq:equality-first}
\minFdel(G_A,\Pi_A,S) - \shrink\sum_{C \in \setcomponents(G_A-S)}\shrink \minFdel(C)
= \minFdel(G_{A_1},\Pi_A,S) - \shrink\sum_{C \in \setcomponents(G_{A_1}-S)}\shrink \minFdel(C).
\intertext{Since by assumption, $R \in \Rtilde$, it follows from the definition of \Rtilde that there is a solution of size~$\minFdel(G_A, \Pi_A, S)$ in $G_A$ breaking $\Pi_{A_1}=\Pi_A$ in $G_{A_1}$ and $\Pi_{A_2}$ in $G_{A_2}$. Thus, }
&\label{eq:equality:geq}
\minFdel(G_A,\Pi_A,S) \geq \minFdel(G_{A_1},\Pi_{A}, S) + \minFdel(G_{A_2},\Pi_{A_2}, S).
\intertext{Since $(\Pi_{A_1},\Pi_{A_2}) \in \splitPi(\Pi_A)$, a solution breaking $\Pi_{A_1}=\Pi_A$ in $G_{A_1}$ can always be combined with a solution breaking $\Pi_{A_2}$ in $G_{A_2}$ to form a solution breaking $\Pi_A$ in $G_A$. Thus,}
&\label{eq:equality:leq}
\minFdel(G_A,\Pi_A,S) \leq \minFdel(G_{A_1},\Pi_{A}, S) + \minFdel(G_{A_2},\Pi_{A_2}, S).
\intertext{From \eqref{eq:equality:geq} and \eqref{eq:equality:leq}, we conclude that}
&\minFdel(G_A,\Pi_A,S) = \minFdel(G_{A_1},\Pi_{A}, S) + \minFdel(G_{A_2},\Pi_{A_2}, S).\nonumber
\intertext{Observe that $G_{A_2}-S \in \setcomponents(G_A - S)$ and that all other connected components of $G_A-S$ are in $\setcomponents(G_{A_1}-S)$. Thereby, we get that}
&\minFdel(G_A,\Pi_A,S) - \shrink \sum_{C \in \setcomponents(G_A-S)}\shrink  \minFdel(C) \nonumber\\
&=\label{eq:equality-second} \minFdel(G_{A_1},\Pi_A,S) + \minFdel(G_{A_2},\piatwo,S) - \minFdel(G_{A_2}-S)- \shrink\sum_{C\in \setcomponents(G_{A_1}-S)}\shrink \minFdel(C).
\intertext{It immediately follows from \eqref{eq:equality-first} and \eqref{eq:equality-second} that}
&\minFdel(G_{A_1},\Pi_A,S) + \minFdel(G_{A_2},\piatwo,S) - \minFdel(G_{A_2}-S)- \shrink\sum_{C\in \setcomponents(G_{A_1}-S)}\shrink \minFdel(C)\nonumber\\
&= \minFdel(G_{A_1},\Pi_A,S) -\shrink \sum_{C \in \setcomponents(G_{A_1}-S)} \shrink\minFdel(C)\nonumber
\intertext{and thereby}
&\label{eq:equality-third}\minFdel(G_{A_2}-S) = \minFdel(G_{A_2},\piatwo,S).
\end{align}
Since $\minFdel(G_{A_2},\piatwo,S) \geq \minFdel(G_{A_2})$ and $\minFdel(G_{A_2}) \geq \minFdel(G_{A_2}-S)$ it follows from \eqref{eq:equality-third} that $\minFdel(G_{A_2}) = \minFdel(G_{A_2},\piatwo,S)$.
By $\xi(R) = \xi(R_B)$ and the way $R$ is chosen, it follows that $R = R_B$, as for every remainder in $R' \in \Rtilde$ we know that $R_B \subseteq R'$. However, this implies that $R_B \in \widetilde{\mathcal{R}}^{\piatwo,\Pi_A}$. But then we have that $(\Pi_A,\piatwo) \in \splitPi(\Pi_A)$, $R_B \in \widetilde{\mathcal{R}}^{\piatwo,\Pi_A}$, and $\minFdel(G_{A_2},\piatwo,S) = \minFdel(G_{A_2})$, which contradicts the starting assumptions of Claim \ref{claim:at_least_one_parameter_decreases}. Thereby, the claim follows.
\end{innerproof}
%
Using Claim \ref{claim:at_least_one_parameter_decreases}, we can now argue that there exist functions $f$ and $g$ such that the required size bounds hold for $\mathcal{R}_N'$ and $\mathcal{Q}'$.
\begin{innerclaim}\label{claim:QandRsmall}
There exist functions $f$ and $g$ such that:
\begin{enumerate}
\item $|\mathcal{R}'_N| \leq f(\td(G_A-S), \iscon(G_A - S), \mu(G_A,\Pi_A,S), \nu(\Pi_A),\xi(R_B), \max_{H \in \F}|V(H)|, |S|)$ and
\item $|\Q'| \leq g(\td(G_A-S), \iscon(G_A - S), \mu(G_A,\Pi_A,S), \nu(\Pi_A),\xi(R_B), \max_{H \in \F}|V(H)|, |S|).$
\end{enumerate}
\end{innerclaim}
\begin{innerproof}

We first order the parameters of $f$ and $g$ by importance. We then show that the sizes of $\mathcal{R}'_N$ and $\Q'$ are bounded by an expression in $f$,~$g$, and a number of constants, where $f$ and $g$ are called with a set of parameters that is lexicographically smaller. We use this to conclude that by induction,   there exist functions $f$ and $g$ depending only on the specified parameters, that satisfy this expression. In this way we bound $|\mathcal{R}_N'|$ and $|\Q'|$.  The ordering on the parameters is the same as in the lemma statement, namely $(\td(G_A-S), \iscon(G_A - S), \mu(G_A,\Pi_A,S), \nu(\Pi_A),\xi(R_B), \max_{H \in \F}|V(H)|, |S|)$ (important to least important). For example, if we compute the size of $\mathcal{R}_N$ on a graph where $G_A - S$ has smaller treedepth, it should  follow that $f$ gives a smaller bound on $|\mathcal{R}_N|$.

If we are in one of the base cases, $|\mathcal{R}'_N|$ and $|\Q'|$ are constant and the result follows from suitably choosing $f$ and $g$.
Otherwise, we do the following case distinction to prove the result.

(\textbf{$\mathbf{G_A-S}$ is not connected}) In this case we do a further case distinction.
\begin{itemize}
\item If  there exists $\Pi_{A_2}$ such that $(\Pi_A,\Pi_{A_2}) \in \splitPi(\Pi_A)$, $\minFdel(G_{A_2}, \Pi_{A_2},S) = \minFdel(G_{A_2})$, and $R_B \in \widetilde{\mathcal{R}}^{\piatwo,\Pi_A}$ , then $|\mathcal{R}_N'| = |\hat{\mathcal{R}}_N^{\Pi_A,\piatwo,R_B}|$, which is bounded by $f$ with the same parameters, by the induction hypothesis. Similarly, $|\Q'| = |\hat{\Q}^{\Pi_A,\piatwo,R_B}|$ which is bounded by $g$ by the induction hypothesis. Thereby, the result follows.
\item Otherwise, it is easy to see from the algorithm that
\begin{align*}
|\mathcal{R}_N'| &\leq \sum_{\substack{\text{chosen}\\\piatwo}}\sum_{R \in \Rtilde_N}|\Rhat{R}|.
\end{align*}
The total number of different prohibitions $\piatwo$ that can be chosen is bounded by $2^{N(0,t,t+\|\F\|,0)}$ (all subsets of the graphs in $\multipieces_{+t}(\F)$), which is constant for our purposes as $t$ is bounded by the treedepth of $G$ and \F is fixed. Furthermore, $|\Rtilde|$ is bounded by $f$ of a graph where $G_A-S$ is connected and $\td(G_A - S)$ does not change, which has smaller parameters in the lexicographical ordering. $|\Rhat{R}|$ is bounded by $f$ where one of the parameters $\mu$,$\nu$ or $\xi$ has decreased, as given by Claim \ref{claim:at_least_one_parameter_decreases} and it is easy to verify that $\td(G_A-S)$ remains constant and the connectedness of the graph can only improve from disconnected to connected. Similarly,
\begin{align*}
\Q' &\leq \sum_{\substack{\text{chosen}\\\piatwo}}\left(|\Qtilde|+\sum_{R \in \Rtilde_N} |\Qhat{R}| \right)\\
\end{align*}
The number of chosen prohibitions $\piatwo$ is again bounded by $2^{N(0,t,t+\|\F\|,0)}$. Then, $|\Qtilde|$ is bounded by $g$ of a graph where $G_A-S$ is connected that has same treedepth. Furthermore, $|\Rtilde|$ is bounded by $f$ of a connected graph with the same treedepth and for each of these (``few'') $R \in \Rtilde$ we compute $\Qhat{R}$. The size of this set is bounded by $g$ where one of the parameters $\mu$,$\nu$ or $\xi$ has decreased, as given by Claim \ref{claim:at_least_one_parameter_decreases} and it is easy to verify that $\td(G_A-S)$ remains constant and the connectedness of the graph can only improve.

\end{itemize}

(\textbf{$\mathbf{G_A - S}$ is connected})  In this case $|\Q'| = |\Q^-| + |\Q^+|$ and $|\mathcal{R}_N'| = |\mathcal{R}_N^-|+|\mathcal{R}_N^+| =  |\mathcal{R}_N^-|+|\mathcal{R}'^+_N|$ and all these sets are computed on  a graph where $G_A - S$ has strictly smaller treedepth. Thereby, the result follows.

(\textbf{Concluding the proof})
The combination of the above inequalities gives expressions for $g$ and $f$, where the right-hand side only contains constants and some multiplication of $g$ and $f$ with (lexicographically) smaller input parameters. One can recursively find suitable $f$ and $g$ to bound the size of $\mathcal{R}_N \leq \mathcal{R}_N'$ and $\Q'$ respectively.
\end{innerproof}
By Claims \ref{claim:main-lemma-disconnected}, \ref{claim:main-lemma-disconnected-else}, and \ref{claim:main-lemma-connected-case}, it follows that $\Q'$ has the required properties and $\mathcal{R}_N \subseteq \mathcal{R}_N'$. The required bounds on $|\mathcal{R}_N|$ and $|\Q'|$ now follow from Claim \ref{claim:QandRsmall}. This concludes the proof of Lemma \ref{lem:subset_Q}.
\end{proof}
Using Lemma \ref{lem:subset_Q}, it is now straightforward to prove Lemma \ref{lem:main}.
\begin{lemma:main:statement}[{Main lemma}] \label{lem:main:exact}
%
\mainlemma
\end{lemma:main:statement}
\begin{proof}
Let \F, \Q, and graph $C \in \lgraphs$ be given. Apply Lemma \ref{lem:subset_Q} with \Q, \F, $G_A := C$, $G_B$ and $G_C$ empty, $\Pi_A  := \F$, $\Pi_B = \Pi_C = \emptyset$ and $R_B = \emptyset$. The lemma gives set $\Q^*\subseteq \Q$, that satisfies the required size bound, since $\td(G_A - S) = \td(C)$, $C$ is connected, $\mu(G_A,\Pi,S) = 0$, $\nu(\Pi_A)$ and $\xi(R_B)$ are bounded by a constant depending only on~$\F$ and~$\td(C)$ (per Claim \ref{claim:parameter_bounds}) and~$|S| = 0$.

Furthermore, for each $R \in \mathcal{R}_{\Q}$ there exists $q \in \Q^*, r\in R$ with $q \leqlb r$. It remains to show that there exists no optimal \F-minor free deletion in $C$ that breaks $\Q^*$. Let $Y \in \minFdelsol(G)$ by any optimal solution, we show that $C-Y$ has some $\Q^*$-minor. 

Let $R \in \mathcal{R}_0$ be the remainder corresponding to $Y$ in $C$.  Let $R' \in \mathcal{R}$ be a minimal remainder with $R' \subseteq R$ and with corresponding solution $Y'$. Since no optimal solution breaks $\Q$ in $C$ by assumption, $C-Y'$ has a $q$-minor for some $q \in \Q$. Since $\extend_{+0}(\Q) = \Q$ by definition, and $\Q \subseteq \multipieces_{+0}(\Q)$ since \Q has an empty boundary, it follows that $q \in R'$. Hereby, $R' \in \mathcal{R}_{\Q}$. Thus, there exist $q^* \in \Q^*$ and $r \in R'$ such that $q^*\leql r$. It is easy to see from the definition of $R'$ and $R'\subseteq R$ that thereby $q^* \leql C-Y$, as required.
\end{proof} 

\section{Kernelization lower bound} \label{sec:lb}
In this section we prove for every fixed~$\eta$ that is sufficiently large, that the degree of any polynomial bounding the kernel size for \textsc{Vertex Cover} parameterized by a treedepth-$\eta$ modulator must be exponential in~$\eta$. This follows by analyzing the treedepth of the gadgets used in existing lower bound constructions for \textsc{Vertex Cover} that targeted other structural parameterizations (cf.~\cite[Thm. 5.3]{Jansen13} and~\cite[Thm. 2]{FominS16}).

For an integer~$t$, define a \emph{path of~$t$ triangles} as the graph~$P^\Delta_t$ obtained from~$t$ vertex-disjoint triangles~$\{a_1, b_1, c_1\}, \ldots, \{a_t, b_t, c_t\}$ by adding the edges~$\{b_i, a_{i+1}\}$ for~$i \in [t-1]$. A \emph{clause gadget of size~$t$}~\cite[Def. 5.10]{Jansen13} is the graph~$C_t$ obtained from a path of~$t$ triangles by  adding an extra vertex~$z$ adjacent to~$a_1$ and~$b_t$.

\begin{proposition} \label{prop:treedepth:path}
A path of~$2^r$ triangles has treedepth at most~$r+3$.
\end{proposition}
\begin{proof}
Proof by induction on~$r$. For~$r=1$, it is easy to verify that the path of two triangles has treedepth four. For~$r \geq 2$ we consider a path~$P^\Delta_t$ of~$t = 2^r$ triangles. Observe that by deleting vertex~$a_{1+t/2}$ in~$P^\Delta_t$, it splits into two connected components that are both subgraphs of~$P^\Delta_{t/2}$. By definition, the treedepth of~$P^\Delta_t$ is at most one plus the maximum treedepth of the two components. Since treedepth does not increase when taking subgraphs, by induction both components have treedepth at most~$(r-1) + 3$. Hence~$P^\Delta_{2^r}$ has treedepth at most~$r+3$.
\end{proof}

Proposition~\ref{prop:treedepth:path} implies that a clause gadget of size~$2^r$ has treedepth at most~$r+4$; the extra vertex that is added to~$P^\Delta_{2^r}$ to form~$C_{2^r}$, increases the treedepth by at most one. Using this gadget, we can prove Theorem~\ref{thm:lowerbound}.

\begin{thm:lowerbound:statement}
\lowerbound
\end{thm:lowerbound:statement}
\begin{proof}
We use the following result by Dell and van Melkebeek~\cite{DellM14}: for each constant~$q \geq 3$,~$\varepsilon > 0$, and decision problem~$L$, there is no polynomial-time algorithm that transforms an instance of \textsc{$q$-cnf-sat} with~$n$ variables into an equivalent instance of~$L$ of bitsize~$\Oh(n^{q-\varepsilon})$, unless \containment.

In earlier work~\cite[Thm. 5.3]{Jansen13}, the first author gave a polynomial-time construction that achieves the following. Starting from a \textsc{$q$-cnf-sat} instance~$\phi$ with~$n$ variables clauses of size~$q$, it constructs a graph~$G$, integer~$k$, and vertex set~$X \subseteq V(G)$ such that the following holds:
\begin{enumerate}
	\item $|X| = 2n$,
	\item $G$ has an independent set of size~$k$ if and only if~$\phi$ is satisfiable, and
	\item each connected component of~$G-X$ is a clause gadget of size~$q$.
\end{enumerate}
Since the dual of an independent set is a vertex cover, it follows that~$G$ has a vertex cover of size~$|V(G)| - k$ if and only if~$\phi$ is satisfiable.

Using this construction we complete the proof. Suppose that for some~$\eta \geq 6$ and~$\varepsilon > 0$, there is a kernel of bitsize~$\Oh(|X|^{2^{\eta - 4}-\varepsilon})$ for \textsc{Vertex Cover} parameterized by a treedepth-$\eta$ modulator. Let~$q := 2^{\eta - 4} \geq 2^2 \geq 4$. Then we can compress an $n$-variable instance~$\phi$ of~\textsc{$q$-cnf-sat} into an instance of~\textsc{Vertex Cover} as follows: 
\begin{enumerate}
	\item Use the construction above to transform~$\phi$ into an equivalent instance~$(G,|X|, k' := |V(G)| - k)$ of \textsc{Vertex Cover} in which each connected component of~$G-X$ is a clause gadget of size~$q$.
	\item It follows that~$X$ is a modulator of size~$2n$ to a graph of treedepth~$\eta$. Apply the hypothetical kernel for \textsc{Vertex Cover} parameterized by treedepth-$\eta$ modulator. By assumption, it outputs an equivalent instance of bitsize~$\Oh(|X|^{2^{\eta - 4}-\varepsilon}) = \Oh((2n)^{q - \varepsilon}) = \Oh(n^{q - \varepsilon})$.
\end{enumerate}
By the cited result of Dell and van Melkebeek~\cite{DellM14}, the existence of such a compression for any constant~$q \geq 3$ implies \containment.
\end{proof}

Since a graph with~$n$ vertices can be represented in~$n^2$ bits, the lower bound on bitsize in Theorem~\ref{thm:lowerbound} implies that \textsc{Vertex Cover} parameterized by treedepth-$\eta$ modulator does not have a kernel with~$\Oh(|X|^{2^{\eta-5}-\varepsilon})$ vertices, unless \containment.

\end{document}